\documentclass{article}
\usepackage{graphicx} 

\RequirePackage{amsmath, amsthm, amssymb, mathtools,hyperref}
\RequirePackage{comment}
\RequirePackage{graphicx}
\usepackage{caption, subcaption}

\RequirePackage{fancyvrb}
\RequirePackage[scale=0.95]{sourcecodepro}
\RequirePackage{mathpazo}

\RequirePackage{dsfont}
\RequirePackage{parskip}
\RequirePackage{tikz}
\usepackage{tikz-cd}

    \usetikzlibrary{arrows.meta}
    \usetikzlibrary{positioning}
    \usetikzlibrary{decorations.markings, calc, arrows} 
    \usetikzlibrary{plotmarks, patterns}
    \usetikzlibrary{decorations.pathreplacing}
\usepackage{pgfplots}
\pgfplotsset{compat=1.18}
\usepgfplotslibrary{groupplots}

\RequirePackage[backend=biber,backref=true,style=numeric,giveninits=true,natbib=true,maxbibnames=12]{biblatex}
\RequirePackage{microtype}

\usepackage{algorithm}
\usepackage{algpseudocode}

\theoremstyle{plain}
\newtheorem{prototheorem}{Theorem}
\newtheorem{theorem}[prototheorem]{Theorem}

\theoremstyle{plain}
\newtheorem{prototheorem2}{theorem}
\newtheorem{corollary}[prototheorem2]{Corollary}

\theoremstyle{plain}
\newtheorem{prototheorem3}{theorem}
\newtheorem{lemma}[prototheorem3]{Lemma}

\theoremstyle{definition}
\newtheorem{definition}{Definition}

\theoremstyle{remark}
\newtheorem{remark}{Remark}

\usepackage{accents}
\newlength{\dhatheight}

\theoremstyle{remark}

\newcommand{\pos}[2]{#1^{(#2)}}

\newcommand{\R}{\mathbb{R}}

\newcommand{\leftmost}{\mathrm{left}}
\newcommand{\rightmost}{\mathrm{right}}

\usepackage{stfloats}

\usepackage{newfloat}
\DeclareFloatingEnvironment[
    fileext=los,
    listname={List of Algorithms},
    name=Listing,
    placement=tbhp,
    within={},
]{mylisting}

\newcommand{\calO}{\mathcal{O}}
\newcommand{\calW}{\mathcal{W}}

\newcommand{\shift}{\mathcal{S}}

\newcommand{\Indc}[1]{{\mathbf{1}\left\{{#1}\right\}}}

\newcommand{\N}{\mathcal{N}}

\newcommand{\tran}{^\intercal}

\newcommand{\old}{{\mathrm{old}}}

\newcommand{\Unif}{\mathrm{Unif}}

\newcommand{\bbN}{\mathbb{N}}

\newcommand{\cat}{{\mathrm{categorical}}}
\newcommand{\ext}{\mathrm{ext}}
\newcommand{\orbit}{\mathrm{orbit}}

\newcommand{\micro}{\mathrm{micro}}

\newcommand{\ind}{\mathrm{index}}
\newcommand{\mn}{\mathrm{MN}}
\newcommand{\bp}{\mathrm{BP}}
\newcommand{\joint}{{\mathrm{joint}}}

\newcommand{\bbZ}{{\mathbb{Z}}}
\newcommand{\calZ}{{\mathbb{A}}}
\newcommand{\bfell}{{\boldsymbol{\ell}}}
\newcommand{\numberthis}{\addtocounter{equation}{1}\tag{\theequation}}

\usetikzlibrary{shapes.geometric}

\title{The Within-Orbit Adaptive Leapfrog No-U-Turn Sampler}

\author{Nawaf Bou-Rabee\thanks{Department of Mathematical Sciences, Rutgers University, \href{mailto:nawaf.bourabee@rutgers.edu}{\texttt{nawaf.bourabee@rutgers.edu}}}
\and
Bob Carpenter\thanks{Center for Computational Mathematics, Flatiron Institute, \href{mailto:bcarpenter@flatironinstitute.org}{\texttt{bcarpenter@flatironinstitute.org}}}
\and
Tore Selland Kleppe\thanks{Department of Mathematics and Physics, University of Stavanger, 
\href{mailto:mtore.kleppe@uis.no}{\texttt{tore.kleppe@uis.no}}}
\and
Sifan Liu\thanks{Center for Computational Mathematics, Flatiron Institute, 
\href{mailto:sliu@flatironinstitute.org}{\texttt{sliu@flatironinstitute.org}}}
}

\begin{document}

\maketitle

\begin{abstract}
Locally adapting parameters within Markov chain Monte Carlo methods while preserving reversibility is notoriously difficult.  The success of the No-U-Turn Sampler (NUTS) largely stems from its clever local adaptation of the integration time in Hamiltonian Monte Carlo via a geometric U-turn condition.  However, posterior distributions frequently exhibit multi-scale geometries with extreme variations in scale, making it necessary to also adapt the leapfrog integrator’s step size locally and dynamically.  Despite its practical importance, this problem has remained  largely open since the introduction of NUTS by Hoffman and Gelman (2014).
To address this issue, we introduce the Within-orbit Adaptive Leapfrog No-U-Turn Sampler (WALNUTS), a generalization of NUTS that adapts the leapfrog step size at fixed intervals of simulated time as the orbit evolves. At each interval, the algorithm selects the largest step size from a dyadic schedule that keeps the energy error below a user-specified threshold. Like NUTS, WALNUTS employs biased progressive state selection to favor states with positions that are further from the initial point along the orbit.  Empirical evaluations on multiscale target distributions, including  Neal’s funnel and the Stock-Watson stochastic volatility time-series model, demonstrate that WALNUTS achieves substantial improvements in sampling efficiency and robustness compared to standard NUTS.  
\end{abstract}

\section{Introduction \& Motivation}

\paragraph{Hamiltonian Monte Carlo and the Role of Discretization}

Hamiltonian Monte Carlo (HMC) is a Markov chain Monte Carlo (MCMC) method for sampling from probability distributions with  continuously differentiable densities $\mu: \mathbb{R}^d \to \mathbb{R}$ that are known up to normalization \cite{DuKePeRo1987,Ne2011}. HMC generates proposals by simulating a measure-preserving  flow on an extended phase space. This is done by introducing an auxiliary momentum variable  $\rho \in \mathbb{R}^d$ and a positive-definite mass matrix $M \in \mathbb{R}^{d \times d}$, and evolving the system according to the Hamiltonian dynamics corresponding to the Hamiltonian function $H(\theta, \rho) = -\log \mu(\theta) + \frac12 \rho^{\top} M^{-1} \rho$.  At all times, both forward and backward, the resulting Hamiltonian flow is time-reversible, exactly conserves phase-space volume and, in the absence of discretization error, the total energy. Hence, this flow preserves the probability measure on $\mathbb{R}^{2d}$ with density proportional to $e^{-H(\theta,\, \rho)}$. 

This measure-preserving property makes Hamiltonian flows a natural tool for constructing MCMC methods.  However, in most cases, the Hamiltonian flow cannot be computed analytically and must be approximated numerically, most commonly using the leapfrog integrator \cite{HaLuWa2010,BoSaActaN2018}.  Leapfrog is a time-reversible method that preserves volume exactly but only approximately conserves energy. These geometric properties makes it a natural choice for constructing Metropolis proposals, as we briefly review in Section~\ref{sec:short_overview_hmc}. Crucially, however, the leapfrog algorithm is only conditionally stable.  The numerical trajectory may diverge if the step size is too large relative to the local curvature of the log target density, $\log \mu$, as quantified  by the largest eigenvalue in magnitude of its  Hessian  \cite{BoVa2012}.  This imposes constraints on step size selection, and highlights the importance of carefully discretizing the Hamiltonian dynamics to ensure efficient sampling.

\paragraph{The No-U-Turn Sampler}

Two tuning decisions are especially critical to the performance of HMC,
\begin{itemize}
\item how far to integrate Hamilton's equations (the integration time), and
\item how finely to discretize the dynamics (the step size).
\end{itemize}
The No-U-Turn Sampler (NUTS) \cite{HoGe2014} addresses the first question by eliminating the need to tune the integration time a priori.  Instead, NUTS adaptively constructs an orbit by simulating Hamiltonian dynamics in both forward and backward time directions.  The orbit is built recursively using a doubling procedure: at each iteration, a direction (forward or backward) is chosen uniformly at random, and the orbit is extended in that direction by doubling its length.   The result is a balanced binary tree of candidate states, where internal nodes  represent recursive doubling stages and leaves correspond to individual states.

The iteration continues until a U-turn condition is met (see Figure~\ref{fig:nuts_intro}).  This condition checks whether the orbit has begun to reverse direction by evaluating inner products between momentum and position vectors across pairs of states, as we briefly review in Section~\ref{sec:short_overview_nuts}. To ensure reversibility, this condition is applied  recursively across all pairs of left and right subtrees in the binary tree built during orbit expansion.  This construction ensures that the probability of  any given orbit is independent of the starting point within the orbit.

Once the orbit is constructed, a proposal is sampled from the candidate states using a biased progressive sampling scheme that preserves detailed balance \cite{HoGe2014,betancourt2017conceptual}. By adapting the integration time to the local geometry of the target distribution, NUTS achieves robust performance without requiring manual tuning.  This local adaptivity of integration time has made NUTS one of the most widely used MCMC methods for sampling from distributions with continuously differentiable densities \cite{HoGe2014,betancourt2017conceptual,carpenter2016stan,BoOb2024}.

\begin{figure}[t]
    \centering
        \includegraphics[width=0.8\linewidth]{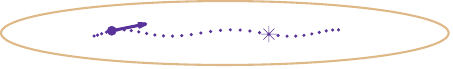} 
    \caption{The line indicates a one-sigma contour of a bivariate Gaussian target. The dots represent an orbit of consecutive leapfrog iterates in position space.  The large purple dot and arrow mark the initial position and momentum, while the asterisk marks the next state selected by the sampler. In this illustration, the U-turn condition in NUTS successfully detects the local scale of the target and terminates the trajectory accordingly.   }
    \label{fig:nuts_intro}
\end{figure}

\paragraph{The Limitations of a Fixed Step Size}

While NUTS effectively eliminates the need to tune the integration time by adapting the orbit length to the local geometry, it still uses a single, fixed step size throughout the entire sampling process. This global step size can be limiting in multiscale settings, where the local curvature of the target distribution may vary dramatically across regions. Because the leapfrog integrator is only conditionally stable, a step size that performs well in flat regions may lead to numerical divergence or large energy errors in regions of high curvature. Conversely, calibrating the step size to accommodate the worst-case curvature results in unnecessarily fine discretization and inefficient exploration in flatter regions.

Recent work has introduced orbit-level  adaptive step-size variants of NUTS that adjust the leapfrog step size globally for each orbit while preserving detailed balance \cite{BouRabeeCarpenterKleppeMarsden2024}. However, these methods do not address the need for finer-grained step size adaptation within each orbit.

\paragraph{Introducing WALNUTS}

To address this limitation, we introduce WALNUTS (Within-orbit Adaptive Leapfrog No-U-Turn Sampler), a generalization of NUTS that incorporates local step size adaptation within each orbit.
The key idea is to allow the leapfrog step size to adapt dynamically to local geometric features of the target distribution, rather than relying on a single, globally calibrated step size. WALNUTS selects the step size independently at each integration step of the orbit, ensuring numerical stability and controlling energy error as the orbit moves through regions of varying curvature. This enables the sampler to integrate stably through high-curvature regions while progressing efficiently through flatter ones without sacrificing detailed balance or requiring manual step size tuning.

To formalize this adaptivity, we distinguish between macro steps and micro steps. A macro step is the unit of simulated time over which candidate states are generated and added to the orbit. Each macro step is implemented by applying the leapfrog integrator over a sequence of micro steps. The micro step size is chosen from a dyadic schedule so that the total energy error over the macro step remains below a user-specified threshold.   At each macro step, WALNUTS searches for the coarsest micro step size that satisfies this energy-based criterion. This allows different segments of the orbit to be integrated at different resolutions, depending on local curvature and stability requirements.

The resulting orbit consists of a sequence of states lying on a fixed macro grid, with each macro step integrated using a locally adapted micro step size.  Importantly, although the micro step size may vary from one macro step to the next, the map from one candidate state to the next remains volume-preserving: an essential property for maintaining reversibility and correctness, as detailed in Section~\ref{sec:reversibility}. Once the full orbit is constructed, WALNUTS samples a single state from among the macro steps using the same biased progressive scheme as NUTS, which preserves detailed balance and ensures correctness. This local adaptivity allows the algorithm to remain both stable in regions of high curvature and efficient in flatter regions without requiring a globally tuned step size.  

\paragraph{Visualization of a WALNUTS Transition Step}

Figure~\ref{fig:walnuts_intro} illustrates the  mechanics of a WALNUTS transition. Each panel shows a different realization, all starting from the same initial condition, in a 2D slice of Neal’s funnel distribution. The orbit is color-coded according to the locally selected micro step size, with red indicating larger values and blue smaller ones.  These gradients visually demonstrate how WALNUTS adjusts its numerical resolution in response to local curvature. Insets display the corresponding weight profiles used during biased progressive sampling. These weights incorporate not only the target density at each state, as in standard NUTS, but also account for the possibility that the reverse trajectory may  follow a different step-size schedule --- ensuring that detailed balance is preserved under local step-size adaptivity.

\begin{figure}[t]
    \centering
        \includegraphics[width=0.31\linewidth]{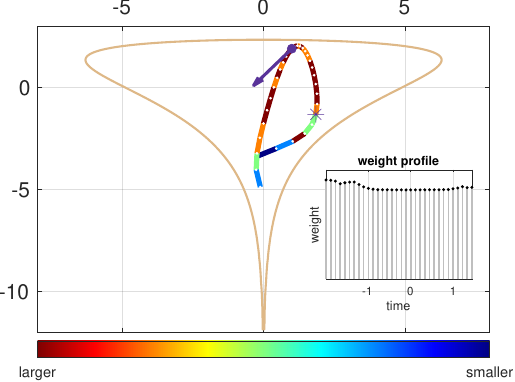} 
    \includegraphics[width=0.31\linewidth]{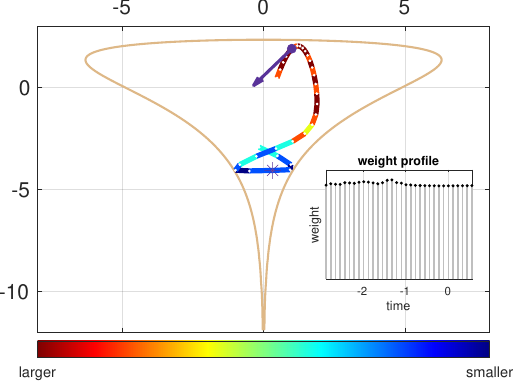} 
    \includegraphics[width=0.31\linewidth]{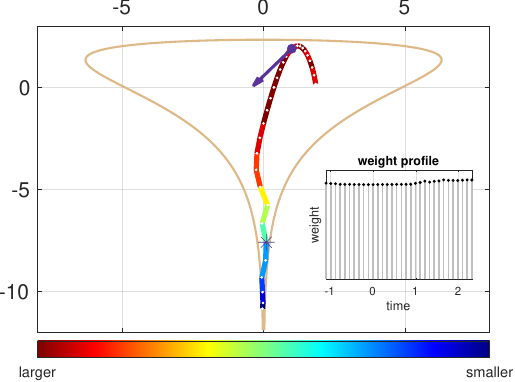}
    \caption{Three stochastic realizations of a WALNUTS transition from the same initial point in a 2D slice of Neal’s funnel, each generated using a different random seed. The purple dot and arrow indicate the initial position and velocity, and the asterisk marks the selected next state. Each panel shows a full orbit represented by a sequence of white dots constructed.  Each macro step is color-coded according to the micro step size required to keep the energy error below a user-specified threshold, with a red-to-blue gradient indicating larger to smaller micro step sizes. Insets show the corresponding weight profiles used in biased progressive sampling. Differences between panels reflect the stochastic nature of state selection and orbit construction forward and backward in time.}
    \label{fig:walnuts_intro}
\end{figure}

\paragraph{General Applicability and Theoretical Considerations}

Although we present adaptive leapfrog integration in the context of NUTS, the underlying idea is more broadly applicable. The same energy-controlled leapfrog step size adaptation can be incorporated, with minimal modification, into other Hamiltonian-based samplers such as standard HMC, multinomial HMC, and other reversible variants, and also non-reversible variants such as the Metropolis-adjusted Kinetic Langevin Algorithm \cite{BouRabeeOberdoerster2024} and generalized HMC \cite{Ho1991,turok2024sampling}. The idea may also be relevant for related Hamiltonian descent optimization methods \cite{fu2025hamiltonian}. 

More generally, while it is well known in geometric integration theory that symplectic integrators with variable step sizes may fail to preserve energy over very long trajectories \cite{stofer1988some, calvo1993development}, this limitation is less restrictive in the context of MCMC. The goals of MCMC differ fundamentally: rather than requiring long-time energy conservation, what matters is the preservation of the correct invariant distribution.  For this purpose, it is sufficient that, conditional on the locally chosen micro step size, the sequence of leapfrog micro steps used to perform a single macro integration step remains volume-preserving. This property ensures that the overall transition kernel is reversible with respect to the target distribution, as established in Section~\ref{sec:reversibility}.

\paragraph{Main Contributions}

This paper introduces WALNUTS, a step-size-adaptive generalization of the No-U-Turn Sampler (NUTS) that dynamically adjusts the leapfrog step size within each macro step of an orbit. The main contributions of the paper are:

\begin{itemize}
\item \textbf{Within-orbit step size adaptation.} We propose a principled framework for dynamic step size selection within Hamiltonian trajectories. At each integration interval (macro step), WALNUTS chooses the coarsest step size from a dyadic schedule that keeps the accumulated energy error below a user-specified threshold. This enables stable integration through high-curvature regions while accelerating through flatter regions. In contrast to global or orbit-level adaptation, this within-orbit mechanism provides the fine-grained control needed to handle multiscale targets effectively.  To our knowledge, this level of local adaptivity has not previously been achieved in reversible HMC-based samplers.

\item \textbf{Reversibility under local adaptivity.} We establish that WALNUTS satisfies detailed balance with respect to the target distribution (Theorem~\ref{thm:walnuts}). To do so, we lift the algorithm to an extended state space $\calZ$ and introduce a carefully designed involution $\Psi : \calZ \to \calZ$, i.e., $\Psi\circ\Psi$ is the identity map (see~\eqref{eq:walnuts-inv}). 
As shown in Lemma~\ref{lem:measure-preserving-involution-Psi}, $\Psi$ preserves the reference measure $\zeta$ on $\calZ$. Lemma~\ref{lem: p joint bp} further shows that the extended target density is invariant under $\Psi$.  Together, these properties imply that WALNUTS satisfies detailed balance without requiring a Metropolis correction---just as in standard NUTS.  

This is a subtle and nontrivial result: while individual leapfrog steps are time-reversible and volume-preserving, their composition with locally varying step sizes typically breaks both properties.  A key contribution of our analysis is the careful design of the measure-preserving involution $\Psi$, which situates WALNUTS firmly within the auxiliary-variable-and-involution framework developed in recent work \cite{AndrieuLeeLivingstone,Glatt-Holtz_Krometis_Mondaini_2023,BouRabeeCarpenterMarsden2024,Glatt-Holtzetal2024}, and  provides a rigorous foundation for incorporating within-orbit local adaptivity into Hamiltonian-based MCMC samplers. 

\item  \textbf{Empirical improvements across multiscale targets.}  We benchmark WALNUTS on challenging posterior geometries, including Neal’s funnel and the Stock-Watson time-series model, and demonstrate substantial gains in sampling efficiency over standard NUTS. The results highlight that WALNUTS adapts effectively to local curvature without interfering with the integration-time adaptivity central to NUTS.  Moreover, WALNUTS  successfully fits complex hierarchical Bayesian models that standard NUTS fails to handle reliably.

\item \textbf{Accessible and efficient implementations.}  To promote adoption and reproducibility, we provide both a pedagogically transparent implementation in python that mirrors the pseudocode in the manuscript, and a memory-efficient C++ version optimized for large-scale applications. These complementary implementations are designed to support both practical deployment and future research. The results and plots can be reproduced from these implementations, which are available from GitHub under a permissive open-source license.\footnote{The code is distributed under the MIT License at \url{https://github.com/bob-carpenter/walnuts}.}  

\end{itemize}

\paragraph{Tuning Considerations.} WALNUTS requires three user-specified parameters: a mass matrix, the energy error threshold that controls the micro step size selection at each macro step, and the macro step size itself, which determines the temporal spacing between states along the final orbit.   Since the theoretical foundations of warmup adaptation remain underdeveloped,  we explore practical strategies for selecting the energy error threshold and macro step size during an initial warmup phase in Appendix~\ref{sec:tuning}. This phase is used solely to calibrate algorithmic parameters prior to collecting posterior samples. Tuning the mass matrix  presents additional challenges, particularly when it is restricted to be diagonal for computational reasons \cite{kleppe2016adaptive,hird2023quantifying,tran2024tuning}.   One approach to diagonal mass matrix preconditioning is by approximating the target density with a multivariate normal distribution having a diagonal covariance matrix \cite{margossian2023shrinkage,margossian2024variational}.


\paragraph{Related Work} Here we situate this work in the broader context of locally adaptive HMC methods.   The AutoMALA sampler has been proposed as a self-tuning version of the Metropolis-adjusted Langevin algorithm (MALA) \cite{kleppe2016adaptive,biron2024automala}, which can be viewed as a one-step variant of HMC.  These methods use a forward–reverse stochastic adaptation scheme that preserves detailed balance.   The more recent AutoStep sampler \cite{liu2024autostep} builds on similar principles and is closely related to earlier work on the GIST sampler \cite{BouRabeeCarpenterMarsden2024}, which provides a general framework for locally adaptive HMC (reviewed in Section~\ref{sec:gist}).

A different but related class of methods involves the delayed rejection algorithm \cite{mira2001metropolis,green2001delayed}, which is a generalization of Metropolis-Hastings to a sequence of proposal moves that can start ``bold'' and become  increasingly ``timid''  \cite{haario2006dram}; i.e., if a first  proposal with a bold hyper-parameter is rejected, then a second proposal with a timider hyper-parameter is attempted, and so on.  The hope is that the delayed rejection probability can be made arbitrarily small with increasingly timid proposals.  This strategy has been applied to adaptively select orbit-level leapfrog step sizes in HMC \cite{modi2023delayed}, and more recently, to adapt the leapfrog step size within each generalized HMC step \cite{turok2024sampling}. In both approaches, however, key tuning parameters such as path length and friction must be specified in advance and are not adapted locally.

\paragraph{Organization of Paper}

The remainder of the paper is organized as follows. Section~\ref{sec:preliminaries} reviews background on HMC and NUTS, introducing key concepts such as leapfrog integration, auxiliary variables, involutions, and orbits, and states a general reversibility theorem that provides the theoretical foundation for our method. Section~\ref{sec:walnuts} presents the WALNUTS algorithm, beginning with the construction of a variable step size leapfrog integrator in Section~\ref{sec:variable_step_int}, and then describing the corresponding modifications to orbit construction in Section~\ref{sec:walnuts-orbit} and integration time selection in Section~\ref{sec:walnuts-int-time}. Section~\ref{sec:reversibility} establishes the reversibility of WALNUTS by verifying the conditions of the general theorem introduced in Section~\ref{sec:preliminaries}. Sections~\ref{sec:gaussian} and~\ref{sec:funnel} evaluate the performance of WALNUTS on multivariate Gaussian targets and Neal’s funnel distribution, respectively. Section~\ref{sec:stock-watson} applies WALNUTS to a high-dimensional Bayesian inference problem in macroeconomics. Appendix~\ref{sec:pseudocode} provides pseudocode to support conceptual understanding of the algorithm, while Appendix~\ref{sec:practical} outlines practical strategies for optimizing the implementation and Appendix~\ref{sec:tuning} concludes with some approaches to parameter tuning to be applied during an initial warmup phase.



\section{Preliminaries}

\label{sec:preliminaries}

We first provide some background on Hamiltonian Monte Carlo and its locally adaptive variants.

\subsection{Hamiltonian Flow} Let $\mu$ be a given target probability measure on $\mathbb{R}^d$ with unnormalized density also denoted by $\mu$.  To sample from $\mu$, HMC-type methods extend the state space $\mathbb{R}^d$ to phase space $\mathbb{R}^{2d}$ by introducing an auxiliary momentum variable $\rho\in\R^d$ and a symmetric positive-definite mass matrix $M\in\R^{d\times d}$. The algorithm approximates a Hamiltonian flow in this extended phase space that preserves the joint distribution $\hat{\mu} = \mu \otimes \N(0, M)$ with density 
\begin{equation} \label{eq:hatmu}
    \hat{\mu}(\theta,\rho) \  \propto \  \exp(-H(\theta, \rho) ) \;, ~~ \text{where} ~~
H(\theta,\rho)  = -\log \mu(\theta) + \frac12 \rho^{\top} M^{-1} \rho  \;.
\end{equation}
The corresponding exact Hamiltonian flow $\varphi_t : \mathbb{R}^{2d} \to \mathbb{R}^{2d}$ maps an initial state $(\theta, \rho) \in \mathbb{R}^{2d}$ to $(\theta_t, \rho_t)$ at time $t \in \mathbb{R}$, where $(\theta_t, \rho_t)$ evolves according to the differential equations
\begin{equation}
\label{eq:hamiltonian_dynamics}
\frac{d}{dt} \theta_t = M^{-1} \rho_t, \quad \frac{d}{dt} \rho_t = \nabla \log \mu(\theta_t), \quad \text{with } (\theta_0, \rho_0) = (\theta, \rho) \;.
\end{equation}
This flow preserves the Hamiltonian function $H$ and Lebesgue measure on $\mathbb{R}^{2d}$.  Consequently, it preserves the extended target distribution $\hat{\mu}$.

\subsection{Fixed Step Size Leapfrog Integrator} 

Because the exact Hamiltonian flow in \eqref{eq:hamiltonian_dynamics} cannot be computed in closed form for a general target $\mu$, it is approximated numerically. HMC-type methods almost always use the leapfrog integrator for this purpose, as it is computationally efficient and preserves key geometric properties of the Hamiltonian flow.

Fix a step size $h>0$. The leapfrog integrator $\Phi_h: \mathbb{R}^{2d} \to \mathbb{R}^{2d}$  updates a state $(\theta, \rho)$ to $(\theta', \rho') = \Phi_h(\theta, \rho)$ according to 
\begin{equation} \label{eq:leapfrog} \theta' \ = \ \theta+ M^{-1} \left( h \rho + \frac{h^2}2 \nabla \log \mu(\theta) \right)  \;, \quad \rho' \ = \  \rho+\frac h2 \left( \nabla  \log \mu(\theta)+ \nabla  \log \mu(\theta') \right). \end{equation}
Besides being explicit, the leapfrog integrator has several important properties.
\begin{itemize}
\item \textbf{Volume preservation.} The map preserves phase space volume, i.e., $\left|\det(D\Phi_h)\right| \equiv 1$, where $D\Phi_h$ denotes the Jacobian of $\Phi_h$.

\item \textbf{Symmetry.} The inverse map satisfies $\Phi_h^{-1} = \Phi_{-h}$. 

\item \textbf{Time reversibility.} The inverse map satisfies \[
\Phi_h^{-1} = \mathcal{F} \circ \Phi_h \circ \mathcal{F}
\] where $\mathcal{F}: \mathbb{R}^{2d} \to \mathbb{R}^{2d}$ is the momentum flip involution defined by $\mathcal{F}(\theta,\rho) = (\theta,-\rho)$ for $(\theta,\rho) \in \mathbb{R}^{2d}$.   In other words, applying a momentum flip before and after the integrator yields the inverse map.
\item  \textbf{Energy error.} the Hamiltonian is not exactly conserved; in general, $(H\circ\Phi_h-H) \not\equiv 0$.
 \end{itemize}

 For any $L \in \mathbb{Z}$, the $L$-step leapfrog map $\Phi_{h}^L:  \mathbb{R}^{2d} \to \mathbb{R}^{2d}$ allows for both forward ($L > 0$) and backward ($L < 0$) integration steps, and is defined recursively as 
 \begin{align*}
    &\Phi_h^{L+1} = \Phi_h \circ \Phi_h^{L},\quad \Phi_h^{L-1} = \Phi_h^{-1} \circ \Phi_h^{L},
 \end{align*}
 where $\Phi^0_h$ is the identity map.   The time-reversibility of the leapfrog integrator extends to  $\Phi_h^L$: applying a momentum flip before and after $\Phi_h^L$ yields its inverse. This property is captured by the following commutative diagram: 
 \begin{center}
 \begin{tikzcd} (\theta, \rho) \arrow[r, "\Phi_h^L"] \arrow[d, "\mathcal{F}"'] & (\theta', \rho') \arrow[d, "\mathcal{F}"] \\ (\theta, -\rho) & (\theta', -\rho') \arrow[l, "\Phi_h^L"'] \end{tikzcd} \end{center} 
It is important to note, however, that the composition of time-reversible maps is not necessarily time-reversible; see \cite[Theorem 4.2]{BoSaActaN2018}. This observation becomes relevant later when we consider compositions of variable step size leapfrog integrators: although each integrator in the composition is individually time-reversible, their composition generally is not.


\subsection{Gibbs Self-Tuning for Locally Adaptive HMC}

\label{sec:gist}

A general and flexible strategy for constructing  Markov chains that are reversible with respect to a given target distribution $\mu$ is to augment the state space with an auxiliary variable $ v \in \mathbb{V}$, apply a measure-preserving involution on the augmented space $\mathbb{A} = \mathbb{R}^d \times \mathbb{V}$, and incorporate a Metropolis–Hastings correction.  This auxiliary-variable-and-involution strategy underlies many advanced MCMC methods  and has been formalized in recent literature \cite{AndrieuLeeLivingstone,Glatt-Holtz_Krometis_Mondaini_2023,Glatt-Holtzetal2024}. 

More recently, this strategy was extended to a broad class of locally adaptive HMC samplers through the GIST (Gibbs Self-Tuning) framework \cite{BouRabeeCarpenterMarsden2024}. GIST constructs adaptive samplers by  Gibbs sampling the HMC algorithm’s tuning parameters---such as path length, step size, and mass matrix---conditionally on the current state. This unifying framework includes randomized HMC \cite{BoSa2017,BoEb2022,kleppe2022connecting}, multinomial HMC \cite{betancourt2017conceptual,xu2021couplings}, the No-U-Turn Sampler \cite{HoGe2014,betancourt2017conceptual}, and the Apogee-to-Apogee Path Sampler \cite{SherlockUrbasLudkin2023Apogee} as special cases.

We now describe a single transition step of a GIST sampler.
Let $ \mathbb{V}$ denote the auxiliary variable space, and let $p_a( v \mid \theta )$ be a conditional probability density on $ \mathbb{V}$ given $\theta \in \mathbb{R}^d$.   In the locally adaptive HMC samplers we will study, $v$ typically consists of the momentum variable $\rho$ along with additional tuning parameters. The extended target density on the augmented space $\mathbb{A} = \mathbb{R}^d \times \mathbb{V}$ is defined by  \[
\hat{\mu}(\theta, v) \ = \ \mu(\theta) p_a(v \mid \theta) \;,
\]
interpreted as a density with respect to a reference measure $\zeta$ on $\mathbb{A}$. This reference measure specifies how we assign ``volume'' or ``weight'' to regions of the augmented state space. For the original state variable $\theta \in \mathbb{R}^d$, we use the standard Lebesgue measure. For the auxiliary variable $v \in \mathbb{V}$, we use a measure which may be counting measure (if $v$ is discrete) or Lebesgue measure (if $v$ is continuous).

Given the conditional density $p_a( v \mid \theta )$, the extended density $\hat{\mu}(\theta, v)$, an involution $\Psi: \mathbb{A} \to \mathbb{A}$ and the current state $\theta_n \in \mathbb{R}^d$, a transition step of a GIST sampler computes an updated state $\theta_{n+1} \in \mathbb{R}^d$ as follows: \begin{enumerate}
    \item \textbf{Auxiliary variable refreshment.} Sample $v_n \sim p_a( \cdot \mid \theta_n)$.
    \item  \textbf{Involution-based proposal.} Compute a proposal $(\tilde{\theta}_{n+1},\tilde{v}_{n+1}) = \Psi(\theta_n, v_n)$.
    \item \textbf{Metropolis correction.}  Accept the proposal with probability \[
    \alpha(\theta_n, v_n) = \min\left( 1, \frac{\hat{\mu}(\tilde{\theta}_n, \tilde{v}_n)}{\hat{\mu}(\theta_n, v_n)} \right)
    \]
    and set \[
\theta_{n+1} = \begin{cases} \tilde{\theta}_n & \text{with probability $ \alpha(\theta_n, v_n)$, } \\
\theta_n & \text{otherwise.} 
\end{cases}
\]
\end{enumerate}
  Consider the Markov chain on $\mathbb{R}^d$ obtained by iterating this GIST transition step from a given initial distribution. Since the auxiliary variables are fully resampled in each step and discarded after they are used, the marginal chain on $\theta$ inherits reversibility with respect to $\mu$, as formalized below.

\medskip

\begin{theorem}[e.g., {\cite{BouRabeeCarpenterMarsden2024}}] \label{thm:AVM_reversibility}
Suppose $\Psi: \mathbb{A} \to \mathbb{A}$ is an involution (i.e., $\Psi^2 = \mathrm{id}$) and preserves the reference measure $\zeta$ on $\mathbb{A}$. Then, the Markov chain defined by the GIST sampler is reversible with respect to $\mu$.
\end{theorem}

In the following subsections, we illustrate how this GIST framework gives rise to increasingly sophisticated samplers, beginning with HMC with a fixed integration time, followed by biased progressive HMC with a randomized integration time, and culminating in the No-U-Turn Sampler.

\subsection{Fixed Integration Time Hamiltonian Monte Carlo}

\label{sec:short_overview_hmc}

Hamiltonian Monte Carlo (HMC) with a fixed integration time not only fits naturally into the auxiliary-variable framework based on measure-preserving involutions, it may well have inspired it.  In this setting, the auxiliary variable is a momentum vector $\rho \in \mathbb{R}^d$ with conditional density  \[
p_a(\rho \mid \theta) = \mathcal{N}(\rho \mid 0, M) =  \frac{1}{(2\pi)^{d/2} \det(M)^{1/2}} \exp\left(-\tfrac{1}{2} \rho^{\top} M^{-1} \rho \right),
\]
corresponding to a Gaussian distribution with mean zero and covariance matrix $M$, independent of $\theta$. The augmented space is $\mathbb{A} = \mathbb{R}^{2d}$, the reference measure $\zeta$ is the standard Lebesgue measure on $\mathbb{R}^{2d}$, and the augmented target distribution is given by \eqref{eq:hatmu}. 

The involution $\Psi: \mathbb{A} \to \mathbb{A}$ is defined by simulating Hamiltonian dynamics using $i\in\bbN$ leapfrog steps with a fixed step size $h>0$, followed by a momentum flip:
\[
\Psi = \mathcal{F} \circ \Phi_h^i ,
\] 
where $\Phi_h^i$ is the $i$-step leapfrog integrator, and $\mathcal{F}(\theta, \rho) = (\theta, -\rho)$ denotes the momentum-flip map. This composition defines a volume-preserving (i.e., $\zeta$-preserving) involution on the augmented space $\mathbb{A}$.

Thus, HMC with fixed integration time corresponds to the following instance of a GIST sampler. The steps of this sampler are summarized below and illustrated schematically in the diagram that follows. \begin{enumerate}
    \item \textbf{Auxiliary variable refreshment.} Sample $\rho_n \sim \mathcal{N}(0, M)$.
    \item \textbf{Involution-based proposal.} Compute $(\tilde{\theta}_{n+1}, \tilde{\rho}_{n+1}) = \Psi(\theta, \rho) = \mathcal{F} \circ \Phi_h^i(\theta_n, \rho_n)$.
    \item \textbf{Metropolis correction.} Accept the proposal with probability
    \[
    \alpha((\theta_n, \rho_n),(\tilde{\theta}_{n+1},\tilde{\rho}_{n+1})) = \min\left(1, \frac{\hat{\mu}(\tilde{\theta}_{n+1}, \tilde{\rho}_{n+1})}{\hat{\mu}(\theta_n, \rho_n)}\right)
    = \min\left(1, \exp\left([H(\theta_n, \rho_n) - H(\tilde{\theta}_{n+1}, \tilde{\rho}_{n+1})]\right)\right).
    \]
\end{enumerate}
After the accept/reject step, the auxiliary variable is discarded. The proposal map $\Psi = \mathcal{F} \circ \Phi_h^L$ is a volume-preserving involution: it satisfies $\Psi^2 = \mathrm{id}$ and $\left|\det(D\Psi)\right| = 1$. By Theorem~\ref{thm:AVM_reversibility}, it follows that the resulting Markov chain on $\theta$ is reversible with respect to $\mu$.

\begin{center}
\begin{tikzpicture}[>=stealth, scale=1.]

\node[draw, rounded corners, align=center] (start) at (-1.2,0) {$\theta_n$ \\[1pt] current state};
\node[draw, rounded corners, align=center] (rho) at (2.1,0) {sample $\rho_n \sim \mathcal{N}(0,M)$ \\[1pt] auxiliary variable};
\node[draw, rounded corners, align=center] (involution) at (6.3,0) {propose \\ $(\tilde{\theta}_{n+1},\tilde{\rho}_{n+1}) = \Psi(\theta_n,\rho_n)$};
\node[draw, rounded corners, align=center] (accept) at (10.25,0) {accept/reject \\ using $\hat{\mu}$};
\node[draw, rounded corners, align=center] (end) at (13,0) {$\theta_{n+1}$ \\[1pt] next state};

\draw[->] (start) -- (rho);
\draw[->] (rho) -- (involution);
\draw[->] (involution) -- (accept);
\draw[->] (accept) -- (end);

\node[below=5pt of rho] {\small Gibbs refreshment};
\node[below=5pt of involution] {\small involution: $\Psi=\mathcal{F} \circ \Phi_h^i$};
\node[below=5pt of accept] {\small corrects for energy error};
\end{tikzpicture}
\end{center}

To better understand why the Metropolis correction is necessary, note that applying $\Psi$ simply swaps the roles of the current and proposed states without changing phase space volume. If $\Psi$ exactly preserved the joint density  (i.e., if $\hat{\mu} \circ \Psi = \hat{\mu}$), no Metropolis correction would be needed. However, because the leapfrog integrator introduces energy error, $\Psi$ does not exactly preserve $\hat{\mu}$, and the acceptance probability serves to correct for this discrepancy.

By definition of the acceptance probability, we have
\[
\hat{\mu}(\theta_n, \rho_n) \; \alpha((\theta_n, \rho_n),(\tilde{\theta}_{n+1},\tilde{\rho}_{n+1})) = \hat{\mu}(\tilde{\theta}_{n+1}, \tilde{\rho}_{n+1}) \; \alpha((\tilde{\theta}_{n+1},\tilde{\rho}_{n+1}), (\theta_n, \rho_n)),
\] 
so the product $\hat{\mu}(\theta_n, \rho_n) \; \alpha((\theta_n, \rho_n,(\tilde{\theta}_{n+1},\tilde{\rho}_{n+1}))$ is symmetric in the current state $(\theta_n, \rho_n)$ and the proposed state $(\tilde{\theta}_{n+1}, \tilde{\rho}_{n+1})$.  This symmetry, combined with the volume-preserving and involutive nature of $\Psi$, ensures that HMC is reversible with respect to $\hat{\mu}$, and thus that the marginal chain on $\theta$ is reversible with respect to $\mu$.

As we will see in the next subsection, this logic extends naturally to biased progressive HMC, where the auxiliary variables include a random integration time.

\subsection{Randomized Integration Time with  Biased Progressive HMC}

\label{sec:short_overview_biased_progressive_hmc}

Biased progressive HMC introduces randomness into the integration time while preserving reversibility.  This is achieved through an extended set of auxiliary variables that consist of a leapfrog trajectory, called an \emph{orbit}, and a randomly selected index that labels a point along this orbit.  This index determines the next state of the chain and can be interpreted as a randomized integration time.

We introduce precise notation and structure here not only to support the proofs that follow, but also because this perspective helps clarify key ideas underlying NUTS. In particular, it provides a clear view of how NUTS operates on an extended space, where randomized integration times emerge naturally and reversibility is maintained without Metropolis-style rejections.

Let $h>0$ be a fixed step size.  Throughout this paper, an orbit refers to a sequence of phase space points computed using the leapfrog integrator and evaluated on a uniformly spaced time grid: \[
 t_k \ := \ k h~~\text{where}~~k \in \mathbb{Z} \;.
 \] 

\begin{definition} \label{defn:orbit}
An \emph{orbit} $\calO \subset \mathbb{R}^{2d}$ is a consecutive sequence of states generated by applying the leapfrog integrator with fixed step size $h>0$ to a given initial condition $(\theta, \rho) \in \mathbb{R}^{2d}$.  It takes the form:  \[
\calO \ = \  ((\theta_a, \rho_a), (\theta_{a+1}, \rho_{a+1}), \dots, (\theta_b, \rho_b)) \;,
\] where $a, b \in \mathbb{Z}$ with $a \le b$.  Each pair $(\theta_k, \rho_k) = \Phi_k^k(\theta, \rho)$ represents the state at time $t_k= k h$ obtained by applying leapfrog steps forward (if $k>0$) or backward (if $k<0$) in time from the initial point, for \[
k \in a{:}b = \{ a, a+1, \dots, b \}.
\] The \emph{length} of the orbit, denoted $|\calO|$, is $b - a + 1$.

\end{definition}

We  define a simple  operation to concatenate orbits, which will be used in orbit construction.

\medskip

\begin{definition} \label{defn:concatenation}
Let $a, b, c \in \mathbb{Z}$ with $a \le b < c$.
Given two orbits  \[
\calO=((\theta_a, \rho_a), \dots, (\theta_b, \rho_b)) \quad \text{and} \quad \widetilde{\calO}=((\theta_{b+1}, \rho_{b+1}),  \dots, (\theta_c, \rho_c)) \;,
\] their \emph{concatenation}, denoted $\calO \odot  \widetilde{\calO}$, is defined as
\[ \calO \odot  \widetilde{\calO} \ =\ ((\theta_a, \rho_a), \dots, (\theta_c,\rho_c))\;, \]
with length $|\calO \odot  \widetilde{\calO}| = c-a+1$.
\end{definition}

Biased progressive HMC augments the state space with the following auxiliary variables:
\begin{itemize}
  \item a momentum vector $\rho \in \mathbb{R}^d$,
\item an orbit $\calO \subset \mathbb{R}^{2d}$ generated from the initial condition $(\theta, \rho)$ via forward and/or backward integration,
\item an index $i \in \mathbb{Z}$ that selects a state within $\calO$ to serve as the next state of the Markov chain.
\end{itemize}  
We consider orbits of fixed length $2^m$ for some fixed $m \in \mathbb{N}$. Each orbit is determined by an initial condition $(\theta, \rho)$ and the index $b \in 0 {:} (2^m - 1)$ labelling its rightmost element. 
 
The augmented state space is 
\[
\mathbb{A} \ := \  \mathbb{R}^{d} \times \mathbb{R}^d \times \mathbb{N} \times \mathbb{Z} ,
\] 
with each augmented state denoted by $z = (\theta, \rho, b, i) \in \mathbb{A}$.   Let $\zeta$ denote the reference measure on $\mathbb{A}$,  given by the product of Lebesgue measure on the continuous components and counting measure on the discrete components.  The joint density on $\mathbb{A}$ takes the form
\[
p_{\joint}(z) \ \propto  \ e^{-H(\theta,\rho)} \cdot p_{\orbit}( b \mid \theta,\rho) \cdot p_{\ind}(i \mid \theta,\rho, b ) ,
\] where the proportionality symbol reflects that $e^{-H(\theta,\rho)}$ may not be normalized.
Here, 
$p_{\orbit}( b \mid \theta,\rho) = \Unif(b \mid 0 {:} (2^m - 1) )$, i.e. $b$ is uniformly distributed over $0{:}(2^m-1)$, and $p_{\ind}(i \mid \theta,\rho, b )$ is defined implicitly by the orbit construction procedure.
To construct the orbit, we draw $m$ independent Bernoulli$(1/2)$ random variables $B_1, \dots, B_m$, which determine the rightmost index of the orbit as
\[
b = \sum_{j=1}^m B_j 2^{j-1} \;,
\]  
and the leftmost index as $a = b - 2^m + 1$. Given $(\theta, \rho)$, the orbit $\calO$ and integration time index $i$ are sampled through the following recursive procedure.
\begin{enumerate}
  \item \textbf{Initialization.}
  Start with the singleton orbit $\mathcal{O}_0 = ((\theta, \rho))$, and set $a_0 = b_0 = i_0 = 0$.
  Sample $m$ i.i.d. Bernoulli$(1/2)$ random variables $B = (B_1, \dots, B_m)$.

  \item \textbf{Recursive construction.} For each $k = 0, 1, \dots, m - 1$:
  \begin{enumerate}
    \item \textbf{Current orbit.} Let $\mathcal{O}_k$ be the current orbit with index range $a_k{:}b_k$.

    \item \textbf{Extension.} Build an extension orbit $\mathcal{O}_k^{\ext}$ using leapfrog steps over the index range $a_k^{\ext}{:}b_k^{\ext}$ where
    \[
    a_k^{\ext} = \begin{cases} b_k + 1 & \text{if } B_{k+1} = 1, \\ a_k - 2^k & \text{if } B_{k+1} = 0, \end{cases}
    \quad
    b_k^{\ext} = \begin{cases} b_k + 2^k & \text{if } B_{k+1} = 1, \\ a_k - 1 & \text{if } B_{k+1} = 0. \end{cases}
    \]

    \item \textbf{Proposal index.} Sample $i_k^{\ext}$ from a categorical distribution over $a_k^{\ext} {:} b_k^{\ext}$ with probabilities proportional to the weights  $w_j \propto e^{-H(\theta_j, \rho_j)} $ where $j \in a_k^{\ext} {:} b_k^{\ext}$.

    \item \textbf{Metropolis acceptance.} Accept $i_k^{\ext}$ with probability:
    \begin{align}\label{equ: acc prob}
    \alpha((a_k, b_k), (a_k^{\ext}, b_k^{\ext})) = \min\left(1, \frac{\sum_{j \in a_k^{\ext}{:} b_k^{\ext}} e^{-H(\theta_j, \rho_j)}}{\sum_{j \in a_k{:}b_k} e^{-H(\theta_j, \rho_j)}} \right) \;. 
    \end{align}
    If accepted, $i_{k+1} = i_k^{\ext}$; otherwise, $i_{k+1} = i_k$.

    \item \textbf{Update orbit.} Concatenate to form the next orbit:
    \[
    \mathcal{O}_{k+1} = \begin{cases}
      \mathcal{O}_k \odot \mathcal{O}_k^{\ext} & \text{if } B_{k+1} = 1, \\
      \mathcal{O}_k^{\ext} \odot \mathcal{O}_k & \text{if } B_{k+1} = 0.
    \end{cases}
    \]
  \end{enumerate}

  \item \textbf{Output.} Return the final orbit $\mathcal{O}$ of length $2^m$ and index $i = i_m$.
\end{enumerate}

This recursive procedure  samples $i_{k+1}$ from the following mixture
\begin{align*}
&p_{\ind}(i_{k+1} \mid \theta, \rho, b_{k+1})
= \alpha((a_k, b_k), (a_k^{\ext}, b_k^{\ext}))  \cdot \frac{e^{-H(\theta_{i_{k+1}}, \rho_{i_{k+1}})}}{\sum_{j \in [a_k^{\ext} : b_k^{\ext}]} e^{-H(\theta_j, \rho_j)}} \cdot \Indc{i_{k+1} \in a_k^{\ext}{:}b_k^{\ext} } \\
& \qquad  + \left[ 1 -  \alpha((a_k, b_k), (a_k^{\ext}, b_k^{\ext})) \right] \cdot p_{\ind}(i_k \mid \theta, \rho, b_k) \cdot \Indc{i_{k+1} \in a_k{:}b_k},
\end{align*}
where $\alpha((a_k, b_k), (a_k^{\ext}, b_k^{\ext}))$ is the Metropolis acceptance probability defined in \eqref{equ: acc prob}.
The key to reversibility is that $\alpha$ is symmetric with respect to swapping the current and extension orbits, provided each orbit is weighted by the total stationary probability it contributes; that is, the sum of $e^{-H(\theta_j, \rho_j)}$ over its indices.

More precisely, let $i$ be the final index.  Then  either $i=0$ or there exists $k \in \{1, \dots, m\}$ such that, \[
i \in  a_k^{\ext}{:}b_k^{\ext} \;.
\]   In either case, we obtain the identity \begin{equation} \label{bphmc:index}
\begin{aligned}
& e^{-H(\theta_0, \rho_0)}  e^{-H(\theta_{i}, \rho_{i})}   \alpha((a_k,b_k), (a_k^{\ext}, b_k^{\ext}))  \sum_{j \in a_k {:} b_k} e^{-H(\theta_j, \rho_j)} \\
& \quad  = e^{-H(\theta_i, \rho_i)} e^{-H(\theta_0, \rho_0)}   \alpha( (a_k^{\ext}, b_k^{\ext}),(a_k,b_k)) \sum_{j \in a_k^{\ext} {:} b_k^{\ext}} e^{-H(\theta_j, \rho_j)} \;. 
\end{aligned}
\end{equation} Additionally, the leapfrog identity $\Phi_h^{b - i}(\theta_i, \rho_i) = \Phi_h^b(\theta, \rho)$ implies that the conditional density over orbits satisfies
\begin{equation} \label{bphmc:orbit}
p_{\orbit}(b \mid \theta, \rho) = p_{\orbit}(b - i \mid \theta_i, \rho_i).
\end{equation}

Now define the map $\Psi : \mathbb{A} \to \mathbb{A}$ by
\begin{equation} \label{bphmc:involution}
\Psi(\theta, \rho, b, i) = (\Phi_h^i(\theta, \rho), b - i, -i),
\end{equation} 
which is a $\zeta$-preserving involution.  Combining the identities \eqref{bphmc:index} and \eqref{bphmc:orbit}, we find that the joint density is invariant under $\Psi$
\[
p_{\text{joint}}(\theta, \rho, b, i) = p_{\text{joint}} \circ \Psi(\theta, \rho, b, i).
\]
In other words, the map $\Psi$ leaves the joint distribution unchanged. As a result, the acceptance probability in the corresponding GIST sampler is always equal to one, and no Metropolis correction is required.

The diagram below summarizes the update steps of biased progressive HMC. It highlights the absence of a Metropolis step, since the proposal involution $\Psi$ preserves both the reference measure $\zeta$ and the joint density.
\begin{center}
\begin{tikzpicture}[>=stealth, node distance=3.2cm, scale=1, every node/.style={scale=1}]
  \node[draw, rounded corners, align=center] (start) at (-2,0) 
    {$\theta_n$ \\[1pt] current state};
    
  \node[draw, rounded corners, align=center, right of=start] (orbit)
    {sample orbit $\calO$ \\[1pt] via doubling};
    
  \node[draw, rounded corners, align=center, right of=orbit] (index)
    {sample index $i$ \\[1pt] from $\calO$};
    
  \node[draw, rounded corners, align=center, right of=index] (involution)
    {apply involution \\ $\Psi(\theta_n, \rho_n, b, i)$};
    
  \node[draw, rounded corners, align=center, right of=involution] (end)
    {$\theta_{n+1}$ \\[1pt] next state};

  \draw[->] (start) -- (orbit);
  \draw[->] (orbit) -- (index);
  \draw[->] (index) -- (involution);
  \draw[->] (involution) -- (end);

  \node[below=.2cm of index, align=center] 
  {\small No Metropolis correction step is required, because \\ 
   \small the involution  $\Psi$ preserves  both $\zeta$ and the joint density};
\end{tikzpicture}
\end{center}

By Theorem~\ref{thm:AVM_reversibility}, the resulting Markov chain is reversible with respect to $\mu$.
We will return to this point in Section~\ref{sec:reversibility}, where we extend this reversibility result to the full WALNUTS algorithm, which jointly adapts both the integration time and the step size. Before that, we consider a conceptually important special case of biased progressive HMC in which the leapfrog integrator is replaced by the exact Hamiltonian flow.

\subsection{Idealized Case: Exact Biased Progressive HMC}

To better understand the behavior of biased progressive HMC, it is instructive to consider an idealized version in which the leapfrog integrator $\Phi_h$ with fixed step size $h>0$ is replaced with the exact Hamiltonian flow $\varphi_h$. In this setting, no numerical integration error is introduced, and the Hamiltonian is conserved exactly at each point along the orbit. As a result, all the Metropolis-style acceptance probabilities are equal to one, and every proposal index is accepted.

This idealization  provides valuable insight into the behavior of biased progressive HMC when leapfrog integration errors are negligible. In particular, the distribution of the final integration time index $i$ becomes independent of the target distribution, depending only on the randomized orbit construction procedure. This yields a clean mathematical description of the sampling law for $i$, and serves as a useful benchmark for understanding how the actual algorithm deviates from ideal behavior in practice.

Under exact integration, the orbit is constructed by repeatedly applying the exact Hamiltonian flow map $\varphi_h$ rather than the leapfrog integrator. The orbit construction follows the same recursive doubling procedure as before, but without any accept/reject steps. At each doubling step, $2^k$ new states are generated by applying $\varphi_h$, either forward or backward in time with equal probability. After $m$ such steps, the final orbit consists of $2^m$ consecutive states, each lying exactly on the Hamiltonian trajectory through the initial point.

The integration time index $i$ is then drawn uniformly from the extension added during the final doubling step.  Although this rule is simple (just uniform sampling from the last added segment) the resulting distribution of $i$ is nonuniform, since the location of the final extension depends on the random sequence of doubling directions. We now briefly review this construction and describe its implications.

Let $m \ge 1$ be the number of doublings. Begin with the singleton orbit $\calO_0 = ((\theta, \rho))$ and initial indices $a_0 = b_0 = 0$. For each doubling step $k \in \{0, \dots, m - 1\}$, extend the current orbit by appending $2^k$ consecutive states generated by repeatedly applying the exact Hamiltonian flow $\varphi_h$ to either the left or right end, each with probability $1/2$. After $m$ doublings, the resulting orbit contains $2^m$ phase space states indexed by consecutive integers that include $0$.

Let $\calO_{m-1}^{\ext}$ denote the extension added during the final doubling step ($k = m - 1$), with index range $a_k^{\ext}{:}b_k^{\ext}$ and length $2^{m-1}$. The final integration time index is then sampled uniformly from this index range:
\[
i \;\sim\; \operatorname{Unif}\bigl(a_k^{\ext}{:} b_k^{\ext}\bigr) \;,
\] giving each index in the last-added extension equal weight. The next results describe the distribution of $i$, as illustrated in Figure~\ref{fig:tri}, and provide a closed-form expression for $\mathbb{E}[\,|i|\,]$.

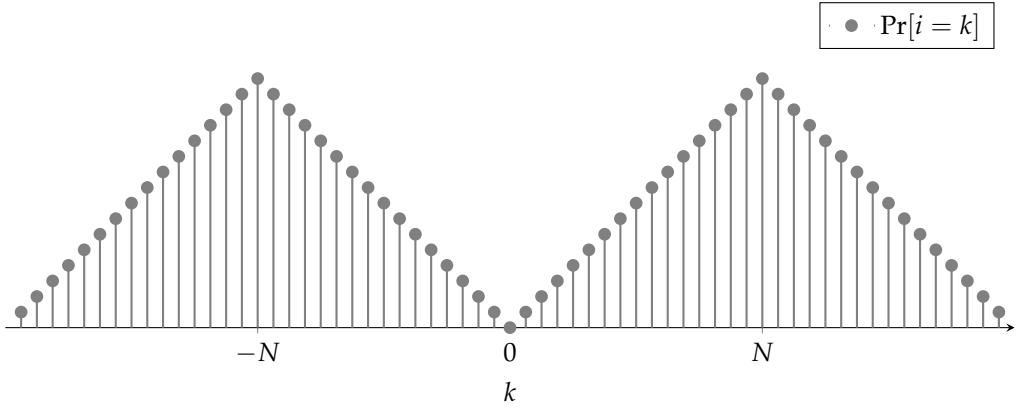
\begin{figure}[t]
\centering
\begin{tikzpicture}
        \def\k{5} 
     
    \begin{axis}[
        width=15cm, height=6cm, 
        xlabel={$k$},
        ylabel={$\Pr[ i' = k ]$ (PMF)},
        grid=none,
        xtick = {-2^(\k-1), 0, 2^(\k-1)},
        xticklabels={$-N$, $0$, $N$},
        xtick distance=1,
        ytick distance=3,
        ymin=0,ymax=1/(1.5*2^(\k-1)),
        xmin=-(2^(\k)-1)-1, xmax=2^(\k)-1+1,
        yticklabels={},
        axis x line=bottom,
        axis y line=none,
        cycle list name=color list]
        
        \addplot [
            thick, 
            gray, 
            domain=-(2*2^(\k-1)-1):(2*2^(\k-1)-1), 
            samples=2*(2^\k-1)+1, 
            ycomb, 
            mark=*] 
        {min(abs(x)/(2*(2^(\k-1))^2), (2*2^(\k-1) - abs(x))/(2*(2^(\k-1))^2))};
        
        \addlegendentry{$\Pr[ i = k ]$}
    \end{axis}
\end{tikzpicture}
\caption{\textit{Probability mass function of the integration time index $i$ in \emph{exact} biased progressive HMC with $m=5$. The triangular shape reflects the random location of the final extension orbit from which $i$ is uniformly sampled.}} 
\label{fig:tri}
\end{figure}

\medskip

\begin{theorem} \label{thm:law_of_iprime}
Fix an integer $m \ge 1$ and set $N=2^{m-1}$.
Let $i$ denote the integration time index selected by the exact biased progressive sampling procedure with uniform weights. Then $i$ follows the symmetric discrete triangular law:
\begin{equation} \label{eq:pmf_of_ip}
\Pr\bigl[i=k\bigr] \ = \ 
\frac{1}{2N^{2}} \, \min( |k|, 2N-|k|) \; \mathbf{1}\{1 \le |k| \le 2 N-1 \}
\end{equation}
for $k \in (1-2N){:}(2N-1)$.
\end{theorem}

\begin{proof}
At each doubling step ($k=0,\dots,m-1$), the Bernoulli variable $B_{k+1}$ indicates whether the  extension orbit $\calO_{k}^{\ext}$ of length $2^{k}$ is appended to the right ($B_{k+1}=1$) or to the left ($B_{k+1}=0$) of the current orbit $\calO_{k}$. In the final doubling step, the orbit $\calO_{m-1}$ has length  $N=2^{m-1}$ with right endpoint index: \[
 b_{m-1} = \sum_{j=1}^{m-1} B_{j} \; 2^{j-1}  \;.
\] which implies $b_{m-1} \sim \Unif(0{:}(N-1))$.  The corresponding left endpoint index is $a_{m-1} = b_{m-1} - N + 1$.

The final bit $B_m$ determines whether the last extension of length $N$ is appended to the right or left:
  \begin{itemize}
    \item If $B_m=1$, then $a_{m-1}^{\ext} = b_{m-1}+1$ and $b_{m-1}^{\ext} = b_{m-1}+N$. 
    \item If $B_m=0$, then $a_{m-1}^{\ext} = a_{m-1}-N$ and $b_{m-1}^{\ext} = a_{m-1} -1$.
\end{itemize}
The  index $i$ is then sampled uniformly from this extension. Let $c \sim \text{Uniform}(1{:}N)$, independent of all previous random variables. Then \begin{align*}
i &= (b_{m-1}+c) B_m + (a_{m-1}-c) (1-B_m) ,  \\
&= (b_{m-1} + c) \cdot B_m + (b_{m-1} - N + 1 - c) \cdot (1 - B_m),
\end{align*} where the second line uses the identity $a_{m-1} = b_{m-1} - N + 1$. This expression shows that $i$ is distributed as the convolution of two independent discrete uniform variables, yielding the symmetric triangular PMF given in the theorem.
\end{proof}

The following corollary computes the expected magnitude of the integration time index under this distribution, providing a simple benchmark for the typical integration time in this idealized setting.

\medskip

\begin{corollary} \label{cor:stats_of_iprime}
It holds that $\mathbb{E}[\,|i|\,]=2^{m-1}$.
\end{corollary}

\begin{proof}
    From Theorem~\ref{thm:law_of_iprime}, the magnitude $|i|$ has the triangular PMF: \[
    \Pr[|i| = k] \ = \  \frac{1}{N^2} \, \min(k, 2 N - k)
    \] for $k \in 1{:}(2N-1)$.  By expanding the expectation using this PMF, and then re-indexing, \begin{align*}
    E[|i|] &= \sum_{k=1}^{2N-1}  k \Pr[|i|=k] = \frac{1}{N^2} \sum_{k=1}^{2N-1}  k \min(k,2 N-k) \\
    &= \frac{1}{N^2} \left( \sum_{k=1}^N k^2 + \sum_{k=N+1}^{2N-1} k (2 N-k) \right)  \\
    &= \frac{1}{N^2} \left( \sum_{k=1}^N k^2 + \sum_{k=1}^{N-1} (k+N) (N-k) \right)  = N
    \end{align*} as required.
\end{proof}

Next, we describe how the orbit length can be locally adapted without breaking reversibility using a geometric stopping rule based on a U-turn condition.

\subsection{Locally Adaptive Path Length with The No-U-Turn Sampler}

\label{sec:short_overview_nuts}

The key innovation of the No-U-Turn Sampler (NUTS) is to locally adapt the orbit length, and thus integration time,  by terminating the doubling procedure in biased progressive HMC when a U-turn condition is met (see Figure~\ref{fig:uturn-condition}).  

\medskip

\begin{definition} \label{defn:u-turn}
An orbit $\calO\ =\ ((\theta_a, \rho_a), \dots, (\theta_b,\rho_b))$ satisfies the \emph{U-turn condition} if: 
\begin{equation}\label{eq:u-turn}
    \min\bigr(\rho_a \cdot M^{-1} (\theta_b-\theta_a),\,\rho_b \cdot M^{-1} (\theta_b-\theta_a)\bigr)\ <\ 0\;.
\end{equation} 
\end{definition}

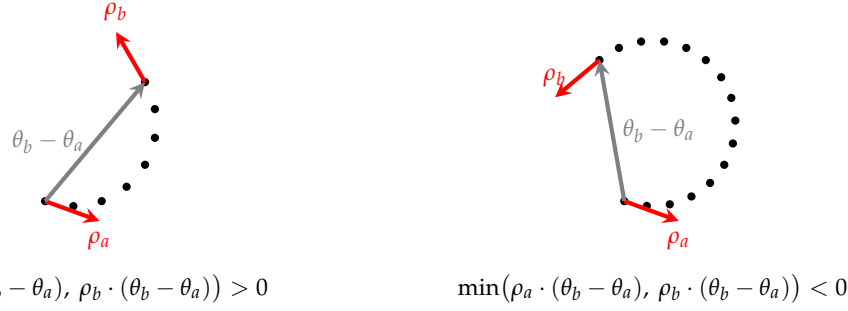
\begin{figure}[t]
\centering
\begin{tikzpicture}[scale=0.55, >=stealth]
    \def\r{2.}
    \def\vxscale{1.4}
    \def\sep{14} 
    \def\startangle{-110}
    \def\angleA{140}  
    \def\angleB{240}  
    \def\nptsA{7}
    \def\nptsB{15}

    \foreach \i in {0,...,7} {
        \pgfmathsetmacro\theta{\startangle + \i*(\angleA/\nptsA)}
        \node[circle, fill=black, inner sep=1.1pt] at ({\r*cos(\theta)}, {\r*sin(\theta)}) {};
    }

    \pgfmathsetmacro\ta{\startangle}
    \pgfmathsetmacro\tb{\startangle + \angleA}
    \coordinate (Ta) at ({\r*cos(\ta)}, {\r*sin(\ta)});
    \coordinate (Tb) at ({\r*cos(\tb)}, {\r*sin(\tb)});
    \coordinate (Ra) at ($(Ta)+({-\vxscale*sin(\ta)}, {\vxscale*cos(\ta)})$);
    \coordinate (Rb) at ($(Tb)+({-\vxscale*sin(\tb)}, {\vxscale*cos(\tb)})$);

    \draw[->, ultra thick, red] (Ta) -- (Ra) node[below] {\small$\rho_a$};
    \draw[->, ultra thick, red] (Tb) -- (Rb) node[above] {\small$\rho_b$};
    \draw[->, ultra thick, gray] (Ta) -- (Tb) node[midway, left] {\small$\theta_b - \theta_a$};

    \node at (0, -4) {\small$\min\bigl(\rho_a \cdot (\theta_b - \theta_a),\ \rho_b \cdot (\theta_b - \theta_a)\bigr) > 0$};

    \foreach \i in {0,...,15} {
        \pgfmathsetmacro\theta{\startangle + \i*(\angleB/\nptsB)}
        \node[circle, fill=black, inner sep=1.1pt] at ({\sep + \r*cos(\theta)}, {\r*sin(\theta)}) {};
    }

    \pgfmathsetmacro\taB{\startangle}
    \pgfmathsetmacro\tbB{\startangle + \angleB}
    \coordinate (TaB) at ({\sep + \r*cos(\taB)}, {\r*sin(\taB)});
    \coordinate (TbB) at ({\sep + \r*cos(\tbB)}, {\r*sin(\tbB)});
    \coordinate (RaB) at ($(TaB)+({-\vxscale*sin(\taB)}, {\vxscale*cos(\taB)})$);
    \coordinate (RbB) at ($(TbB)+({-\vxscale*sin(\tbB)}, {\vxscale*cos(\tbB)})$);

    \draw[->, ultra thick, red] (TaB) -- (RaB) node[below] {\small$\rho_a$};
    \draw[->, ultra thick, red] (TbB) -- (RbB) node[above] {\small$\rho_b$};
    \draw[->, ultra thick, gray] (TaB) -- (TbB) node[midway, right] {\small$\theta_b - \theta_a$};

    \node at (\sep, -4) {\small$\min\bigl(\rho_a \cdot (\theta_b - \theta_a),\ \rho_b \cdot (\theta_b - \theta_a)\bigr) < 0$};
\end{tikzpicture}

\caption{Position-space projections of two orbits: one that does not satisfy the U-turn condition (left), and one that does (right). For each case, the velocity vectors $\rho_a$ and $\rho_b$, as well as the displacement vector $\theta_b - \theta_a$, are shown, assuming the mass matrix $M = I$.}
\label{fig:uturn-condition}
\end{figure}

To ensure reversibility, NUTS also checks a \emph{sub-U-turn condition}.  Because NUTS constructs orbits through iterative doubling, the orbit length is always a power of two: $|\calO| = 2^m$ for some $m \in \mathbb{N}$.   For such orbits, we define a hierarchy of sub-orbits generated by recursive halving:
\begin{equation}\label{eq:sub-orbits}
 \Bigr\{ \calO_{i,j} \ :\ i \in \bigr\{ 0, 1, \dots, \log_2 |\mathcal O| \bigr\}, \, j \in \bigr\{ 1, 2, \dots, 2^i \bigr\}  \Bigr\}
\end{equation}
where, for each $i \in \{ 0, 1, \dots, \log_2 |\mathcal O| \}$, $\mathcal O_{i,j}$ are defined to be the unique orbits of size $|\mathcal O| 2^{-i}$ such that
\[ \mathcal O\ =\ \mathcal O_{i,1}\odot\mathcal O_{i,2}\odot\cdots\odot\mathcal O_{i,2^i} \;. \]

\begin{definition}
An orbit $\mathcal{O}$ is said to satisfy the \emph{sub-U-turn condition} if at least one of its sub-orbits, as defined in~\eqref{eq:sub-orbits}, satisfies the U-turn condition given in Definition~\ref{defn:u-turn}.
\end{definition}

Sub-U-turn checks play a key role in ensuring reversibility. When the orbit construction procedure is rerun starting from the $i$-th point along the final orbit, the sequence of doubling steps required to recover the same final orbit is uniquely determined by the binary representation of $b - i$.  However, if sub-U-turn checks are not enforced, the algorithm may terminate too early during this rerun, producing a different final orbit. This breaks the symmetry of the orbit construction with respect to the starting point along the orbit.

Given an initial state $(\theta, \rho) \in \mathbb{R}^{2d}$ and a maximum number of doublings $m_{\max} \in \mathbb{N}$, NUTS constructs an adaptive-length orbit as follows.

\begin{enumerate}
    \item \textbf{Initialization.}  
    Start with the singleton orbit $\mathcal{O}_0 = ((\theta, \rho))$, and set the index range $a_0 = b_0 = 0$.  
    Sample $m_{\max}$ i.i.d. Bernoulli$(1/2)$ random variables $B = (B_1, \dots, B_{m_{\max}})$.

    \item \textbf{Doubling steps.}   For each step $k \in 0{:}(m_{\max}-1)$,  perform the following:

\noindent
    \begin{enumerate}

\item \textbf{Current orbit.} Let the current orbit be \[
\calO_k  = \left((\theta_{a_k},\rho_{a_k}),\dots,(\theta_{b_k},\rho_{b_k})\right)
\] with length $|\calO_k| = 2^k$ and index range $a_k{:}b_k$.
\item \textbf{Extension.}  Construct an orbit $\calO_k^{\ext}$ of length $2^k$ over the index range $a_k^{\ext} {:} b_k^{\ext}$, defined by the value of $B_{k+1}$:
\[
a_k^{\ext} = 
\begin{cases}
b_k + 1 & \text{if } B_{k+1} = 1, \\
a_k - 2^k & \text{if } B_{k+1} = 0,
\end{cases}
\quad
b_k^{\ext} = 
\begin{cases}
b_k + 2^k & \text{if } B_{k+1} = 1, \\
a_k - 1 & \text{if } B_{k+1} = 0.
\end{cases}
\]
        \item \textbf{Sub-U-turn check.}  
 If $\calO_k^{\ext}$ satisfies the sub-U-turn condition, terminate and set $\calO = \calO_k$.
        \item \textbf{Orbit update.}  
 If not, define the updated orbit: $\calO_{k+1}= \begin{cases}  \calO_m\odot \calO_k^{\ext} & \text{if $B_{k+1}=1$,} \\ 
\calO_k^{\ext} \odot \calO_k  & \text{if $B_{k+1}=0$.} 
\end{cases}$
        \item \textbf{U-turn check.}  
 If $\calO_{k+1}$ satisfies the U-turn condition or if $k+1 = m_{\max}$, terminate and set $\calO = \calO_{k+1}$.
\end{enumerate}
   \item \textbf{Output.}  
    The final orbit $\mathcal{O}$ has length $2^m$ for some $m \le m_{\max}$ and is fully determined by its length and rightmost index.
\end{enumerate}

The integration time index is then selected exactly as in biased progressive HMC.  For clarity, the integration time index selection has been omitted in the above description, as it follows identically from the biased progressive HMC case described in Section~\ref{sec:short_overview_biased_progressive_hmc}.

\section{WALNUTS} 

\label{sec:walnuts}

WALNUTS addresses a key limitation in traditional NUTS: fixed step size integrators can fail to capture local variations in geometry, especially in stiff or funnel-shaped targets. By introducing a locally adaptive variable step size leapfrog integrator within a fixed macro time grid, WALNUTS improves robustness while maintaining reversibility.  It does so using the auxiliary variable GIST framework introduced in Section~\ref{sec:gist}, which provides a principled foundation for both path length and step size adaptation.

At a high level, each transition step of WALNUTS uses a locally adaptive leapfrog integrator with variable micro step sizes to construct an orbit $\calO$ along a macro time grid.  Once the orbit is built, WALNUTS selects the next state by sampling a point in $\calO$, thereby randomizing the integration time.  The general framework for orbit construction and state selection was introduced in Section~\ref{sec:preliminaries}; what follows highlights the modifications needed to incorporate a variable step size leapfrog integrator.

\subsection{Variable Step Size Leapfrog Integrator}
\label{sec:variable_step_int}

The leapfrog (or Verlet) integrator is the standard method used to numerically approximate Hamiltonian dynamics in HMC \cite{BoSaActaN2018, HaLuWa2010, LeRe2004}. In this section, we develop a locally adaptive extension that varies the step size within each macro time interval based on local energy error.  

\paragraph{Macro/Micro grids and integration notation}

In our framework, we assume a uniform \emph{macro} time grid to anchor the overall trajectory, while allowing for variable \emph{micro} step sizes within each macro interval to locally control integration errors. Each macro interval is subdivided into a number of micro steps, with finer subdivisions used in regions of high curvature or energy variation. This two-level structure enables local adaptivity while preserving a globally consistent time grid, as illustrated in Figure~\ref{fig:stencil}. We restrict attention to uniform macro grids throughout; the advantages of this choice for reversibility are discussed in Remark~\ref{rmk:fixed_macro}.

To formalize this, we introduce a fixed macro step size $h >0 $ and define the evenly-spaced time grid \[
 t_k \ := \ k h~~\text{where}~~k \in \mathbb{Z} \;.
 \]  Each macro interval $[t_k, t_{k+1}]$ is subdivided into $\ell_{k,k+1}$ micro steps, each of size $h \ell_{k,k+1}^{-1}$. We denote by $\{\ell_{k,k+1} \}_{k \in \mathbb{Z}}$ the sequence of micro step counts associated with each macro interval.  For any $m, n \in \mathbb{Z}$, the evolution of the numerical solution from time $t_n$ to $t_m$ using these variable step sizes is defined by the following composition of leapfrog maps: \begin{equation} \label{eq:Phimn}
\Phi_{m,n} = \begin{cases} \Phi_{ h \ell_{m-1,m}^{-1}}^{\ell_{m-1,m}} \circ  \dots \circ \Phi_{ h \ell_{n,n+1}^{-1} }^{\ell_{n,n+1}}  & \text{if $m > n$} \\
\mathcal{F} \circ \Phi_{ h \ell_{m,m+1}^{-1}}^{\ell_{m,m+1}} \circ  \dots \circ \Phi_{h \ell_{n-1,n}^{-1} }^{\ell_{n-1,n}} \circ \mathcal{F} & \text{if $m<n$}
\end{cases}
\end{equation} 
with $\Phi_{m,m}$ defined as the identity map and $\mathcal{F}$ denoting the momentum flip.   The forward evolution ($m > n$) advances the numerical solution from $t_n$ to $t_m$ using a sequence of forward leapfrog steps, while the backward evolution ($m<n$) requires a momentum flip before and after in order to advance the numerical solution using backward leapfrog steps. 

If the micro step counts $\{ \ell_{k,k+1} \}_{k \in \mathbb{Z}}$ are fixed in advance or treated as auxiliary variables in an enlarged state space, then $\Phi_{m,n}$ is volume-preserving, as it is a composition of volume-preserving leapfrog maps \cite[Proposition 2.3]{BoSaActaN2018}. However, if $\{ \ell_{k,k+1} \}_{k \in \mathbb{Z}}$ depends on the state  during integration, then volume preservation is no longer guaranteed. In either case, $\Phi_{m,n}$ is not, in general, time-reversible: although each leapfrog step is time-reversible, the composition of time-reversible maps with varying step sizes does not, in general, yield a time-reversible map \cite[Theorem 4.2]{BoSaActaN2018}.

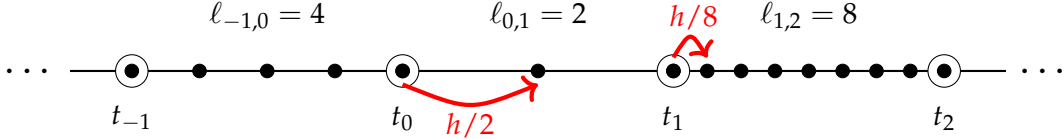
\begin{figure}[t]
    \begin{center}
    \begin{tikzpicture}[scale=1.8]
   \draw[-, thick](1,0.0) -- (9,0.0);
   \node[black,scale=1.1] at (2,-0.35) {$t_{-1}$};
   \node[black,scale=1.1] at (4,-0.35) {$t_{0}$};
   \node[black,scale=1.1] at (6,-0.35) {$t_{1}$};
   \node[black,scale=1.1] at (8,-0.35) {$t_{2}$};
   \node[black,scale=1.1] at (3,0.4) {$\ell_{-1,0}=4$};
   \node[black,scale=1.1] at (5,0.4) {$\ell_{0,1}=2$};
   \node[black,scale=1.1] at (7,0.4) {$\ell_{1,2}=8$};
   \node[black, scale=1.5,fill=white] at (8.75,0.0) {$\dotsc$};
   \node[black, scale=1.5,fill=white] at (1.25,0.0) {$\dotsc$};
   \foreach \x in {2.,4.,6.,8.} 
   {
   \filldraw[color=black,fill=white] (\x,0) circle (0.12); 
   };
   \foreach \n in {2,2.5,3,3.5,4}
   {
   \filldraw[color=black,fill=black] (\n,0) circle (0.05); 
   }
   \foreach \n in {4,5,6}
   {
   \filldraw[color=black,fill=black] (\n,0) circle (0.05); 
   }
   \foreach \n in {6,6.25,6.5,6.75,7,7.25,7.5,7.75,8}
   {
   \filldraw[color=black,fill=black] (\n,0) circle (0.05); 
   }
   \draw[->, red, ultra thick] (4,-0.1) .. controls (4.5, -0.3) ..  (5, -0.1);
   \node[red,scale=1.1] at (4.5, -0.4) {$h/2$};
   \draw[->, red, ultra thick] (6,0.1) .. controls (6.125, 0.3) ..  (6.25, 0.1);
   \node[red,scale=1.1] at (6.15, 0.4) {$h/8$};
   \end{tikzpicture}
   \end{center}
   \caption{The time grid used to advance the numerical solution from  $t_0$ to $t_1$ consists of two microsteps, each of size $h/2$. Similarly, the transition from $t_1$ to $t_2$ is achieved using eight microsteps, each of size $h/8$.}
   \label{fig:stencil}
\end{figure}

\paragraph{Micro step selection using local energy error}

In locally adaptive HMC, the primary goal of local step size adaptivity is to select the largest leapfrog step size that ensures the energy error remains below a specified threshold. Here, we describe a general procedure for determining such step sizes.  Given an initial condition $(\theta, \rho) \in \mathbb{R}^{2d}$, a macro step size $h$, and a prescribed maximum allowable energy error $\delta > 0$, we define 
\begin{equation} \label{eq:micro}
    \micro(\theta,\rho, \mu, M, h, \delta) \ := \ \min\{ \ell \in2^{\bbN}:\,  H^+_\ell - H^-_\ell \leq \delta \}
\end{equation}
where $2^{\mathbb{N}} := \{ 2^n : n \in \mathbb{N} \}$. Here, $H^+_\ell$ and $H^-_\ell$ denote the maximum and minimum Hamiltonian values encountered along the leapfrog trajectory, starting from $(\theta, \rho)$ and integrated forward over $[0, h]$ with  micro step size $h \ell^{-1}$. Specifically,
\begin{align*}
    H^+_{\ell} = \max \{ H(\theta^{(j)}, \rho^{(j)}),\, 0\leq j\leq \ell \},\quad
    H^-_{\ell} = \min \{ H(\theta^{(j)}, \rho^{(j)}),\, 0\leq j\leq \ell \}.
\end{align*}
The leapfrog iterates $(\theta_j,\rho_j)$ are obtained by applying $j$ forward leapfrog steps with step size $h \ell^{-1}$, i.e., 
\[ (\theta^{(j)}, \rho^{(j)}) = \Phi_{h \ell^{-1}}^j(\theta, \rho) \quad 0\leq j\leq \ell  \;.
\]
In other words, $\micro(\theta,\rho, \mu, M, h, \delta)$ returns the minimal power of two such that the leapfrog trajectory stays within energy error $\delta$.

\smallskip

\begin{remark}
In the optimized implementation of WALNUTS, the energy error used to assess the validity of a micro trajectory is computed as the absolute difference of the Hamiltonian evaluated at the endpoints, i.e., $|H(\theta^{(0)}, \rho^{(0)}) - H(\theta^{(\ell)}, \rho^{(\ell)})|$. This choice improves practical performance by reducing mismatches between forward and backward micro steps, and ensures that the energy error threshold $\delta$ has a consistent interpretation across different values of the step size reduction factor~$\ell$. While this simplification deviates slightly from the theoretical formulation, the core theoretical guarantees are expected to remain intact.
\end{remark}

The next lemma shows that the micro step size factor $\ell = \micro(\theta,\rho,\mu, M, h, \delta)$ is always greater than or equal to the value obtained by applying the same procedure at the endpoint of the corresponding forward trajectory, but with the momentum reversed.

\medskip

\begin{lemma} \label{lemma:forward_backward}
For any $(\theta, \rho) \in \mathbb{R}^{2d}$, and energy error threshold $\delta >0$,  the forward step size reduction factor satisfies:
\[
\ell = \micro(\theta,\rho,\mu, M, h, \delta)  \ge \micro \left(  \theta', -\rho', \mu, M, h, \delta  \right), \quad \text{where} \quad  ( \theta', \rho') = \Phi_{h \ell^{-1}}^{\ell}(\theta, \rho) \;.
\] 
\end{lemma}

\begin{proof}
Let $\ell = \micro(\theta,\rho, \delta, h)$.  By definition of micro in \eqref{eq:micro}, we have:  \begin{equation}
\label{eq:energy_error_forward}
\max \{ H(\theta^{(j)}, \rho^{(j)}) \,:\,  0\leq j\leq \ell \} -  \min \{ H(\theta^{(j)}, \rho^{(j)}) \,:\, 0\leq j\leq \ell \} \le \delta
\end{equation}  where the leapfrog iterates are given by \[ 
(\theta^{(j)},\rho^{(j)})=  \Phi_{h \ell^{-1}}^j  (\theta, \rho), \quad 0\leq j\leq \ell \;.
\]  
\begin{minipage}{\textwidth}
\begin{minipage}{0.5\textwidth}
For the backward evolution, we consider the trajectory starting from $(\theta', \rho')$. As shown in the grey dots in the image on the right, by applying a momentum flip $\mathcal{F}$, the backward orbit is given by: 
 \[
 \Phi_{h \ell^{-1}}^j \circ \mathcal{F}  (\theta', \rho') = (\theta^{(\ell-j)},-\rho^{(\ell-j)}) \;, \quad 0\leq j\leq \ell \;.
\] 
\end{minipage}
\begin{minipage}{0.5\textwidth}
\centering
\begin{tikzpicture}[scale=1.5]

    \draw[ultra thick,red] plot [smooth, tension=0.75] coordinates 
        { (0,0.4) (0.6,0.7) (1.2,0.85) (1.8,0.7) (2,0.4) };

    \foreach \x/\y in {0/0.4, 0.6/0.7, 1.2/0.85, 1.8/0.7, 2/0.4} {
        \filldraw[red] (\x,\y) circle (1.75pt);
    }

    \draw[ultra thick,gray] plot [smooth, tension=0.75] coordinates 
        { (2,-0.4) (1.8,-0.7) (1.2,-0.85) (0.6,-0.7) (0,-0.4) };

    \foreach \x/\y in {1.8/-0.7, 1.2/-0.85, 0.6/-0.7} {
        \filldraw[gray] (\x,\y) circle (1.75pt);
    }

    \node[red] at (-0.3,0.5) { $(\theta, \rho)$};
    \node[gray] at (-0.4,-0.5) { $(\theta, -\rho)$};
    \node[red] at (2.4,0.5) { $(\theta', \rho')$};
    \node[gray] at (2.5,-0.5) { $(\theta', -\rho')$};

    \draw[{Latex}-, black, dashed,  ultra thick, bend right=0] (0,0.35) to (0,-0.35);
    \draw[-{Latex}, black, dashed,  ultra thick, bend left=0] (2,0.35) to (2,-0.35);

    \node[black] at (-0.2,0) {\small $\mathcal{F}$};
    \node[black] at (2.2,0) {\small $\mathcal{F}$};

    \filldraw[red] (0,0.4) circle (1.75pt);
    \filldraw[gray] (0,-0.4) circle (1.75pt);
    \filldraw[red] (2,0.4) circle (1.75pt);
    \filldraw[gray] (2,-0.4) circle (1.75pt);

\end{tikzpicture}
\end{minipage}
\end{minipage}

\medskip

Since $H\circ \mathcal{F} \equiv H$, it follows that the backward trajectory at step size $h \ell^{-1}$ also satisfies \eqref{eq:energy_error_forward}. Thus, the backward critical step size reduction factor cannot exceed $\ell$, ensuring that $\ell \ge \micro\left( \theta', \textcolor{blue}{-}\rho', \delta, h \right)$.
\end{proof}

\paragraph{Randomized step size selection within orbits}
On the same macro step, Lemma~\ref{lemma:forward_backward} suggests that the forward and backward  step size reduction factors may differ.   To address this issue, we introduce randomization in the step size selection process.

Let $p_{\micro}$ be the distribution from which the step size reduction factor is sampled given the critical step size reduction factor.
Given an initial point $(\theta, \rho) $ at time $t_0$, we recursively define the sequence of points $ (\theta_j, \rho_j) $ at times $ t_j $ as follows:
\begin{itemize}
    \item \textbf{Base case ($ j = 0 $)}:  
    \[
    (\theta_0, \rho_0) = (\theta, \rho).
    \]
    \item \textbf{Forward recursion ($ j > 0 $)}:  
    \[
    (\theta_j, \rho_j) = \Phi_{h \ell_{j-1,j}^{-1}}^{\ell_{j-1,j}}(\theta_{j-1}, \rho_{j-1}),
    \]
    where the step size reduction factor \( \ell_{j-1,j} \) is sampled from the conditional distribution:
    \[
    \ell_{j-1,j} \sim p_{\micro}(\cdot \mid \micro(\theta_{j-1},\rho_{j-1},\mu, M, h, \delta)).
    \]
    \item \textbf{Backward recursion ($ j < 0 $)}:  
    \[
    (\theta_j, \rho_j) = \mathcal{F} \circ \Phi_{h \ell_{j,j+1}^{-1}}^{\ell_{j,j+1}} \circ \mathcal{F} (\theta_{j+1}, \rho_{j+1}),
    \]
    where \( \ell_{j,j+1} \) is sampled from the analogous distribution:
    \[
    \ell_{j,j+1} \sim p_{\micro}(\cdot \mid \micro(\theta_{j+1},-\rho_{j+1},\mu, M, h, \delta)).
    \]
\end{itemize}

A simple choice for the distribution $p_{\micro}$ is a uniform distribution over two candidate micro step size factors: the deterministic baseline value $\widetilde{\ell} = \texttt{micro}(\theta, \rho, \mu, M, h, \delta)$, as defined in \eqref{eq:micro}, and its immediate successor $\widetilde{\ell} + 1$. Specifically,
\[
p_{\micro}(k \mid \widetilde{\ell}) 
= 
\frac{1}{2} \, \mathbb{1}_{ \{ \widetilde{\ell},\ \widetilde{\ell}+1 \} } (k) \;.
\]
Note that the support of $p_{\micro}$ is limited to larger micro step size factors relative to $\widetilde{\ell}$. This restriction is justified by Lemma~\ref{lemma:forward_backward}, which guarantees that the forward micro step size factor is always greater than or equal to the backward one over the same macro step. Therefore, if a given step size yields an acceptable energy error in the forward direction, it will also do so in the backward direction. This allows us to safely omit smaller micro step factors from consideration.

\medskip

\begin{remark}
The recently proposed step size adaptive NUTS algorithm \cite{BouRabeeCarpenterKleppeMarsden2024} uses the same random number of leapfrog micro steps within each macro step. Specifically, for each macro step spanning the index range $a{:}b$, the number of microsteps is given by a random variable $\ell$ such that $ \ell_{a,a+1} = \ell_{a+1,a+2} = \dots = \ell_{b-1,b} = \ell $ with $ \ell $ drawn from a distribution that depends on $(\theta, \rho)$,  $a$, and $b$. 
\end{remark}

\subsection{Orbit Construction with Variable Step Size Leapfrog Integration}

\label{sec:walnuts-orbit}

Orbit construction proceeds exactly as in NUTS, except that fixed step size leapfrog integration is replaced by variable step size leapfrog integration during the orbit extension steps described in Section~\ref{sec:short_overview_nuts}.  Specifically, by \eqref{eq:Phimn}, the overall evolution from time $t_0 = 0$ to time $t_j$ is given by
\begin{equation} \label{eq:Phi_j0_bfell}
\left. \Phi_{j,0} \right|_{\bfell} = \begin{cases} \Phi_{ h \ell_{j-1,j}^{-1}}^{\ell_{j-1,j}} \circ  \dots \circ \Phi_{ h \ell_{0,1}^{-1} }^{\ell_{0,1}}  & \text{if $j > 0$} \\
\mathcal{F} \circ \Phi_{ h \ell_{j,j+1}^{-1}}^{\ell_{j,j+1}} \circ  \dots \circ \Phi_{h \ell_{-1,0}^{-1} }^{\ell_{-1,0}} \circ \mathcal{F} & \text{if $j<0$}
\end{cases}
\end{equation}
where $\Phi_{0,0}$ is the identity map, and $\mathcal{F}$ denotes the momentum flip operator. Each composition uses a different number of micro steps, governed by the micro step size reduction factors $\ell_{k,k+1}$ stored in the tuple $\bfell$.  Given two endpoint indices $a, b \in \mathbb{Z}$ with $a \le b$, the orbit over the index range $a{:}b = \{a, a+1, \dots, b\}$, starting at $(\theta, \rho)$ at time $t_0 = 0$, is defined by
\[
\mathcal{O} = (\left. \Phi_{a,0} \right|_{\bfell}(\theta,\rho), \dots, \left. \Phi_{b,0} \right|_{\bfell}(\theta,\rho) ) = ((\theta_a, \rho_a), \dots, (\theta_b, \rho_b)).
\]
The figure below illustrates this construction. Each macro-step interval $[t_k, t_{k+1}]$ is subdivided into $\ell_{k,k+1}$ micro steps. The number of micro steps may vary across intervals, as shown by the differing densities of black points in the two macro-steps depicted. These locally adaptive refinements allow WALNUTS to dynamically adjust the resolution of leapfrog integration within each macro step:

\begin{center}
    \begin{tikzpicture}[scale=2.25]
    \draw[-, thick] (2,0.0) -- (4.5,0.0);  
    \draw[-, thick] (5.5,0.0) -- (8,0.0);  

    \node[black,scale=1.25] at (2,-0.35) {$t_a$};  
    \node[black,scale=1.25] at (4,-0.35) {$t_{a+1}$}; 
    \node[black,scale=1.25] at (6,-0.35) {$t_{b-1}$}; 
    \node[black,scale=1.25] at (8,-0.35) {$t_b$};  

    \node[black, scale=1.5] at (5,0.0) {$\cdots$};  

    \node[black,scale=1.1] at (3,0.4) {$\ell_{a,a+1}$ micro steps};  
    \node[black,scale=1.1] at (7,0.4) {$\ell_{b-1,b}$ micro steps};  

    \foreach \x in {2,4,6,8} 
    {
    \filldraw[color=black,fill=white] (\x,0) circle (0.12); 
    };

    \foreach \n in {2.0,2.5,3,3.5,4.0}
    {
    \filldraw[color=black,fill=black] (\n,0) circle (0.05); 
    }

    \foreach \n in {6,6.25,6.5,6.75,7,7.25,7.5,7.75,8}
    {
    \filldraw[color=black,fill=black] (\n,0) circle (0.05); 
    }

    \draw[<-, red,  ultra thick] (3,-0.1) .. controls (3.25, -0.3) ..  (3.5, -0.1);
    \node[red,scale=1.] at (3.25, -0.4) {$h \ell_{a,a+1}^{-1}$};
    
    \draw[->, red, ultra thick] (7,-0.1) .. controls (7.125, -0.3) ..  (7.25, -0.1);
    \node[red,scale=1.] at (7.15, -0.4) {$h \ell_{b-1,b}^{-1}$};
    \end{tikzpicture}

\end{center}

\subsection{Randomized Integration Time via Biased Progressive Sampling}

\label{sec:walnuts-int-time}

In WALNUTS, the index of the next state in the Markov chain is selected via biased progressive sampling, as in biased progressive HMC. This index corresponds to a randomly chosen point along the orbit and may be interpreted as a randomized integration time.

We use $\cat(\mathbb{S},q)$ to denote the categorical distribution supported on a discrete set $\mathbb{S}$ with the probability of selecting $x \in \mathbb{S}$ proportional to $q(x)$.   In view of Lemma~\ref{lemma:forward_backward}, \[
p_{\micro}( \ell_{k,k+1} \mid \micro(\theta_k, \rho_k, \mu, M, h, \delta) )
\] may not equal \[
p_{\micro}( \ell_{k,k+1} \mid \micro(\theta_{k+1}, -\rho_{k+1}, \mu, M, h, \delta) ) \;.
\]  To ensure reversibility, we assign a weight to each point in the orbit, defined recursively as follows.   
\begin{itemize}
\item 
The initial weight is \begin{equation} \label{eq:w0}
w_0 = \mu(\theta_0) e^{-\frac{1}{2} \rho_0\tran M^{-1} \rho_0} \;.
\end{equation}
\item
For $j>0$, 
\begin{equation} 
w_j =  \dfrac{\mu(\theta_{j}) e^{-\frac{1}{2} (\rho_j)\tran M^{-1} \rho_j}}{\mu(\theta_{j-1}) e^{-\frac{1}{2} (\rho_{j-1})\tran M^{-1} \rho_{j-1}}}  \dfrac{p_{\micro}( \ell_{j-1,j}  \mid \micro(\theta_{j}, -\rho_{j}, \mu, M, h, \delta) )}{p_{\micro}(\ell_{j-1,j}  \mid \micro(\theta_{j-1}, \rho_{j-1}, \mu, M, h, \delta) )}  w_{j-1} \;. \label{eq:wts-forw}
\end{equation}
\item
While for $j<0$,
\begin{equation}
w_j =\dfrac{\mu(\theta_{j}) e^{-\frac{1}{2} (\rho_j)\tran M^{-1} \rho_j}}{\mu(\theta_{j+1}) e^{-\frac{1}{2} (\rho_{j+1})\tran M^{-1} \rho_{j+1}}}  \dfrac{p_{\micro}( \ell_{j,j+1}  \mid \micro(\theta_{j}, \rho_{j}, \mu, M, h, \delta) )}{p_{\micro}( \ell_{j,j+1}  \mid \micro(\theta_{j+1}, -\rho_{j+1}, \mu, M, h, \delta) )}  w_{j+1} \;.
\label{eq:wts-backw}
\end{equation}
\end{itemize}

  Given an orbit with endpoints indices $a,b \in \mathbb{Z}$, we define the corresponding collection of weights by \[
\calW = (w_a, \dots, w_b) \;.
\]
Biased progressive sampling in WALNUTS selects the index of the next state iteratively as follows.  The process begins with the initial index $i_0 = 0$.  At the $k$-th doubling step ($k \in \{0, \dots, m_{\max}-1\}$), let $\calO_k$, $\calW_k$ denote the current orbit and its weights, and $\calO_{k}^{\ext}$, $\calW_k^{\ext}$ denote the proposed extension and its weights.  An extension index is sampled according to \[ i_{k}^{\ext} \sim \cat(a_k^{\ext}{:}b_k^{\ext},\calW^{\ext}_k)
\] and accepted with probability given by the acceptance function \begin{equation} \label{eq:walnuts-alpha}
\alpha(\calW_k, \calW_k^{\ext}) :=  \min\left(1, \dfrac{\sum \calW^{\ext}_{m} }
    {\sum \calW_m} \right) \;.
\end{equation}
Here, the expression $\sum \calW$ is shorthand for the sum of all the weights in the collection $\calW$, i.e., $\sum \calW = \sum_{w \in \calW} w$.  The index update is:
\begin{equation} \label{eq:bpupdate}
\begin{aligned}
  i_{k+1} = \begin{cases} i_{k}^{\ext} & \text{with probability  $\alpha(\calW_k, \calW_k^{\ext})$,}
 \\
i_k & \text{otherwise}.
 \end{cases}
\end{aligned} 
\end{equation}
This Metropolis-style rule favors integration times in the extension orbit while ensuring reversibility.

\medskip

\begin{remark} \label{rmk:fixed_macro}
An key advantage of WALNUTS is its use of a fixed macro time grid: each call to the micro function in \eqref{eq:micro} targets the same segment of the exact Hamiltonian trajectory, but with successively finer resolution. This design substantially improves the probability of reversibility compared to methods that do not operate on a fixed macro grid.  To highlight the contrast, consider an alternative approach \citep[used e.g. by][]{biron2024automala,kleppe2016adaptive}  that adaptively reduces the step size $h$ until the energy error falls below a threshold, then performs a single integration step using the refined step size, and finally checks reversibility by integrating backward with successively smaller step sizes. If the forward step involves a reduction in step size, the backward checks will evaluate points that precede (in time) the original initial condition, regions of the trajectory that are not aligned with the forward pass. This temporal misalignment increases the likelihood of reversibility failures. In WALNUTS, by contrast, all refinements are anchored to the same macro grid, ensuring that both forward and backward trajectories traverse the same time intervals.  
\end{remark}

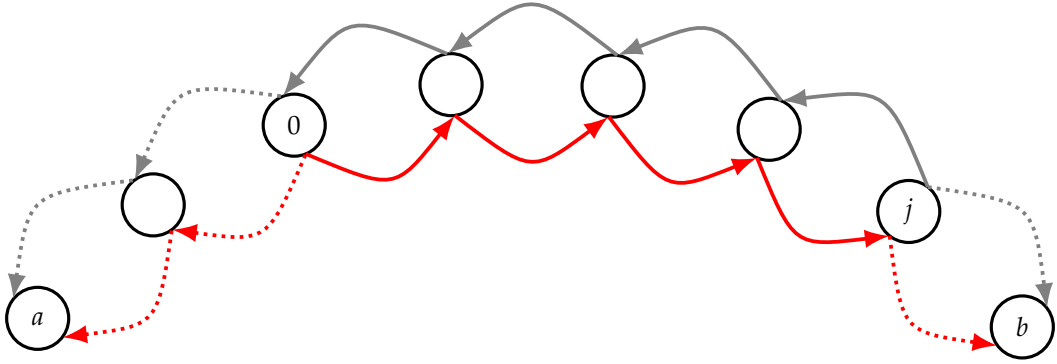
\begin{figure}[t]
    \begin{center}
    \begin{tikzpicture}[scale=0.55]
    \begin{scope}[very thick,decoration={
        markings,
        mark=at position 1.0 with {\arrow[scale=3]{stealth}}}
        ]
    \pgfmathsetmacro\A{15}
    \pgfmathsetmacro\R{0.75}
    \pgfmathsetmacro\leapHeight{2} 
    \pgfmathsetmacro\circleShift{1.5*\R} 
    
    \pgfmathsetmacro\a{-2}  
    \pgfmathsetmacro\b{5}  
    \pgfmathsetmacro\bm{4}  
    \pgfmathsetmacro\L{4}  
    \pgfmathsetmacro\Lplus{max(\L,0)}  
    \pgfmathsetmacro\Lminus{min(\L,0)}  
    \pgfmathsetmacro\phaseshift{-112}  
    \pgfmathsetmacro\om{15}  
    
    \foreach \i in {\a,...,\b} {
        \filldraw[color=black,fill=none] ({\A*cos(\om*\i+\phaseshift)}, {-\A*sin(\om*\i+\phaseshift)}) circle (\R);
    }
    \node[] at ({\A*cos(\om*\Lplus+\phaseshift)}, {-\A*sin(\om*\Lplus+\phaseshift)}) {$j$};
    \node[] at ({\A*cos(\om*\Lminus+\phaseshift)}, {-\A*sin(\om*\Lminus+\phaseshift)}) {$0$};
    \node[] at ({\A*cos(\om*\a+\phaseshift)}, {-\A*sin(\om*\a+\phaseshift)}) {$a$};
    \node[] at ({\A*cos(\om*\b+\phaseshift)}, {-\A*sin(\om*\b+\phaseshift)}) {$b$};
    
    \foreach \i in {\a,...,\bm} {
        \pgfmathtruncatemacro\n{\i}
        \pgfmathtruncatemacro\nnext{\i+1}
           
        \foreach \shift/\color/\leap/\arrowDirection in {0/gray/\leapHeight/right} {
            \pgfmathsetmacro\thetaStart{\om*\n+\phaseshift}
            \pgfmathsetmacro\thetaEnd{\om*\nnext+\phaseshift}
           
            \pgfmathsetmacro\xstart{(\A+\R)*cos(\thetaStart)}
            \pgfmathsetmacro\ystart{-(\A+\R)*sin(\thetaStart)}
           
            \pgfmathsetmacro\xend{(\A+\R)*cos(\thetaEnd)}
            \pgfmathsetmacro\yend{-(\A+\R)*sin(\thetaEnd)}
           
            \pgfmathsetmacro\xcontrol{(\xstart + \xend)*0.55}
            \pgfmathsetmacro\ycontrol{(\ystart + \yend)*0.55}
           
            \ifnum\nnext>\Lplus
                \ifnum\n<\Lminus
                    \draw[color=\color, -{Latex}, line width=0.5mm] (\xstart, \ystart) .. controls (\xcontrol, \ycontrol) .. (\xend, \yend);
                \else
                    \draw[color=\color, -{Latex}, line width=0.5mm, dotted] (\xstart, \ystart) .. controls (\xcontrol, \ycontrol) .. (\xend, \yend);
                \fi
            \else
                \ifnum\nnext>\Lminus
                    \draw[color=\color, {Latex}-, line width=0.5mm] (\xstart, \ystart) .. controls (\xcontrol, \ycontrol) .. (\xend, \yend);
                \else
                    \draw[color=\color, {Latex}-, line width=0.5mm, dotted] (\xstart, \ystart) .. controls (\xcontrol, \ycontrol) .. (\xend, \yend);
                \fi
            \fi
        }
        \foreach \shift/\color/\leap/\arrowDirection in {-1.25*\circleShift/red/-\leapHeight/left} {
            \pgfmathsetmacro\thetaStart{\om*\n+\phaseshift}
            \pgfmathsetmacro\thetaEnd{\om*\nnext+\phaseshift}
           
            \pgfmathsetmacro\xstart{(\A-\R)*cos(\thetaStart)}
            \pgfmathsetmacro\ystart{-(\A-\R)*sin(\thetaStart)}
           
            \pgfmathsetmacro\xend{(\A-\R)*cos(\thetaEnd)}
            \pgfmathsetmacro\yend{-(\A-\R)*sin(\thetaEnd)}
           
            \pgfmathsetmacro\xcontrol{(\xstart + \xend)*0.45}
            \pgfmathsetmacro\ycontrol{(\ystart + \yend)*0.45}
           
            \ifnum\nnext>\Lminus 
                \ifnum\n<\Lplus
                    \draw[color=\color, -{Latex},  line width=0.5mm] (\xstart, \ystart) .. controls (\xcontrol, \ycontrol) .. (\xend, \yend);
                \else
                    \draw[color=\color, -{Latex}, line width=0.5mm, dotted] (\xstart, \ystart) .. controls (\xcontrol, \ycontrol) .. (\xend, \yend);
                \fi
            \else 
                \ifnum\n>\Lplus+1
                    \draw[color=\color, {Latex}-, line width=0.5mm] (\xstart, \ystart) .. controls (\xcontrol, \ycontrol) .. (\xend, \yend);
                \else
                    \draw[color=\color, {Latex}-, line width=0.5mm, dotted] (\xstart, \ystart) .. controls (\xcontrol, \ycontrol) .. (\xend, \yend);
                \fi
            \fi
        }
    }
    \end{scope}
    \end{tikzpicture}
    \end{center}

    \caption{In WALNUTS, the probability of producing the same orbit starting from the $i$-th state in the orbit depends only on the transitions over the subintervals between $\min(i,0)$ and $\max(i,0)$. In the illustration, $i=4$. Red arrows indicate the directions of macro steps starting from the initial state indexed at $0$. Gray arrows represent the directions of macro steps starting from $i$ in order to reach $a$ and $b$. The dotted lines mark subintervals where the macro steps move in the same direction. On such intervals,  the conditional distribution of the step size reduction factor remains unchanged.}
    \label{fig:walnuts_ap}
\end{figure}

\subsection{Reversibility of WALNUTS}

\label{sec:reversibility}

We present a rigorous proof of the reversibility of WALNUTS by interpreting it as a GIST sampler, as described in Section~\ref{sec:gist}, and verifying that it satisfies the conditions of Theorem~\ref{thm:AVM_reversibility}.  While the overall structure of the proof is concise, establishing reversibility requires careful handling of the variable step size leapfrog integrator and its interaction with orbit construction and state selection.

To streamline the proof, we introduce some additional notation. For brevity, we write $\ell_j=\ell_{j,j+1}$ for the micro step size factor associated with the macro step interval $[t_j, t_{j+1}]$. By construction, $\ell_j$ denotes the step size reduction factor used when expanding the orbit away from the initial point $(\theta_0, \rho_0)$: forward from $(\theta_j, \rho_j)$ to $(\theta_{j+1}, \rho_{j+1})$ if $j > 0$, and backward from $(\theta_{j+1}, \rho_{j+1})$ to $(\theta_j, \rho_j)$ if $j < 0$.   With this notation, \eqref{eq:Phi_j0_bfell} simplifies to:
 \begin{equation} \label{eq:Phi_j0_bfell_2}
\left. \Phi_{j,0} \right|_{\bfell} = \begin{cases} \Phi_{ h \ell_{j-1}^{-1}}^{\ell_{j-1}} \circ  \dots \circ \Phi_{ h \ell_{0}^{-1} }^{\ell_{0}}  & \text{if $j > 0$} \\
\mathcal{F} \circ \Phi_{ h \ell_{j}^{-1}}^{\ell_{j}} \circ  \dots \circ \Phi_{h \ell_{-1}^{-1} }^{\ell_{-1}} \circ \mathcal{F} & \text{if $j<0$}
\end{cases}
\end{equation}
and $\Phi_{0,0}$ is the identity map.

As described in Section~\ref{sec:walnuts-orbit},  a WALNUTS orbit  \[
\calO \ = \  ((\theta_a, \rho_a), \dots, (\theta_b, \rho_b) )  
\] is uniquely determined by: 
\begin{itemize}
    \item an initial point $(\theta, \rho) \in \mathbb{R}^{2d}$;
    \item $m$: the number of doublings; 
    \item $b$: the rightmost index of the orbit; and,
    \item $\bfell=\{\ell_{j}\}_{a \leq j\leq b-1 }$: a collection of $2^{m-1}$ micro step size factors, one for each macro step.
\end{itemize}
The leftmost index is given by $a = b - 2^m + 1$. Accordingly, we use $(\theta, \rho, \calO)$ and the tuple $(\theta, \rho, m, b, \bfell)$ interchangeably. Once the micro step size factors $\bfell$ are sampled, they determine the sequence of forward and backward leapfrog steps used to generate the orbit from the initial point $(\theta, \rho)$.

We treat $m,b,\bfell$ and the integration time index $i \in a{:}b$ as auxiliary variables, and define the augmented state 
\begin{align*}
&z \ := \ (\theta,\rho, m, b, \bfell, i)  
\end{align*}
taking values in the augmented state space
\begin{align*}
\calZ \ := \ \R^d\times \R^d\times \bbN \times \bbN \times \bbN^{2^m-1} \times \bbZ \;.
\end{align*}

We denote the conditional over the orbit by $p_\orbit(m,b,\bfell\mid\theta,\rho)$, as defined by the orbit construction procedure in Section~\ref{sec:walnuts-orbit}. 
The conditional over the integration time  $p_\ind(i\mid \theta,\rho,m,b,\bfell)$ is defined by  biased progressive sampling, as described in Section~\ref{sec:walnuts-int-time}. To verify the conditions of Theorem~\ref{thm:AVM_reversibility}, it is helpful to also consider a simpler alternative.  To distinguish between these two integration time selection strategies, we introduce the following notation, which extends to the corresponding extended target distributions as well:

\begin{itemize}
\item \textbf{Biased progressive sampling} (superscript $\bp$) is the method used in WALNUTS.  It defines $p_\ind^{\bp}$ recursively. Let $\calO_k$ denote the orbit at $k$-th doubling step and $\calW_k$ its associated weights. Then,
\begin{equation} \label{eq:p_ind_bp}
    \begin{aligned}
       p_{\ind}^{\bp}(j\mid\theta,\rho,\calO_{k+1}) &= \alpha(\calW_k, \calW_k^{\ext}) \cdot \frac{w_j}{\sum \calW_k^{\ext} } \cdot \Indc{j \in a_k^{\ext}{:}b_k^{\ext}} \\
       &\quad + \left(1 - \alpha(\calW_k, \calW_k^{\ext}) \right) \cdot p_{\ind}^{\bp}(j\mid\theta,\rho,\calO_k) \cdot \Indc{j \in a_k{:}b_k},
    \end{aligned}
    \end{equation}
    where $\calW_k^{\ext}$ are the weights associated with the extension orbit $\calO_{k}^{\ext}$, and $\alpha(\calW_k, \calW_k^{\ext})$ is the acceptance probability defined in \eqref{eq:walnuts-alpha}.
    
    \item \textbf{Multinomial sampling} (superscript $\mn$) draws the integration time index from a categorical distribution proportional to the orbit weights:
    \begin{equation} \label{eq:p_ind_mn}
    p_{\ind}^{\mn}(j\mid \theta,\rho,\calO) = \frac{w_j}{\sum \calW }, \quad \text{for } j \in a{:}b
    \end{equation}
    where $\calW = (w_a, \dots, w_b)$ are the weights associated with the orbit $\calO$.

\end{itemize}

The superscripts $\bp$ and 
$\mn$ also apply to the corresponding extended target distributions.  Let $\zeta$ denote the reference measure on $\mathbb{A}$,  given by the product of Lebesgue measure on the first two components and counting measure on the remaining components.   The extended target density with respect to $\zeta$ under biased progressive sampling is given by
\begin{align} \label{eq:walnuts-joint}
    p_{\joint}^{\bp}(z) \ &\propto  \ e^{-H(\theta,\rho)} \cdot p_{\orbit}(m, b, \bfell\mid \theta,\rho) \cdot p_{\ind}^{\bp}(i \mid \theta,\rho, m, b, \bfell ) 
\end{align}
and the analogous definition $p_{\joint}^{\mn}$ is obtained by replacing $p_{\ind}^{\bp}$ with $p_{\ind}^{\mn}$.  The proportionality symbol reflects that $e^{-H(\theta,\rho)}$ is not necessarily normalized.

For any integer $k \in \mathbb{Z}$, we introduce the shift operator $\mathcal{S}^k$ which acts on $\bfell$ as follows: \[
(\shift^k \bfell)_j := \ell_{j-k} \;.
\] 
This operator maps the micro step size factor associated with the macro step $[t_j, t_{j+1}]$ to the one associated with $[t_{j-k}, t_{j-k+1}]$.
 In terms of this shift operator, define the map $\Psi: \calZ \to \calZ$ by  \begin{align} \label{eq:walnuts-inv}
    \Psi(z)  \ = \ (\left. \Phi_{i,0} \right|_{\bfell}(\theta,\rho), m, b-i, \shift^{i}\bfell, -i ) \;.
\end{align}
This map is carefully designed to have the key properties stated in the lemmas that follow: in particular, $\Psi$ will be shown to be a $\zeta$-preserving involution (Lemma~\ref{lem:measure-preserving-involution-Psi}) under which the extended target distribution remains invariant (Lemma~\ref{lem: p joint bp}). These properties are essential for understanding the reversibility of WALNUTS.

\medskip

\begin{lemma}
\label{lem:measure-preserving-involution-Psi}
    The mapping $\Psi$ is a $\zeta$-preserving involution on $\calZ$.
\end{lemma}
\begin{proof}
    The map $\Psi$ is an involution, since for all $z\in\mathcal Z$,
\begin{align*}
	\Psi\circ\Psi(z)\ &=\ \Psi(\left. \Phi_{i,0} \right|_{\bfell} (\theta,\rho),m, b-i, \shift^{i}\bfell, -i ) \\
    \ &=\  (\left. \Phi_{-i,0} \right|_{\shift^{i} \bfell} (\theta_i,\rho_i),m, b, \bfell, i ) 
    \ = \ z\;.
\end{align*}
In the last step, we used the identity \[
\left. \Phi_{-i,0} \right|_{\shift^{i} \bfell} \circ \left. \Phi_{0,i} \right|_{\bfell} \ = \  \operatorname{id}_{\mathbb{R}^{2d}} \;,
\] which follows from the definition of $\left. \Phi_{0,i} \right|_{\bfell}$ in \eqref{eq:Phi_j0_bfell_2} and the fact that $(\shift^i \bfell)_j = \ell_{j-i}$ reindexes the micro step size factors appropriately for reversal.  The following diagram summarizes the composition for $i>0$,  with red arrows indicating the forward trajectory constructed by $\left. \Phi_{0,i} \right|_{\bfell}$ and gray arrows indicating the reversed trajectory constructed by $\left. \Phi_{-i,0} \right|_{\shift^{i} \bfell}$, which returns the state to its original position:
\begin{center}
\begin{tikzpicture}[>=Stealth, node distance=2.5cm]
\node (x0) at (0,0) {$(\theta, \rho)$};
\node (x1) [right of=x0] {$(\theta_1, \rho_1)$};
\node (x2) [right of=x1] {$\cdots$};
\node (xi-1) [right of=x2] {$(\theta_{i-1}, \rho_{i-1})$};
\node (xi) [right of=xi-1] {$(\theta_i, \rho_i)$};

\draw[->, red, ultra thick, bend left=45] (x0) to node[above] {\small $\ell_0$} (x1);
\draw[->, red, ultra thick, bend left=45] (x1) to node[above] {\small $\ell_1$} (x2);
\draw[->, red, ultra thick, bend left=45] (x2) to node[above] {\small $\ell_{i-2}$} (xi-1);
\draw[->, red, ultra thick, bend left=45] (xi-1) to node[above] {\small $\ell_{i-1}$} (xi);

\draw[<-, gray, ultra thick, bend right=45] (xi-1) to node[below] {\small $(\shift^i \ell)_{-1}$} (xi);
\draw[<-, gray, ultra thick, bend right=45] (x2) to node[below] {\small $(\shift^i \ell)_{-2}$} (xi-1);
\draw[<-, gray, ultra thick, bend right=45] (x1) to node[below] {\small $(\shift^i \ell)_{-i+1}$} (x2);
\draw[<-, gray, ultra thick, bend right=45] (x0) to node[below] {\small $(\shift^i \ell)_{-i}$} (x1);
\end{tikzpicture}
\end{center}
The transition from $(\theta,\rho)$ to $(\theta_i,\rho_i)$ is a composition of multiple leapfrog steps, therefore the transition is measure preserving in the first two components. The counting measure on $\bbZ$ is invariant under reflection. The counting measure on $\bbN$ is invariant under translation. This proves that $\Psi$ preserves  $\zeta$.
\end{proof}

For notational brevity, we define
\begin{align*}
    p^+_j \ &= \ p_{\micro}(\ell_{j} \mid \micro(\theta_{j},\rho_{j},\mu, M, h, \delta)) \;, \\
    p^-_j \ &= \  p_{\micro}(\ell_{j} \mid \micro(\theta_{j+1},-\rho_{j+1},\mu, M, h, \delta)) \;.
\end{align*}
And for any integers $j\leq k$, we write
\begin{align*}
    p^+_{j:k} \ &= \  p^+_j\cdot p^+_{j+1}\cdots p^+_k,  \\ 
    p^-_{j:k} \ &= \  p^-_j\cdot p^-_{j+1}\cdots p^-_k.
\end{align*}
If $j>k$, we use the convention that $p^+_{j{:}k}=1$ and $p^-_{j{:}k}=1$.

\medskip

\begin{lemma} \label{lem:orbit}
    The orbit selection kernel of WALNUTS is given by
    \begin{align*}
        p_{\orbit}(m, b, \bfell\mid \theta,\rho ) 
        =\frac{1}{2^m}\cdot p^+_{0{:}b-1} \cdot p^-_{a{:}-1}.
    \end{align*}
\end{lemma}
\begin{proof}[Proof of Lemma~\ref{lem:orbit}]
    A WALNUTS orbit is generated by first sampling a sequence of Bernoulli variables $B=(B_1,\ldots,B_{m_{\max}})$ and  micro step size factors $\{\ell_j\}_{a_{m_{\max}}\leq j < b_{m_{\max}}}$. Given $B$ and $\{\ell_j\}_{a_{m_{\max}}\leq j < b_{m_{\max}} }$, the selection of $\calO=\calO_m$ is determined deterministically by the NUTS-style doubling procedure, and one of the following conditions must hold: 
    \begin{enumerate}
        \item $\calO_m$ does not satisfy the sub-U-turn condition, but its extension $\calO_m^{\ext}$ does;
        \item $\calO_m$ satisfies the U-turn condition, but none of its sub-orbits do;
        \item The maximum number of doublings is reached, i.e., $m=m_{\max}$.
    \end{enumerate}
    Let us write this deterministic selection as $\calO_m = \text{NUTS} (B, \{\ell_j\}_{a_{m_{\max}}\leq j < b_{m_{\max}} })$.
    Then, the probability of selecting orbit $\calO_m$ is
    \begin{align*}
        &p_\orbit(\calO_m\mid \theta,\rho ) = \sum_{B' }\sum_{\ell'} p(B') \cdot p(\ell'\mid B' )\cdot \Indc{\text{NUTS}(B',\ell' ) = \calO_m },
    \end{align*}
    where the sums range over all possible values of $ B' \in \{0,1\}^{m_{\max}} $ and $ \ell' = \{ \ell'_j \}_{a_{m_{\max}} \leq j < b_{m_{\max}}} $, and where
    \begin{align*}
        p(B') &= \frac{1}{2^{m_{\max}}},\\
        p(\ell'\mid B') &= \prod_{j=0}^{b_{\max-1} } p_\micro(\ell'_j \mid \micro(\theta'_j,\rho'_j, \mu, M, h, \delta) ) \cdot \prod_{j=a_{\max} }^{-1} p_\micro(\ell'_j \mid \micro(\theta'_{j+1}, -\rho'_{j+1}, \mu, M, h, \delta) ),
    \end{align*}
    and $(\theta'_j,\rho'_j)$ represents the $j$-th state in the orbit constructed from $B',\ell'$.
    
    The indicator $\Indc{\text{NUTS}(B',\ell' ) = \calO_m }$ is equal to 1 if and only if 
    \begin{align*}
        (B'_1,\ldots,B'_m)= (B_1,\ldots,B_m),\text{ and }
        \ell'_j=\ell_j \text{ for all } a_m\leq j\leq b_m-1 .
    \end{align*}
Hence, the probability of selecting $\calO_m$ is equal to the probability of matching the first $ m $ Bernoulli steps and the corresponding micro step factors:
    \begin{align*}
        \frac{1}{2^m} \cdot \prod_{j=0}^{b_m-1} p_\micro(\ell_j\mid \micro(\theta_j,\rho_j, \delta)) \cdot \prod_{j=a_m}^{-1} p_\micro(\ell_j\mid \micro(\theta_{j+1}, -\rho_{j+1}, \delta) ) = \frac{1}{2^m} \cdot p^+_{0:b-1} \cdot p^-_{a:-1}.
    \end{align*}
\end{proof}

\begin{lemma} \label{lem:mn}
    The index selection kernel for multinomial sampling is given by
    \begin{align*}
        p_{\ind}^{\mn}(j\mid \theta,\rho, \calO ) = \frac{e^{-H(\theta_j,\rho_j)}}{Z(\theta, \rho, \calO)}  \cdot \begin{cases}
             \dfrac{p^-_{0:j-1}} {p^+_{0:j-1} } & 0\leq j\leq b\\
             \dfrac{p^+_{j:-1}}{p^-_{j:-1}} & a\leq j < 0\\
        \end{cases}
    \end{align*}
    where the normalizing constant is given by
    \begin{align*}
        Z(\theta,\rho,\calO) = \sum_{j=0}^b e^{-H(\theta_j,\rho_j)} \frac{p^-_{0:j-1}} {p^+_{0:j-1} } + \sum_{j=a}^{-1} e^{-H(\theta_j,\rho_j)} \frac{p^+_{j:-1}}{p^-_{j:-1}}.
    \end{align*}
\end{lemma}
Lemma~\ref{lem:mn} follows directly from the definitions of $p_{\ind}^{\mn}$ in \eqref{eq:p_ind_mn} and the weights $w_j$ in \eqref{eq:w0}, \eqref{eq:wts-forw} and \eqref{eq:wts-backw}.

\medskip

\begin{lemma}\label{lem: p joint mn}
The extended target density with multinomial sampling is invariant under $\Psi$ in \eqref{eq:walnuts-inv}, i.e., $p_\joint^\mn\circ\Psi=p_\joint^\mn$.
\end{lemma}
\begin{proof}
The case $i=0$ is trivial, as $\Psi$ is the identity.  We therefore assume $i > 0$; the case $i < 0$ is analogous and omitted. By Lemmas~\ref{lem:orbit} and~\ref{lem:mn}, the extended target evaluated at $z$ satisfies:
\begin{align}
    p_{\joint}^\mn(z) \ &\propto \  e^{-H(\theta,\rho)} \cdot 
    \frac{1}{2^m} p^+_{0{:}b-1} p^-_{a{:}-1}\cdot
    \frac{1}{Z(\theta,\rho,\calO)} e^{-H(\theta_i,\rho_i)} \frac{p^-_{0:i-1}} {p^+_{0:i-1} } \nonumber \\
   \  &\propto \ e^{-H(\theta,\rho) - H(\theta_i,\rho_i)} \frac{1}{2^m} p^-_{a:i-1} p^+_{i:b-1} \frac{1}{Z(\theta,\rho,\calO)} \;. \label{eq:pmn_z}
\end{align}
Now consider the transformed point $z' = \Psi(z)$, which shifts the proposal index $i$ to $0$ and the initial index $0$ to $-i$, and reindexes the micro step size factors via the shift operator: $\shift^i \bfell$ satisfies $(\shift^i \bfell)_j = \ell_{j - i}$. Thus, micro step factors and integration times in the reverse orbit are properly aligned,  and hence, we simply write: $(\theta_i,\rho_i,\calO ) = (\theta_i, \rho_i, m, \shift^i \bfell, b-i)$ which corresponds to the same orbit now viewed from the new initial point $(\theta_i, \rho_i)$. Evaluating the extended target at $z'$ yields:
\begin{align}
    p_{\joint}^\mn(z')
   \  &\propto \  e^{-H(\theta_i,\rho_i)} \cdot \frac{1}{2^m} p^+_{i:b-1}  p^-_{a:i-1} \cdot \frac{1}{Z(\theta_i,\rho_i,\calO )}  e^{-H(\theta,\rho)} \frac{ p^+_{0:i-1} }{p^-_{0:i-1} } \nonumber \\
   \  &\propto \  e^{-H(\theta,\rho) - H(\theta_i,\rho_i)} \frac{1}{2^m}p^-_{a:-1} p^+_{0:b-1}  \frac{1}{Z(\theta_i,\rho_i,\calO )} \;. \label{eq:pmn_zp}
\end{align}
Comparing \eqref{eq:pmn_z} to \eqref{eq:pmn_zp}, to prove $p_\joint^\mn(z') = p_\joint^\mn(z)$, it suffices to show
\begin{align} \label{eq:Z_identity}
   Z(\theta_i,\rho_i,\calO) \cdot p^-_{0:i-1} \ = \  Z(\theta,\rho,\calO) \cdot p^+_{0:i-1} \;.
\end{align}

Let $\mu_j := e^{-H(\theta_j, \rho_j)}$ denote the unnormalized weight at index $j$. We first expand $Z(\theta, \rho, \calO)$:
\begin{align*}
    Z(\theta,\rho,\calO)\cdot p^+_{0:i-1} &= \Big(\sum_{j=0}^b \mu_j \frac{p^-_{0:j-1}} {p^+_{0:j-1} } + \sum_{j=a}^{-1} \mu_j \frac{p^+_{j:-1}}{p^-_{j:-1}} \Big) \cdot p^+_{0:i-1}  \\ 
     &=\sum_{j=0}^i \mu_j p^-_{0:j-1} p^+_{j:i-1} + \sum_{j=i+1}^b \mu_j \frac{p^-_{0:j-1}}{p^+_{i:j-1}}  + \sum_{j=a}^{-1} \mu_j \frac{p^+_{j:i-1}}{p^-_{j:-1}} \\
    &=\sum_{j=0}^{i-1} \mu_j p^-_{0:j-1} p^+_{j:i-1} + \sum_{j=i}^b \mu_j \frac{p^-_{0:j-1}}{p^+_{i:j-1}}  + \sum_{j=a}^{-1} \mu_j \frac{p^+_{j:i-1}}{p^-_{j:-1}}.\numberthis\label{eq:Ztheta}
\end{align*}
Here we used that $p^+_{j{:}k}=1$ and $p^-_{j{:}k}=1$, if $j>k$.
Next, we expand $Z(\theta_i, \rho_i, \calO)$:
\begin{align*}
    Z(\theta_i,\rho_i,\calO)\cdot p^-_{0:i-1} &= \Big( \sum_{j=i}^b \mu_j  \frac{p^-_{i:j-1}}{p^+_{i:j-1}} + \sum_{j=a}^{i-1} \mu_j  \frac{p^+_{j:i-1}}{p^-_{j:i-1}} \Big) \cdot  p^-_{0:i-1} \\
    &=\sum_{j=i}^b \mu_j \frac{p^-_{0:j-1}}{p^+_{i:j-1}} + \sum_{j=0}^{i-1}\mu_j p^-_{0:j-1}p^+_{j:i-1} +  \sum_{j=a}^{-1} \mu_j \frac{p^+_{j:i-1}}{p^-_{j:-1}}.
    \numberthis\label{eq:Zthetaprime}
\end{align*}
The right-hand sides of \eqref{eq:Ztheta} and \eqref{eq:Zthetaprime} are identical term-by-term. Therefore, \eqref{eq:Z_identity} holds, and the invariance $p_\joint^\mn \circ \Psi = p_\joint^\mn$ follows. \end{proof}

Lemmas~\ref{lem:measure-preserving-involution-Psi} and~\ref{lem: p joint mn} verify the conditions of Theorem~\ref{thm:AVM_reversibility} in the case of WALNUTS with multinomial integration time sampling. Consequently, the associated transition kernel is reversible. We do not state this as a separate theorem, however, since our primary focus is on WALNUTS with biased progressive sampling.

The following result is key to understanding why WALNUTS, like biased progressive HMC and NUTS, always accepts the proposed next state without the need for a Metropolis correction. In the GIST framework (Section~\ref{sec:gist}), the acceptance probability for transitioning from one extended state to another is determined by a Metropolis ratio of extended target densities. Lemma~\ref{lem: p joint bp} shows that this Metropolis ratio is always equal to one, thereby ensuring automatic acceptance and preserving detailed balance without the need to explicitly compute an acceptance probability.

\medskip
\begin{lemma}\label{lem: p joint bp}
The extended target density with biased progressive sampling in \eqref{eq:walnuts-joint} is invariant under $\Psi$ in \eqref{eq:walnuts-inv}, i.e.,  $p_\joint^{\bp}\circ\Psi=p_\joint^{\bp}$.
\end{lemma}
\begin{proof}

The proof is by induction on the number of doublings $m$. For the base case $m=1$, assume $a=0$, $b=1$, and $i=1$. The claim for $i=0$ is trivial, and the proof for $a=-1$, $b=0$ follows by symmetry.  The orbit consists of two states, $(\theta_0,\rho_0)$ and $(\theta_1,\rho_1)$, and the orbit and index selection kernels under biased progressive sampling are:
\begin{align*}
    &p_\orbit(m,b,\ell_0 \mid\theta_0,\rho_0 ) =\frac12 p^+_0, \quad p_\orbit(m,b-1,\ell_0 \mid\theta_1,\rho_1 ) =\frac12 p^-_0,\\
    &p_\ind^{\bp}(1\mid \theta_0, \rho_0, \calO ) = \min(1, \frac{e^{-H(\theta_1,\rho_1)} \frac{p^-_0}{p^+_0} } {e^{-H(\theta_0,\rho_0)}} ), \quad
    p_\ind^{\bp}(-1\mid \theta_1, \rho_1, \calO ) = \min(1, \frac{e^{-H(\theta_0,\rho_0)} \frac{p^+_0}{p^-_0} } {e^{-H(\theta_1,\rho_1)}} ).
\end{align*}
Thus, the extended target evaluated at $z = (\theta, \rho, m, b, \bfell, i) = (\theta,\rho,1,1,\ell_0,1)$ is
\begin{align*}
    p_\joint^\bp(\theta,\rho,1,1,\ell_0,1)&\propto e^{-H(\theta_0,\rho_0)} \cdot p^+_0 \cdot \min(1, \frac{e^{-H(\theta_1,\rho_1)} \frac{p^-_0}{p^+_0} } {e^{-H(\theta_0,\rho_0)}} )\\
    &= \min(e^{-H(\theta_0,\rho_0)}\cdot p^+_0,\; e^{-H(\theta_1,\rho_1)} \cdot p^-_0 ).
\end{align*}
Similarly, the extended target evaluated at $z' = \Psi(z)$ is
\begin{align*}
    p_\joint^\bp(\theta_1,\rho_1,1,0,\ell_0,-1)&\propto e^{-H(\theta_1,\rho_1)} \cdot p^-_0 \cdot \min(1, \frac{e^{-H(\theta_0,\rho_0)} \frac{p^+_0}{p^-_0} } {e^{-H(\theta_1,\rho_1)}} )\\
    &= \min(e^{-H(\theta_0,\rho_0)}\cdot p^+_0,\; e^{-H(\theta_1,\rho_1)} \cdot p^-_0 ).
\end{align*}
This confirms that the extended target is invariant under $\Psi$ when the orbit contains two states.

Assume that the claim is true when $\calO$ has size $2,2^2,\ldots,2^{m}$. We now prove it holds when $\calO$ has size $2^{m+1}$. There are two cases depending on whether $(\theta_i,\rho_i)$ belongs to $\calO_m^\ext$ or $\calO_{m}$.
\begin{description}
   \item[Case 1: $(\theta_i, \rho_i) \in \calO_m^\ext$.]
    By \eqref{eq:p_ind_bp}, the probability of choosing $i$ under biased progressive sampling is
    \begin{align*}
        p_\ind^{\bp}(i\mid \theta_0,\rho_0, \calO_{m+1}) &= 
        \alpha(\calW_m, \calW_m^{\ext}) \cdot \frac{w_i}{\sum \calW_m^{\ext}} 
        \\
        &=\alpha(\calW_m, \calW_m^{\ext}) \cdot \frac{w_i}{\sum  \calW_{m+1}}    \frac{\sum \calW_{m+1} }{\sum  \calW^\ext_m }\\
        &=\alpha(\calW_m, \calW_m^{\ext}) \cdot p_{\ind}^\mn(i\mid \theta_0,\rho_0,\calO_{m+1} ) \cdot \frac{\sum \calW_{m+1} }{ \sum  \calW^\ext_m } \\
        &= p_{\ind}^\mn(i\mid \theta_0,\rho_0,\calO_{m+1} ) \cdot \min\left(\frac{1}{\sum  \calW^\ext_m}, \frac{1}{\sum \calW_m} \right) \cdot \sum \calW_{m+1}.
    \end{align*}
    Thus,
    \begin{align*}
        &p_\joint^\bp(\theta_0,\rho_0,m,b,\bfell,i) \\
        &\qquad \propto e^{-H(\theta_0,\rho_0)}\cdot p_\orbit(m,b,\bfell\mid\theta_0,\rho_0) \cdot p_{\ind}^\mn(i\mid \theta_0,\rho_0,\calO_{m+1}) \cdot \min\left(\frac{1}{\sum  \calW^\ext_m}, \frac{1}{\sum \calW_m} \right) \cdot \sum \calW_{m+1} \\
        &\qquad \propto p_{\joint}^\mn(\theta_0,\rho_0, \calO_{m+1}, i) \cdot \min\left(\frac{1}{\sum  \calW^\ext_m}, \frac{1}{\sum \calW_m} \right) \cdot \sum \calW_{m+1}.
    \end{align*}

Similarly, for the transformed point:
    \begin{align*}
        p_\ind^{\bp}(-i\mid \theta_i,\rho_i, \calO_{m+1}) &= p_{\ind}^\mn(-i\mid \theta_i,\rho_i,\calO_{m+1}) \cdot \min\left(\frac{1}{\sum  \calW^\ext_m}, \frac{1}{\sum \calW_m} \right) \cdot \sum \calW_{m+1}.
    \end{align*}
    and
    \begin{align*}
        &p_\joint^\bp(\theta_i,\rho_i,m,b,\bfell,-i)  = p_{\joint}^\mn(\theta_i,\rho_i, \calO, -i) \cdot \min\left(\frac{1}{\sum  \calW^\ext_m}, \frac{1}{\sum \calW_m} \right) \cdot \sum \calW_{m+1}.
    \end{align*}
    By Lemma~\ref{lem: p joint mn}, we have
    \begin{align*}
         p_{\joint}^\mn(\theta_0,\rho_0, \calO, i)=p_\joint^\mn(\theta_i,\rho_i,\calO,-i),
    \end{align*}
    and hence, $p_{\joint}^\bp(\theta_0,\rho_0, \calO, i)=p_\joint^\bp(\theta_i,\rho_i,\calO,-i)$.  This concludes the proof for the case $(\theta_i,\rho_i)\in\calO_m^\ext$.
    
    \item[Case 2: $(\theta_i, \rho_i) \in \calO_m$.]  By \eqref{eq:p_ind_bp}, we have
    \begin{align*}
        p_\ind^{\bp}(i\mid \theta_0,\rho_0,\calO_{m+1})=(1 - \alpha(\calW_m, \calW_m^{\ext}) ) \cdot p_\ind^{\bp}(i\mid\theta_0,\rho_0,\calO_{m}).
    \end{align*}

    Applying the induction hypothesis to $\calO_m$ gives
    \begin{align*}
        e^{-H(\theta,\rho)}\cdot p_{\orbit}(\calO_m\mid\theta,\rho)\cdot p_\ind^{\bp}(i\mid\theta,\rho,\calO_m)=e^{-H(\theta_i,\rho_i)}\cdot p_{\orbit}(\calO_m\mid\theta_i,\rho_i)\cdot p_\ind^{\bp}(-i\mid\theta_i,\rho_i,\calO_m).
    \end{align*}

Suppose $\calO_m^\ext$ is the right half of $\calO_{m+1}$, and let $b_m$ be the rightmost index of $\calO_m$.
(The argument for the other direction is the same.) Then we have
    \begin{align*}
        p_\orbit(\calO_{m+1} \mid\theta,\rho) &=\frac12 p_{\orbit}(\calO_m\mid\theta,\rho) \cdot \prod_{j=b_m}^{j=b_{m+1} - 1} p^+_j,\\
        p_\orbit(\calO_{m+1} \mid\theta_i,\rho_i) &=\frac12 p_{\orbit}(\calO_m \mid\theta_i,\rho_i) \cdot \prod_{j=b_m}^{j=b_{m+1} - 1} p^+_j.
    \end{align*}
Combining these identities confirms that
\[
    p_\joint^\bp(\theta, \rho, \calO_{m+1}, i) = p_\joint^\bp(\theta_i, \rho_i, \calO_{m+1}, -i).
\]
\end{description}
\end{proof}

Lemmas~\ref{lem:measure-preserving-involution-Psi} and~\ref{lem: p joint bp} verify the conditions of Theorem~\ref{thm:AVM_reversibility} for WALNUTS with biased progressive integration time sampling. As a result, the following reversibility result holds.

\medskip

\begin{theorem}
\label{thm:walnuts}
The WALNUTS transition kernel is reversible with respect to the target distribution~$\mu$.
\end{theorem}

\subsection{Pseudocode Implementation of WALNUTS}

WALNUTS requires the following inputs: an unnormalized target density $\mu$, a mass matrix $M$, a macro step size $h > 0$, a maximum number of doublings $m_{\max}$, and a leapfrog energy error threshold $\delta > 0$. The probability distribution $p_{\micro}$, which governs the sampling of micro step size reduction factors, is defined internally in terms of the output of the \texttt{micro} subroutine (see Listing~\ref{algo:micro_lf}).  Complete pseudocode for all components is summarized below and detailed in Appendix~\ref{sec:pseudocode}.  We emphasize that, unlike NUTS, WALNUTS incorporates variable step size leapfrog integration as a central feature of its design.

The main WALNUTS algorithm is given in Listing~\ref{algo:WALNUTS}, which implements orbit construction and biased progressive state selection. Orbit expansion is handled by the \texttt{extend-orbit-forward} and \texttt{extend-orbit-backward} procedures (Listings~\ref{algo:extend-forward} and~\ref{algo:extend-backward}, respectively), which generate new segments of the variable step size leapfrog trajectory along with corresponding weights. The U-turn and sub-U-turn conditions used to terminate orbit growth are checked by the \texttt{U-turn} and \texttt{sub-U-turn} subroutines (Listings~\ref{algo:U-turn} and~\ref{algo:sub-U-turn}). Finally, the local step size selection is handled by the \texttt{micro} subroutine (Listing~\ref{algo:micro_lf}), which adaptively determines the number of micro steps needed to keep the energy error within the specified threshold $\delta$.

\section{Empirical Evaluations}

\label{sec:empirical}

The results and plots presented below can be reproduced using a Python implementation of WALNUTS, which is available on GitHub under a permissive open-source license.\footnote{The code is released under the MIT License and can be found at \url{https://github.com/bob-carpenter/walnuts}.}

\subsection{Setup for Numerical Experiments}
Throughout this and the following sections, unless otherwise specified, WALNUTS is run with a default micro step distribution $p_{\micro}$ defined as follows. Given a current state $(\theta, \rho) \in \mathbb{R}^{2d}$, the function call
\[
\widetilde{\ell} = \texttt{micro}(\theta, \rho, \mu, M, h, \delta)
\]
returns the smallest power of two $\widetilde{\ell} \in 2^{\mathbb{N}}$ such that the energy error of the leapfrog trajectory satisfies $H^+_{\widetilde{\ell}} - H^-_{\widetilde{\ell}} \leq \delta$, as defined in~\eqref{eq:micro}.

The \emph{randomized two-point}  distribution is then defined by:
\begin{equation} \tag{R2P}
p_{\micro}(\ell \mid \widetilde{\ell}) =
\begin{cases}
    2/3 & \text{if } \ell = \widetilde{\ell}, \\
    1/3 & \text{if } \ell = \widetilde{\ell}+1, \\
    0   & \text{otherwise}.
\end{cases}
\end{equation}
The corresponding WALNUTS method using this distribution is referred to as \emph{WALNUTS–R2P}.

In contrast, the \emph{deterministic} variant corresponds to the degenerate case:
\begin{equation} \tag{D}
p_{\text{micro}}(\ell \mid \widetilde{\ell}) = 
\begin{cases}
1 & \text{if } \ell = \widetilde{\ell}, \\
0 & \text{otherwise}.
\end{cases}
\end{equation}
The corresponding method is referred to as \emph{WALNUTS–D}.
These two variants will be compared in the numerical studies that follow.

When measuring performance in terms of gradient evaluations, both the forward and backward $\micro$-computations are counted, even though backward $\micro$-computations may be run on different a processor and thus only adding negligible (wall clock) computation time. 

To avoid the \emph{looping phenomenon}, we jitter the macro step size throughout. NUTS and WALNUTS rely on the U-turn condition to adaptively truncate their orbit lengths.  However, for certain fixed step sizes, this condition can become ineffective: the orbit selection procedure may fail to detect a U-turn and instead expand to the maximum allowed orbit length  and severely impair mixing \cite[Figure 4(a) and 5(b)]{BoOb2024}. This behavior, known as the looping phenomenon, arises from resonance between the doubling procedure and a natural frequency of the underlying Hamiltonian dynamics, causing the physical orbit length to repeatedly miss the U-turn detection range.   A simple and practical remedy is to jitter the leapfrog step size---either once per orbit or at each leapfrog step---to break this periodicity. As shown empirically, such randomization mitigates looping behavior and improves performance \cite[Figure 5(c) and 5(d)]{BoOb2024}. Throughout our experiments, we uniformly jitter the macro step size by $\pm20\%$ at each integration step.

\subsection{High-Dimensional Gaussian}

\begin{figure}
\centering
\begin{tikzpicture}
\begin{axis}[
    width=12cm,
    height=5cm,
    domain=-0.1:4.5,
    samples=200,
    axis lines=middle,
    xlabel={$\|\theta\|^2$},
    ylabel={density},
    tick style={line width=2pt},
    major tick length=6pt,
    xtick={1,3},
    xticklabels={$0$, $d$},
    ytick=\empty,
    axis line style={->, very thick},
    xlabel style={at={(axis description cs:1.0,0)},anchor=west},
    ylabel style={at={(axis description cs:0.04,1)},anchor=south},
    clip=false,
 ]

\addplot[gray, ultra thick, domain=0:4.5] {exp(-0.5*((x - 3)^2)/0.1)};

\addplot [
    draw=none,
    fill=red,
    fill opacity=0.2,
    domain=2.4:3.6
] {exp(-0.5*((x - 3)^2)/0.1)} \closedcycle;

\addplot[red, dashed, ultra thick] coordinates {(2.4,0) (2.4,{1})};
\addplot[red, dashed, ultra thick] coordinates {(3.6,0) (3.6,{1})};

\addplot[red, thin] coordinates {(3,0) (3,{exp(-0.5*((3 - 3)^2)/0.1)})};

\node[red] at (axis cs:2.4,0.4) [anchor=east] {$d - \alpha$};
\node[red] at (axis cs:3.6,0.4) [anchor=west] {$d + \alpha$};

\draw[->, thick] (axis cs:3.9,0.88) -- (axis cs:3.62,0.7);
\node at (axis cs:3.9,0.95) [anchor=west] {$\alpha = O(\sqrt{d})$};

\draw[->, blue, thick] (axis cs:0.65,0.4) -- (axis cs:1,0.02);
\node[blue] at (axis cs:0.35,0.45) [anchor=west] {mode};

\end{axis}
\end{tikzpicture}
\caption{\textbf{Concentration of  High-Dimensional Gaussian.} Illustration of the concentration of $|\theta|^2$ for $\theta \sim \mathcal{N}(0,I_d)$. The shaded region shows the high-probability shell defined in \eqref{eq:gauss-concentration}, centered around $|\theta|^2 = d$ with width $O(\sqrt{d})$. The mode at $\theta = 0$, corresponding to $|\theta|^2 = 0$,  lies well outside this concentration region.}
\label{fig:gaussian-concentration}
\end{figure}
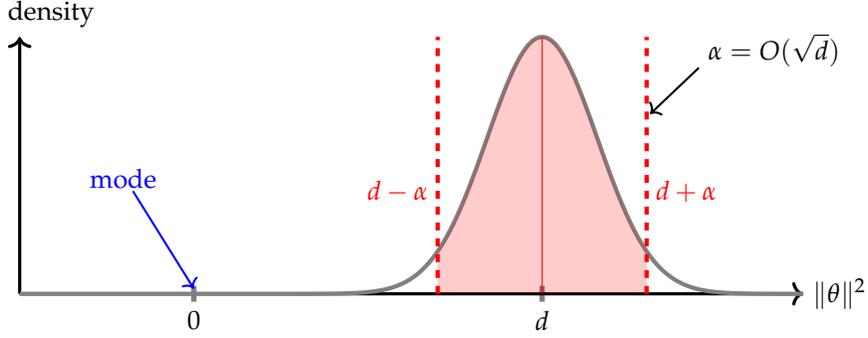

\label{sec:gaussian}

\begin{figure}
    \centering
    \includegraphics[width=0.99\linewidth]{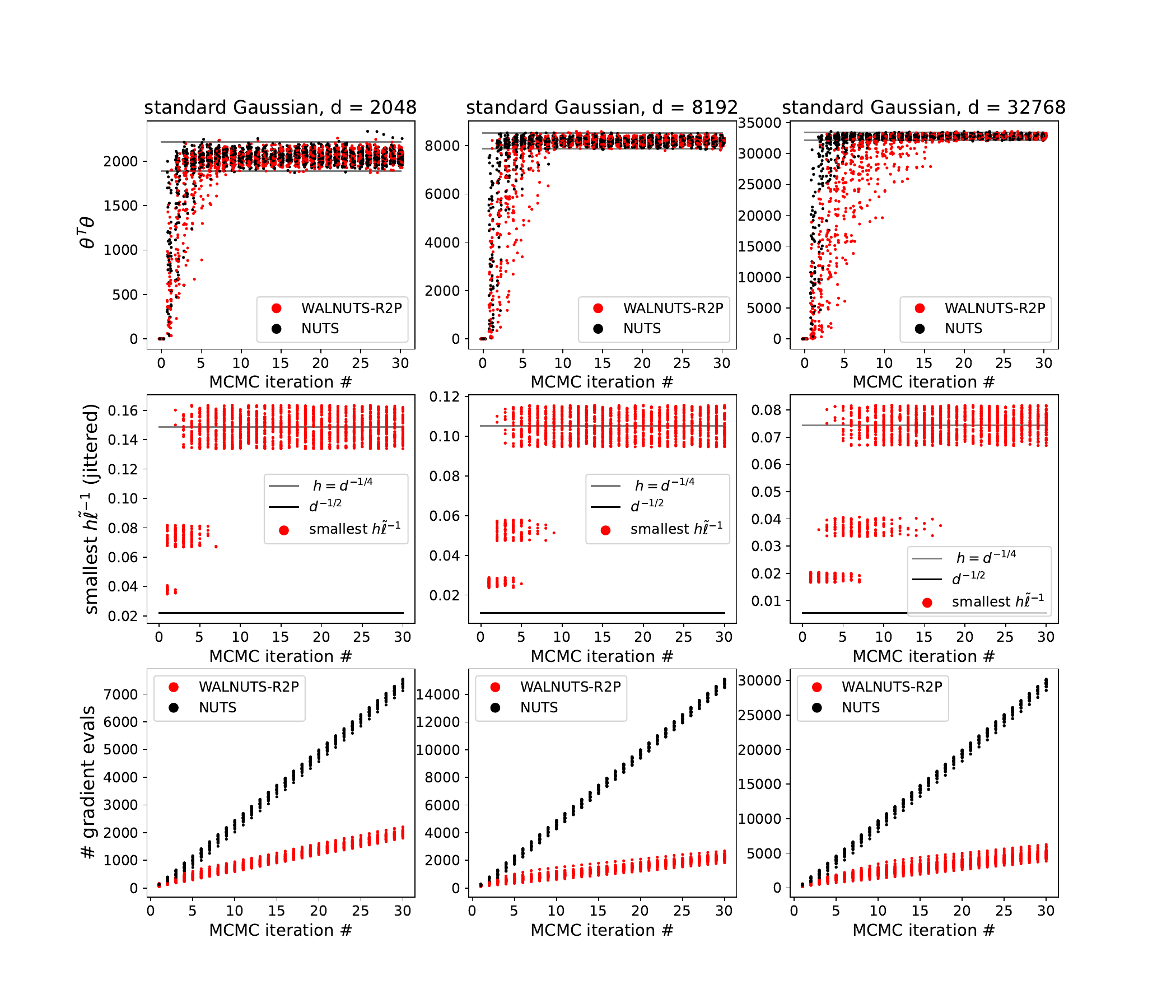}
    \caption{ \textbf{Cold-Start Behavior and Step Size Adaptation in High-Dimensional Gaussians.} In all cases, 50 independent chains were initialized at the mode $\theta = 0$ and run for 30 MCMC iterations. WALNUTS was tuned for the stationary regime using $h = d^{-1/4}$ and $\delta = 0.3$, while NUTS was tuned specifically for the cold-start regime using $h = d^{-1/2}$. The upper panels show  $|\theta|^2$ per MCMC iteration for a standard Gaussian target distribution in the indicated dimension $d$. Red horizontal lines indicate the intervals covering 99\% of the probability mass of $|\theta|^2$. The middle panels show the smallest value of $h\tilde{\ell}^{-1}$ observed within each orbit (jittered for readability), which may be interpreted as the local micro step size selected by \texttt{micro}. The lower panels show the cumulative gradient count after each MCMC iteration for the same simulations. }
    \label{fig:gauss-transient}
\end{figure}

\begin{figure}
    \centering
    \includegraphics[width=0.9\linewidth]{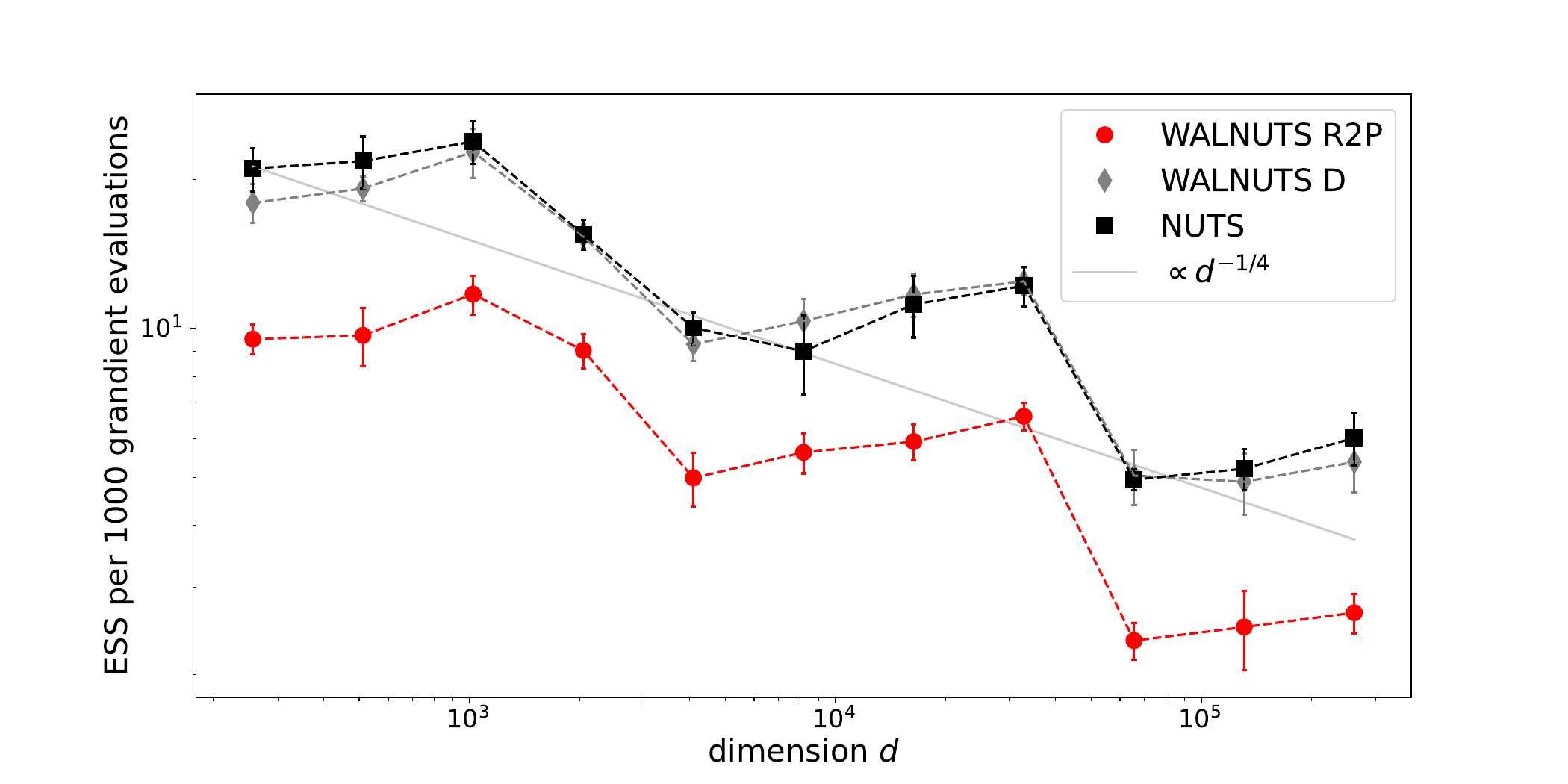}
    \caption{\textbf{Efficiency of WALNUTS in High-Dimensional Gaussians.} Effective sample size (ESS) per 1000 gradient evaluations for a standard Gaussian target in varying dimensions, under warm-start initialization. Results for each dimension are based on 10 chains, each run for 1000 iterations. Mean ESS per 1000 gradient evaluations is indicated by symbols; vertical bars show 95\% confidence intervals for the exact ESS. All methods used a step size of $h = 1.4 \, d^{-1/4}$, which was found by trial and error to be near-optimal for NUTS. Both WALNUTS variants used $\delta = 0.3$.}
    \label{fig:gauss-ESS}
\end{figure}

\paragraph{Geometry of Gaussian concentration}

This part evaluates the performance of WALNUTS on a $d$-dimensional standard Gaussian target with density $\mu(\theta) \propto \exp(-|\theta|^2/2)$ for $\theta \in \mathbb{R}^d$. We compare WALNUTS to the standard No-U-Turn Sampler (NUTS) using a fixed step size leapfrog integrator. A $d \times d$-identity mass matrix $M=I_d$ is used, so that $\rho \sim \mathcal{N}(0, I_d)$.

It is well known that in high dimensions, the standard Gaussian concentrates its mass in a thin spherical shell around the sphere of radius $\sqrt{d}$. Specifically,
\begin{equation} \label{eq:gauss-concentration}
\mu(D_\alpha^c) \leq 2 \exp(-\alpha^2 / 8d), \quad
D_\alpha = \left\{ \theta \in \mathbb{R}^d : \left| \|\theta\|^2 - d \right| \leq \alpha \right\}, \quad \alpha > 0,
\end{equation}
as shown, for example, in \cite[Lemma 1]{Vershynin}. This implies that most of the probability mass lies within a shell of width $O(\sqrt{d})$ around the sphere of radius $\sqrt{d}$. Figure~\ref{fig:gaussian-concentration} illustrates this concentration phenomenon.

We consider two initialization strategies: one in which chains start at the mode $\theta = 0$, which lies well outside the concentration region; and one in which chains are initialized from the stationary distribution. Cold starts introduce substantial transient behavior. For HMC-type methods with fixed step size leapfrog integrators, controlling the energy error during this transient phase typically requires step sizes of order $O(d^{-1/2})$. By contrast, once the chain enters the concentration region, step sizes of order $O(d^{-1/4})$ are sufficient for mixing, as quantified for NUTS in \cite[Theorem 2]{BoOb2024}.  We emphasize that this is a sufficient condition for mixing and not a stability condition: leapfrog integration remains stable uniformly in $d$ as long as $h \le 2$ in this example.

To make this precise, let $\theta \in D_\alpha$ and $\rho \sim \mathcal{N}(0, I)$. Then the leapfrog energy error after $i$ steps satisfies
\begin{equation} \label{eq:deltaH}
\left| H\circ\Phi_h^i(\theta,\rho) - H(\theta,\rho) \right|
= \frac{1}{8} h^2 \left| |\Pi(\Phi_h^i(\theta,\rho))|^2 - |\theta|^2 \right| = O(h^2 d^{1/2}), \quad i \in \mathbb{Z} \;,
\end{equation}
where $\Pi$ is the projection on the position component.  This bound follows from the fact that the leapfrog integrator preserves the modified Hamiltonian $H_h(\theta,\rho) = H(\theta,\rho) - \tfrac{1}{8} h^2 |\theta|^2$, i.e., $H_h \circ \Phi_h \equiv H_h$, and that the leapfrog trajectory remains in a shell of comparable radius with high probability.

This estimate highlights two distinct regimes. When the chain is initialized at the mode, $|\theta|^2 = 0$, while $|\Pi(\Phi_h^i(\theta, \rho))|^2 = O(d)$ with high probability. As a result, the energy error is $O(h^2 d)$ during the transient phase. Once the chain reaches the concentration region $D_\alpha$, both $|\theta|^2$ and $|\Pi(\Phi_h^i(\theta, \rho))|^2$ are close to $d$, and their difference is only $O(d^{1/2})$ with high probability. This reduces the energy error to $O(h^2 d^{1/2})$, allowing for larger step sizes and more efficient sampling.

WALNUTS adapts to this variation in leapfrog energy error through its two-level step size mechanism. The macro step size $h$ is tuned for efficient mixing in the high-density shell, i.e., $h = O(d^{-1/4})$. The micro step size, on the other hand, is adjusted locally via the \texttt{micro} routine to ensure that the total energy error over each macro interval remains below a fixed threshold $\delta$. When initialized far from the concentration region, like the mode, WALNUTS selects smaller micro step sizes to compensate for elevated energy error. Once the orbit enters the high-probability shell, the micro step size increases accordingly. In this way, WALNUTS dynamically adapts its step size based on local energy errors, maintaining energy accuracy cheaply.

\paragraph{Numerical results for high-dimensional Gaussian}

The first set of simulations, summarized in Figure~\ref{fig:gauss-transient}, considers a setting in which NUTS and WALNUTS-R2P chains are initialized at the mode of a $d$-dimensional standard Gaussian target. In these simulations, NUTS is tuned for the transient phase using a step size $h=d^{-1/2}$, which, as discussed above, is optimal during the transient phase. WALNUTS-R2P is instead tuned for the stationary regime using a macro step size of  macro step size $h=d^{-1/4}$ (and $\delta=0.3$), corresponding to the optimal (fixed) step size in the stationary regime. From the upper set of plots in Figure~\ref{fig:gauss-transient}, it is seen that transient-tuned NUTS arrives at the mode slightly faster in terms of MCMC iterations than WALNUTS, with the gap growing modestly with increasing dimension.

The mid-level set of plots show the smallest $h\tilde{\ell}^{-1}$ used by WALNUTS in each orbit, which may be interpreted as the smallest micro step size used (modulus macro step size jittering and randomization of step size factors (R2P)) selected by \texttt{micro}.  As seen in the plots, WALNUTS uses smaller step sizes during the transient phase and gradually transitions to behavior similar to optimally tuned NUTS in the stationary phase.

The lower set of plots shows the cumulative number of gradient evaluations for WALNUTS and NUTS.  As expected, due to the lack of step size adaptivity for NUTS, the benefit of small fixed step size during the transient phase is soon overwhelmed by the cost of such small step sizes in the stationary phase.

Though it is seen that the initial WALNUTS step sizes are somewhat larger than $h^{-1/2}$ (indicated by black line in mid-level plots) for this choice of $\delta$, and thereby arriving at the high probability shell somewhat slower than NUTS, this choice of $\delta$ appears to be a reasonable when computational cost is taken into account. 



Figure~\ref{fig:gauss-ESS} shows the estimated effective sample size (ESS) per 1000 gradient evaluations for warm-start simulations targeting a standard Gaussian distribution, focusing on ESS for the scalar quantity $|\theta|^2$. All three methods—NUTS, WALNUTS-D, and WALNUTS-R2P—used the same macro step size $h = 1.4 \, d^{-1/4}$. For this relatively simple target, which is well-suited to fixed step size HMC methods like NUTS, WALNUTS-R2P exhibits somewhat lower efficiency, primarily due to the added cost of frequently applying higher simulation fidelity (i.e., using $\ell = \tilde{\ell} + 1$). WALNUTS-D and NUTS show comparable efficiency overall. At higher dimensions, WALNUTS-D appears to slightly outperform NUTS in some cases, though the Monte Carlo variability is too large to support a definitive conclusion. The step-like pattern in the results reflects dimension-dependent shifts in the number of doublings used by the no-U-turn condition.

To summarize these numerical results, WALNUTS demonstrates substantial promise in two key respects. First, it effectively handles transient regimes without getting stuck near the initialization point, avoiding the cold-start issues that can hinder other methods. Second, for target distributions where standard HMC-based methods like NUTS are known to perform well, WALNUTS achieves comparable efficiency. These results suggest that WALNUTS combines robustness to initialization with competitive performance in well-conditioned, high-dimensional settings.

\subsection{Neal's Funnel}

\label{sec:funnel}

\begin{figure}
\centering
\begin{tikzpicture}[scale=0.85]

\begin{scope}[xshift=0cm, yshift=3.5cm, scale=0.8]
\draw[line width=1pt] plot[samples=100,domain=-3.01:3.01] ({-e^(\x/2)},\x);
\draw[line width=1pt] plot[samples=100,domain=-3.01:3.01] ({e^(\x/2)},\x);
\draw[line width=0.7pt] ({-exp(3/2)},3.01) arc (270:130:0.149) coordinate[pos=1] (end1) {};
\draw[line width=0.7pt] ({exp(3/2)},3.01) arc (-90:40:0.149) coordinate[pos=1] (end2) {};
\draw[line width=0.7pt] (end1) -- (end2);
\draw[line width=0.7pt] ({-exp(-3/2)},-3) arc (-170:10:0.226);
\node[left] at (-0.5,-2.5) { Neck};
\node[left] at (-3.2,2) { Mouth};
\draw[->] (2.5,-2.5) -- (3.7,-2.5) node[pos=1,below right] { $x$};
\draw[->] (2.5,-2.5) -- (2.5,-1.3) node[pos=1,above left] { $\omega$};

\draw[ultra thick, gray] ({-exp(-2.0/2)}, -2.0) -- ({exp(-2.0/2)}, -2.0);
\draw[ultra thick, gray] ({-exp(-2.0/2)}, -2.2) -- ({-exp(-2.0/2)}, -1.8);
\draw[ultra thick, gray] ({exp(-2.0/2)}, -2.2) -- ({exp(-2.0/2)}, -1.8);

\draw[ultra thick, gray] ({-exp(2.0/2)}, 2.0) -- ({exp(2.0/2)}, 2.0);
\draw[ultra thick, gray] ({-exp(2.0/2)}, 1.8) -- ({-exp(2.0/2)}, 2.2);
\draw[ultra thick, gray] ({exp(2.0/2)}, 1.8) -- ({exp(2.0/2)}, 2.2);

\fill[red!20, opacity=0.4, draw=red, dashed]
  plot[samples=100, domain=1.8:2.2]
  ({-exp(\x/2)}, \x)
  -- plot[samples=100, domain=2.2:1.8]
  ({exp(\x/2)}, \x)
  -- cycle;


  \fill[red!20, opacity=0.4, draw=red, dashed]
    plot[samples=100, domain=1.8:2.2]
    ({-exp(\x/2) * exp(-2)}, {\x - 4})
    -- plot[samples=100, domain=2.2:1.8]
    ({exp(\x/2) * exp(-2)}, {\x - 4})
    -- cycle;

\end{scope}

\begin{scope}[xshift=7.cm, yshift=4.5cm]
\begin{axis}[
    width=6cm,
    height=4cm,
    title={},
    xlabel={ $\omega$},
    ylabel={ $\ln(\lambda)$},
    domain=-6:6,
    samples=100,
    grid=major,
    ymin=-3,
    ymax=4,
    xmin=-4,
    xmax=4,
    xtick={-4,-2,0,2,4},
    ytick={0,2,4},
    tick label style={font=},
    title style={font=\small},
    label style={font=},
]
\addplot[blue, thick=0.8pt, smooth] {1/9 > exp(-x) ? ln(1/9) : -x};
\end{axis}
\end{scope}

\begin{scope}[xshift=7.cm, yshift=0.5cm]
\begin{axis}[
    width=6cm,
    height=4cm,
    title={},
    xlabel={ $\omega$},
    ylabel={ $\ln(\kappa)$},
    domain=-6:6,
    samples=100,
    grid=major,
    ymin=0,
    ymax=6,
    xmin=-4,
    xmax=4,
    xtick={-4,-2,0,2,4},
    ytick={0,2,4,6},
    tick label style={font=},
    title style={font=\small},
    label style={font=},
]
\addplot[red, thick=0.8pt, smooth] {ln(9) + abs(x)};
\end{axis}
\end{scope}

\end{tikzpicture}
\caption{\textbf{Curvature, Conditioning, and Scale Invariance in Neal's Funnel.} Illustration of Neal's funnel geometry (left), along with the spectral radius $\lambda(\omega) = \max(1/9, e^{-\omega})$ (top right; see~\eqref{eq:spectral-radius}) and the condition number $\kappa(\omega) = 9 \cdot \max(e^{\omega}, e^{-\omega})$ (bottom right; see~\eqref{eq:condition-number}) plotted as functions of the funnel axis variable $\omega$. The neck and mouth regions have the same shape up to scale: the width at level $\omega$ is proportional to $ e^{\omega/2}$, with relative rate of change $\frac{d}{d\omega} \text{width}(\omega) = \frac{1}{2} \cdot \text{width}(\omega)$. The translucent bands show a mouth segment and its horizontally rescaled copy overlaid on the neck, highlighting the funnel’s local scale invariance: each horizontal slice has the same shape up to scale. This structure makes it difficult for methods such as ensemble samplers to use local geometric cues to guide movement into or out of the neck.   }
\label{fig:funnel-geometry}
\end{figure}
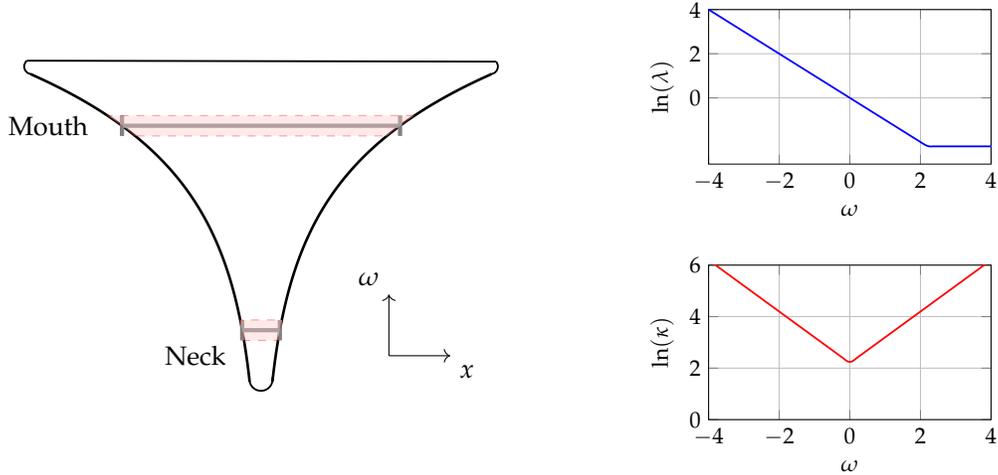

\begin{figure}
    \centering
    \includegraphics[width=0.99\linewidth]{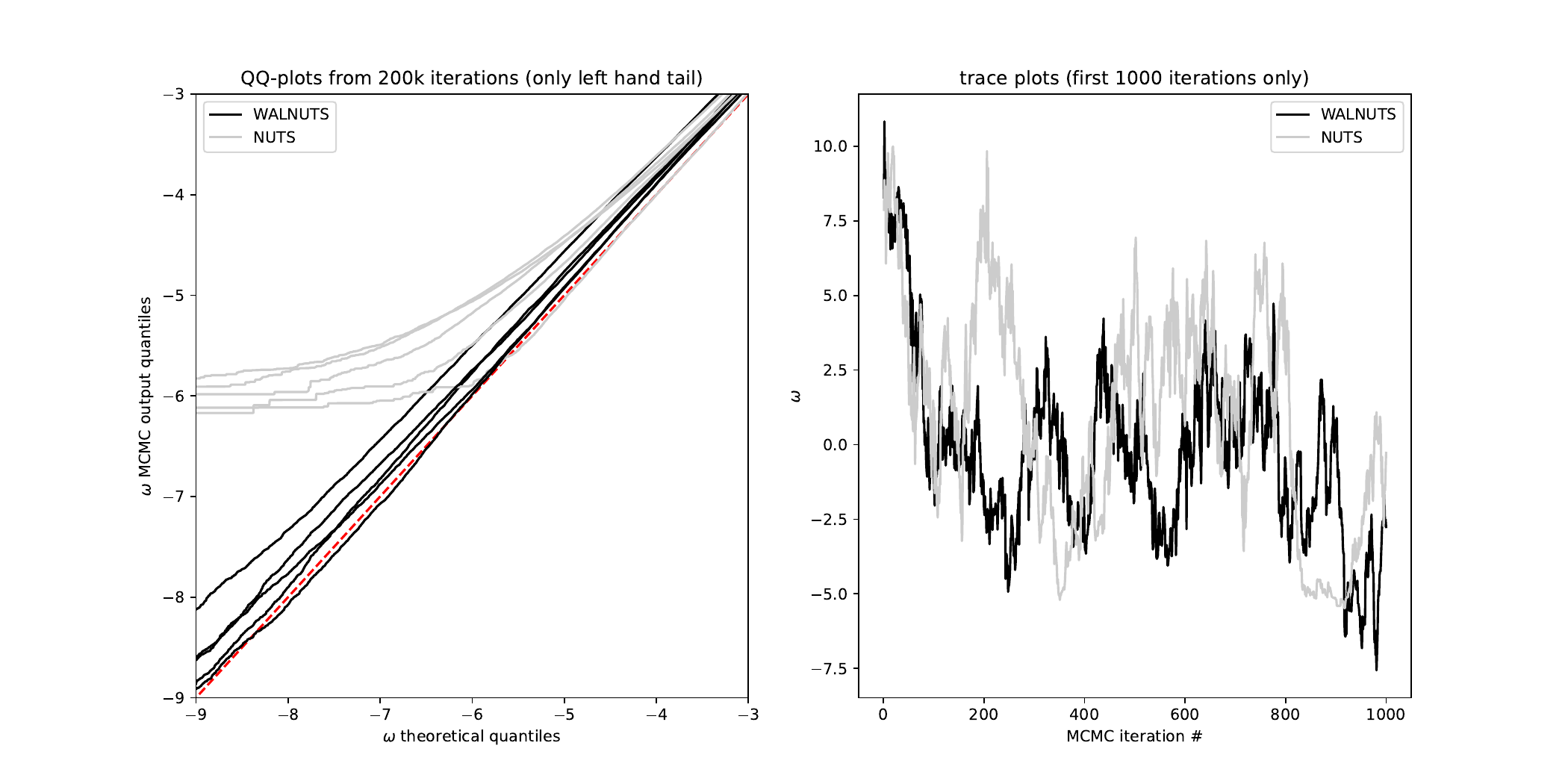}
    \caption{\textbf{Tail accuracy and mixing in Neal’s funnel: WALNUTS vs. NUTS} The left panel shows QQ-plots (restricted to the far left tail) for WALNUTS-R2P (black) and NUTS (gray), each based on 200{,}000 iterations. WALNUTS was run with tuning parameters $\delta = 0.21$ and $h = 0.36$, while NUTS used a step size $h_{\text{NUTS}} = 0.11$, chosen so that the average orbit length matched the number of gradient evaluations used by WALNUTS. With comparable computational cost (NUTS used 104\% of WALNUTS's gradient evaluations), WALNUTS successfully explores the full support of the target, while NUTS fails to reach deep into the neck of the funnel. The right panel shows trace plots of $\omega$, illustrating that both samplers mix very slowly. This poor mixing explains the variability in the QQ-plots, particularly in the far tail.}
    \label{fig:funnel-qq}
\end{figure}

\begin{figure}
    \centering
    \includegraphics[width=0.99\linewidth]{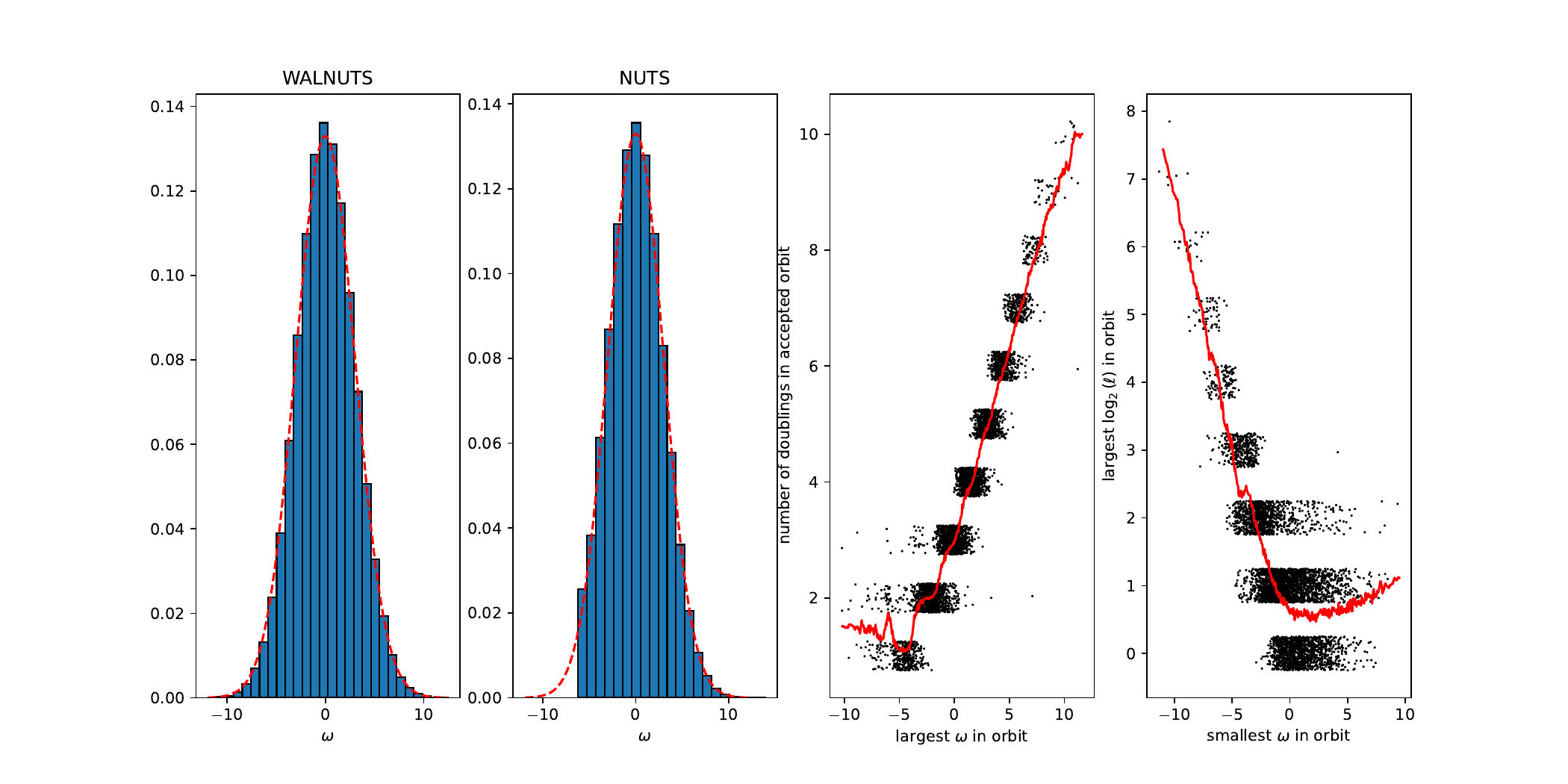}
    \caption{\textbf{Diagnostics for WALNUTS and NUTS on Neal’s Funnel.}  Various diagnostics for WALNUTS-R2P (auto-tuned with $\delta = 0.21$, macro step size $h = 0.36$, 1 million iterations) applied to Neal’s funnel distribution \eqref{eq:funnel_distr}. The left panel shows a histogram of sampled $\omega$ values overlaid with the true $N(0,9)$ density. The second panel shows the corresponding histogram for NUTS, using 104\% of WALNUTS’s total gradient evaluations. The third panel plots (thinned and jittered) the number of doublings in accepted WALNUTS orbits versus the largest $\omega$ visited during each orbit. The right panel plots (thinned and jittered) the maximum number of micro step size doublings ($\log_2(\ell)$) within each orbit versus the smallest $\omega$ visited. In both rightmost panels, nonparametric KNN regression curves are overlaid in red.
    }
    \label{fig:funnel-diagnositics}
\end{figure}

\begin{figure}
    \centering
    \includegraphics[width=0.99\linewidth]{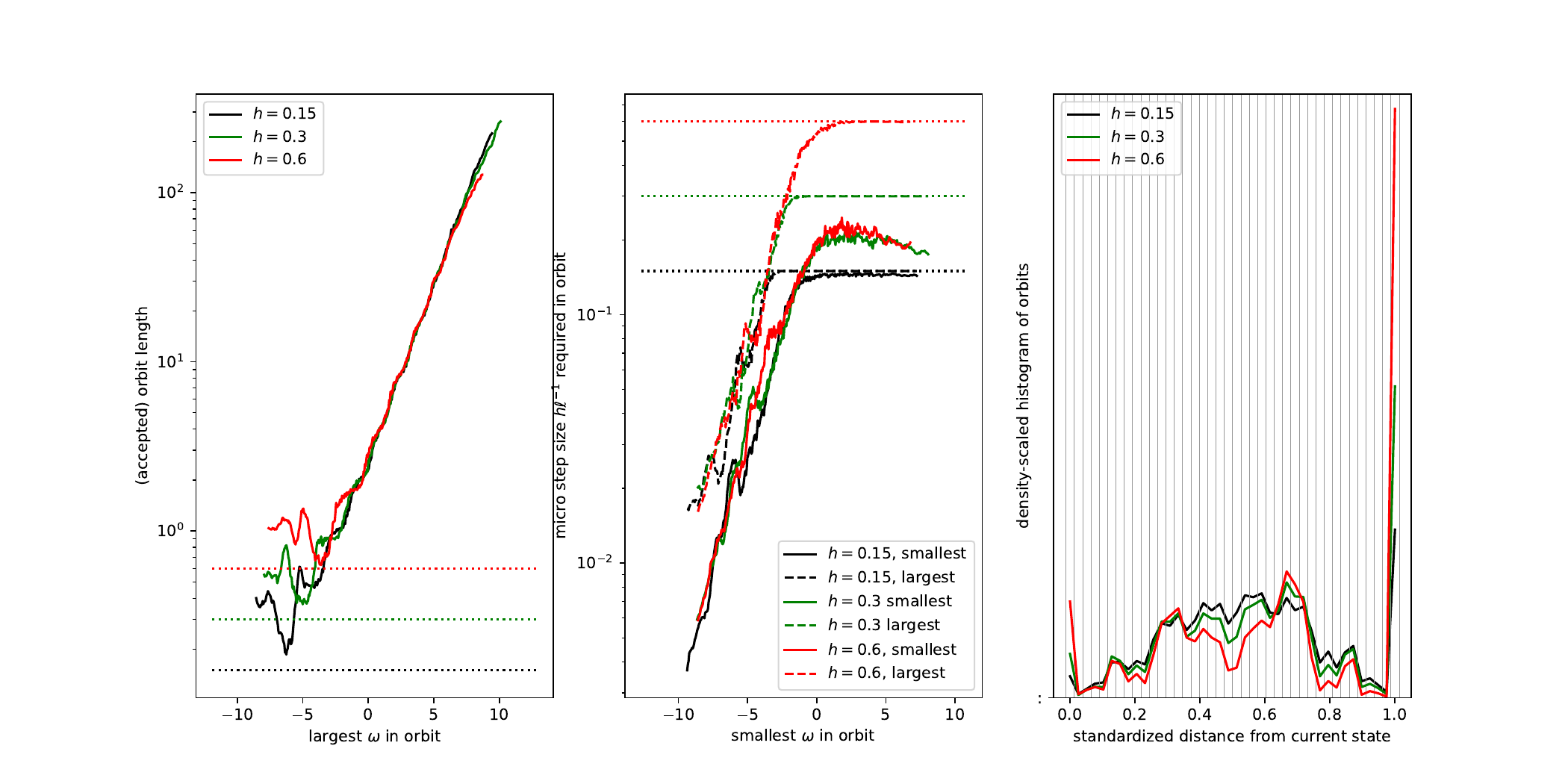}
    \caption{ \textbf{Influence of macro-step size on WALNUTS in Neal’s Funnel.} This figure illustrates the effect of the macro step size $h$ on WALNUTS behavior when applied to Neal's funnel distribution~\eqref{eq:funnel_distr}, based on 50{,}000 MCMC iterations for each value of $h$.
The left panel shows KNN regression curves of orbit (integration time) length as a function of the largest $\omega$ encountered in the orbit, for three macro step sizes ($h = 0.15$, $0.3$, $0.6$; shown as dotted lines), with $\delta = 0.1$ fixed throughout. Orbit lengths are broadly similar across values of $h$, except in the left-hand tail, where many accepted orbits consist of a single integration step.
The middle panel shows KNN regression curves of the largest and smallest micro step sizes $h\ell^{-1}$ used in each accepted orbit, plotted against the smallest $\omega$ visited. For large $\omega$, smaller macro step sizes act as a ceiling on $h\ell^{-1}$, limiting step size even when larger values would be allowed by the energy error threshold.
The right panel shows histograms of the relative time distance between the initial and selected states in each accepted orbit. The distributions are broadly consistent with the idealized triangular law from biased progressive HMC.  }
    \label{fig:funnel-mul-step-size}
\end{figure}

\begin{figure}
    \centering
    \includegraphics[width=0.99\linewidth]{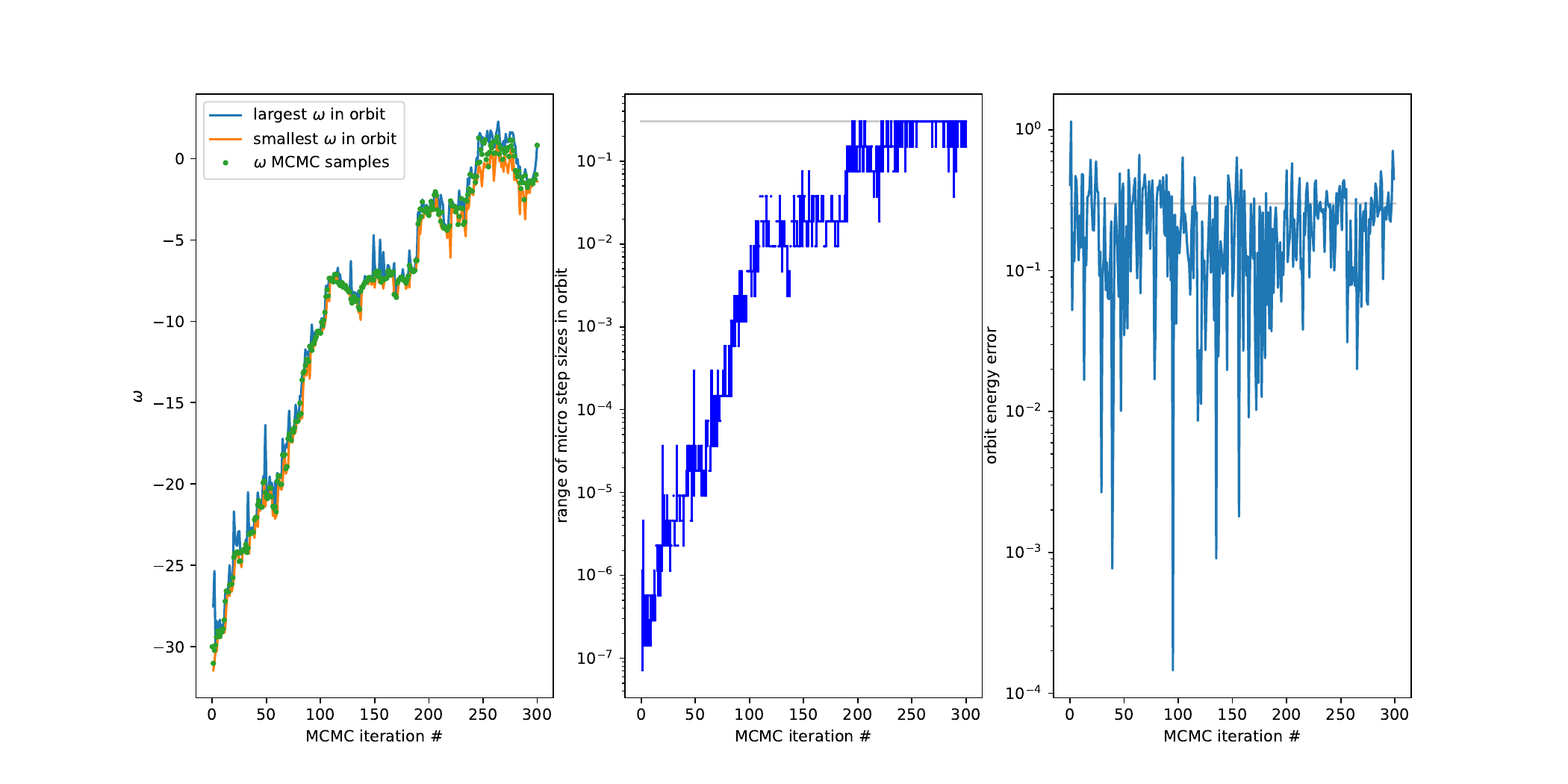}
    \caption{ \textbf{WALNUTS Diagnostics from a Cold Start Deep in the Neck of Neal’s Funnel.} The left panel shows the first 300 WALNUTS samples of $\omega$ from a run on Neal’s funnel distribution, initialized in a cold state deep in the neck at $\omega = -30$, with $x_i = 0$ for $i = 1, \dots, 10$. Green dots indicate sampled values; solid lines show the $\omega$-range of each orbit. The middle panel displays the range of micro step sizes $h\ell^{-1}$ used per orbit, with the macro step size $h = 0.3$ shown as a horizontal grey line.
The right panel shows the total energy error per orbit, defined as the difference between the largest and smallest Hamiltonian values. The local energy error tolerance $\delta = 0.3$ is indicated by a dashed line.}
    \label{fig:funnel-transient}
\end{figure}

\paragraph{Geometry of Neal's funnel}

Consider the $(d+1)$-dimensional funnel distribution introduced by Neal \cite{Neal2003Slice}, defined over $\theta = (\omega, x)$ with $\omega \in \mathbb{R}$ and $x = (x_1, \dots, x_d) \in \mathbb{R}^d$, where
\begin{equation}
    \omega \sim \N(0,9),\quad x_i\mid \omega \sim \N\left(0,{\exp(\omega)}\right),\;i=1,\dots,d. \label{eq:funnel_distr}
\end{equation}

The joint distribution $\mu(\omega, x)$ exhibits a funnel-like geometry, with $\omega$ forming the funnel axis  and may be thought of as a model problem for Bayesian hierarchical models. Several authors \cite{betancourt2013hamiltonianmontecarlohierarchical,kleppe2022connecting,BouRabeeCarpenterKleppeMarsden2024} have found that HMC-like methods with fixed step sizes may fail to properly explore such target distributions. 

This difficulty stems from the model’s extreme variation in scale: for large positive $\omega$, the conditional variances $\operatorname{Var}(x_i \mid \omega) = e^{\omega}$ are large and the target is wide and diffuse---this region corresponds to the mouth of the funnel.  Conversely, for large negative $\omega$, the conditional variances shrink exponentially, concentrating the mass near a narrow neck. These features are illustrated in Figure~\ref{fig:funnel-geometry}, where the funnel's ``wide mouth'' and ``narrow neck'' encode this dramatic variation in scale.

The challenge of sampling from Neal's funnel can be quantified in terms of the curvature of the potential energy landscape associated with the target distribution.  The potential energy is given by $U(\omega, x) = -\log \mu(\omega, x)$, and its curvature varies drastically across regions of the domain. In particular, the Hessian evaluated along the funnel axis at $(\omega, 0)$ is diagonal:
\[
D^2 U(\omega, 0) = \operatorname{diag}\left( \tfrac{1}{9}, e^{-\omega}, \dots, e^{-\omega} \right) \;.
\]
The corresponding \emph{spectral radius}, defined as the largest eigenvalue of the Hessian, is given by
\begin{equation} \label{eq:spectral-radius}
    \lambda(\omega) := \max \left( \tfrac{1}{9}, e^{-\omega} \right) \;.
\end{equation}
This implies that the local Lipschitz constant $L$ of $\nabla U$ scales with $\lambda(\omega)$, and the leapfrog integrator is numerically stable only if
\[
h \leq \frac{2}{\sqrt{L}} \leq 2 e^{-\omega/2} \;.
\]
As a result, exploring the neck region ($\omega \ll 0$) requires exponentially small step sizes to maintain stability, while larger step sizes are needed to mix efficiently in the mouth region ($\omega \gg 0$).

Furthermore, the \emph{condition number} of the Hessian matrix evaluated along the funnel axis (i.e., the ratio of the largest to smallest eigenvalue of $D^2 U(\omega,0)$) is given by
\begin{equation} \label{eq:condition-number}
    \kappa(\omega) := 9 \cdot \max(e^{\omega}, e^{-\omega}) \;,
\end{equation}
which grows exponentially in $|\omega|$. This makes the funnel an extreme example of an ill-conditioned target. Any fixed-step-size implementation of NUTS must adopt a small enough leapfrog step size to avoid numerical divergences in the neck, thereby incurring significant inefficiency in the mouth. Conversely, if tuned for the mouth, such a method becomes numerically unstable or inaccurate in the neck. Since both regions carry significant posterior mass, this trade-off cannot be ignored.

These challenges make Neal’s funnel a rigorous test for samplers that adapt to local geometry. WALNUTS addresses this issue by dynamically adjusting the step size \emph{within} each leapfrog path. Specifically, WALNUTS first chooses a macro step size $h$ that is suitable for the well-conditioned regions of the target (e.g., the mouth of the funnel), allowing rapid progress when the target is flat. Within each macro step, however, WALNUTS reduces the micro step size until the local energy error falls below the user-specified tolerance $\delta$. This gives rise to a sequence of micro step sizes, each of which may vary depending on the local curvature along the orbit. As a result, the variable step size leapfrog integrator within WALNUTS automatically uses finer resolution when passing through stiff regions such as the narrow neck (where stability conditions are more restrictive), and coarser resolution in flatter regions where larger steps are feasible (see Figure~\ref{fig:walnuts_intro}).

\medskip

\begin{remark}[Mode of Neal's Funnel]
\label{rmk:funnel-mode}
Analogously to the case of a high-dimensional Gaussian, the potential energy $U(\omega, x) = -\log \mu(\omega, x)$ attains its minimum at a mode that lies in a region of extremely low posterior mass. Completing the square in $U$, we find that the unique minimizer is given by
\[
(\omega^*, x^*) = \left( -\tfrac{9d}{2}, 0 \right) \;.
\]
Indeed, expanding the expression for $U$,
\[
U(\omega, x) = \frac{\omega^2}{18} + \frac{\|x\|^2}{2e^{\omega}} + \frac{d}{2} \omega + \text{const} \;,
\]
we see that $U$ is minimized at $x = 0$ and $\omega = -\tfrac{9d}{2}$. This mode lies deep in the neck of the funnel, far in the lower tail of the marginal $\omega \sim \mathcal{N}(0, 9)$, and hence, in a region with negligible posterior mass. Much like the mode of a high-dimensional standard Gaussian, it is not representative of typical samples from the target.
\end{remark}

\paragraph{Numerical results for Neal's funnel}

Figure~\ref{fig:funnel-qq}, left panel, shows QQ-plots of the marginal distribution of $\omega$ based on MCMC output using either WALNUTS-R2P or NUTS, with comparable computational budgets. WALNUTS provides a significantly more accurate representation of the far left tail, demonstrating its ability to explore high-curvature regions deep in the neck of the funnel. However, both methods exhibit slow mixing, as evidenced by the trace plots in the right panel of Figure~\ref{fig:funnel-qq}. Notably, the poor mixing persists across the entire support of the target distribution, suggesting that the orbit length selection mechanism in NUTS (and inherited by WALNUTS) tends to produce trajectories that are too short on average.  Addressing this issue is an important direction for future work.

Figure~\ref{fig:funnel-diagnositics} presents diagnostic information from a warm-start run of WALNUTS applied to the funnel distribution \eqref{eq:funnel_distr}, using $1$ million MCMC iterations with auto-tuned parameters $\delta = 0.21$ and macro step size $h_0 = 0.36$. The far-left panel shows a histogram of sampled $\omega$ values, which indicates no apparent difficulty exploring the left-hand tail of the distribution (see also Figure~\ref{fig:funnel-qq}). For comparison, the second panel from the left shows the corresponding histogram for NUTS, run at similar computational cost (104\% of the gradient evaluations used by WALNUTS), which reveals a clear failure to explore the left-hand tail.

The second panel from the right displays the number of orbit doublings used in each accepted orbit, plotted against the largest value of $\omega$ visited during that orbit. This shows that WALNUTS adapts by generating longer trajectories in the broad mouth of the funnel (i.e., regions with large $\omega$), where the target is flatter and longer integration times are beneficial.

The rightmost panel of Figure~\ref{fig:funnel-diagnositics} shows the largest number of micro step size halvings (i.e., $\log_2(\ell)$) used within each orbit, plotted against the smallest value of $\omega$ visited in that orbit. It is evident that WALNUTS automatically selects smaller step sizes in the narrow neck of the funnel (i.e., regions with small $\omega$), where stability constraints are more restrictive, while using larger step sizes in the broad mouth where the geometry is more forgiving.

Figure~\ref{fig:funnel-mul-step-size} examines how the macro step size $h$ influences orbit (integration time) length, micro step sizes, and the potential accumulation of local integration errors. The left panel shows KNN regressions of orbit length as a function of the largest $\omega$ visited within each orbit. For large values of $\omega$, all considered macro step sizes result in similar orbit lengths. In the narrow neck of the funnel, orbit lengths vary because they often consist of only a single integration step, and are therefore bounded below by the macro step size.

The middle panel of Figure~\ref{fig:funnel-mul-step-size} shows KNN regressions of the largest and smallest micro step sizes used by \texttt{micro} in each orbit, plotted as a function of the smallest $\omega$ visited within that orbit. For small values of $\omega$, all macro step sizes yield similar results, since the upper bound imposed by the macro step size on the micro step sizes is rarely active. As $\omega$ increases, however, the results begin to diverge because the macro step size increasingly constrains the allowable micro step sizes.

The rightmost panel of Figure~\ref{fig:funnel-mul-step-size} presents histograms of the relative orbit time distance between the initial state $(\theta_0, \rho_0)$ and the selected state $(\tilde\theta, \tilde\rho)$, normalized by the total accepted orbit time. In the idealized setting of biased progressive HMC with exact Hamiltonian simulation and fixed orbit length, this statistic follows a symmetric triangular distribution centered at $0.5$ (see Theorem~\ref{thm:law_of_iprime}). The empirical histograms are broadly consistent with this behavior and do not exhibit signs of degeneracy. The peaks at standardized distances 0 and 1 are attributable to orbits consisting of a single integration step, which occur more frequently as the macro step size increases.

Figure~\ref{fig:funnel-transient} shows a single cold-start simulation consisting of 300 WALNUTS MCMC iterations, initialized at $\omega = -30$ (i.e., 10 standard deviations below the mean) with $x_i = 0$ for $i = 1, \dots, 10$.
The left panel displays the MCMC samples as green dots, along with the extent of each orbit in the $\omega$-direction. The sampler has no difficulty moving even through the highly ill-conditioned regions encountered early in the simulation.
The middle panel shows the range of micro step sizes used in each orbit, again indicating that WALNUTS adapts as expected.  The right panel plots the total energy error per orbit, defined as the difference between the largest and smallest Hamiltonian values, suggesting that energy errors do not exhibit problematic constructive accumulation.

The choice of initial configuration is motivated by the fact that the funnel distribution has a single mode at $\omega = -45$ and $x_i = 0$ for all $i$ (see Remark~\ref{rmk:funnel-mode}). While there is no theoretical obstacle to starting the chain at the mode, extrapolation from Figure~\ref{fig:funnel-transient} suggests that doing so would require micro step sizes near machine precision and would entail considerable computational effort.

This example supports the conclusion that the proposed locally adaptive step size procedure functions as intended, successfully identifying regions that require smaller step sizes. As discussed above, however, orbit-length selection remains a challenge for this model: NUTS-based strategies lead to slow mixing, and WALNUTS is not immune to this issue. While this suggests that further refinement of orbit-length selection methods, particularly for multiscale target distributions, could improve performance, such developments are beyond the scope of the present work---more on this point in Section~\ref{sec:conclusion} below.

\subsection{The Stock-Watson model}

\label{sec:stock-watson}

\begin{figure}
    \centering
    \includegraphics[width=0.45\linewidth]{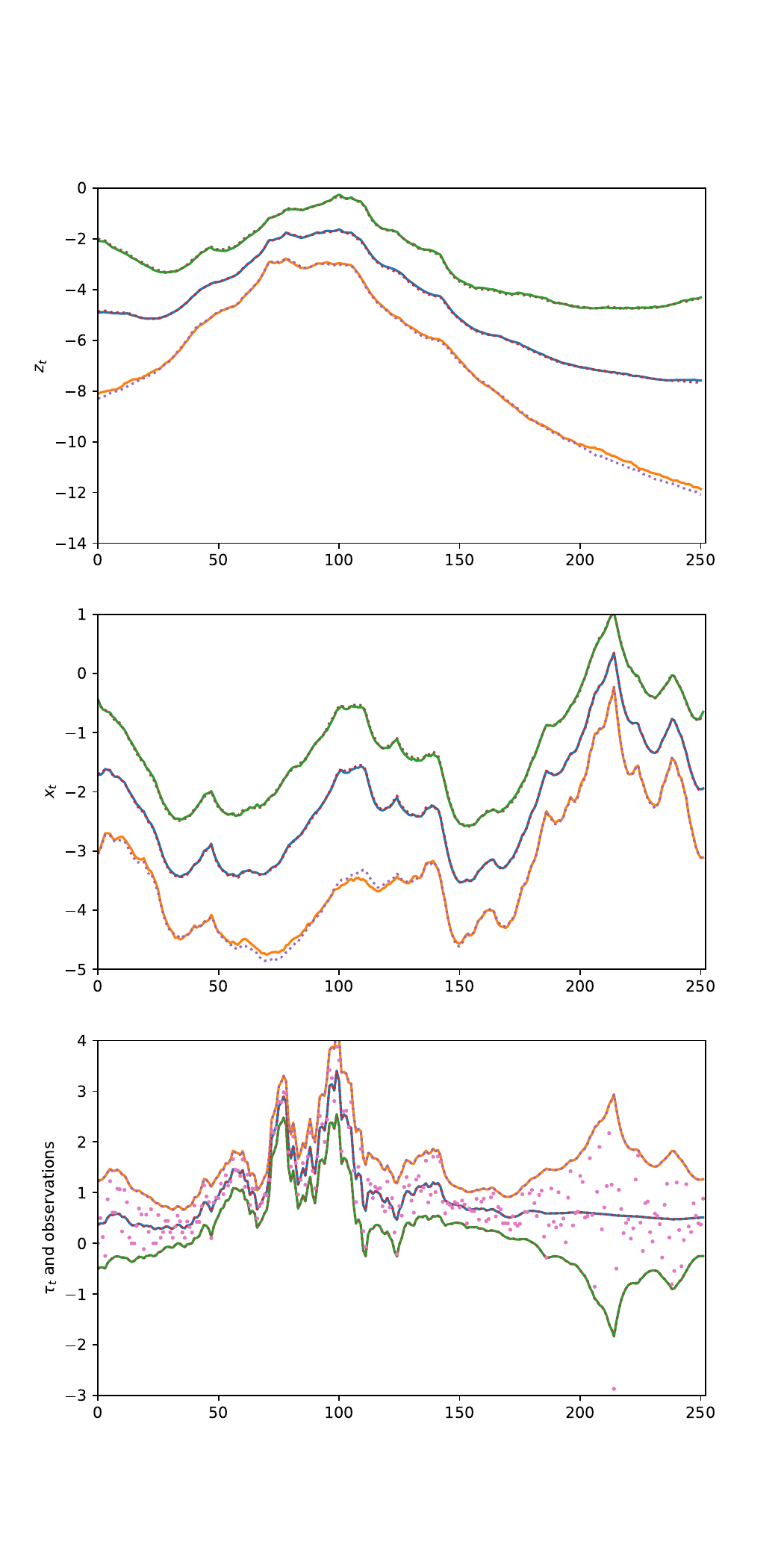}
    \caption{\textbf{Posterior distributions of the latent factors $(z,x,\tau)$ under the Stock Watson model.} The top panel shows the 0.05, 0.5, and 0.95 quantiles of the persistent volatility factor $z_t$: solid lines correspond to WALNUTS, and dotted lines to NUTS. The middle panel presents the same quantiles for the transient volatility factor $x_t$. The bottom panel displays the posterior median of $\tau_t$ along with the bands $\tau_t \pm 2\exp(0.5 x_t)$ (medians plotted), representing a scale-aware uncertainty interval. Observations $y_t$ are shown as dots.
Results from WALNUTS-D are nearly identical and omitted for clarity.
    }
    \label{fig:sw-factors}
\end{figure}

\begin{figure}
    \centering
    \includegraphics[width=0.99\linewidth]{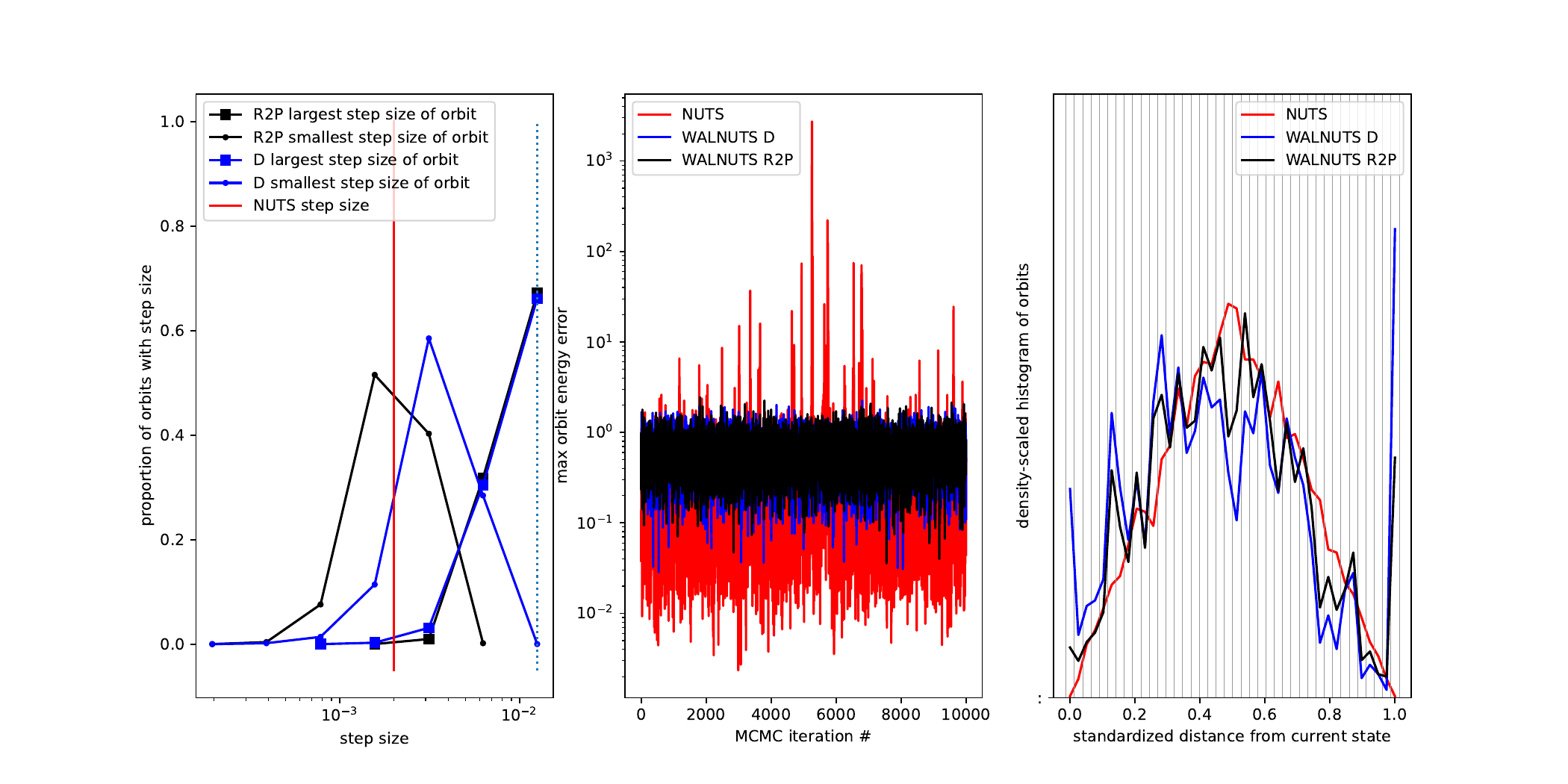}
    \caption{\textbf{Distribution of step sizes and orbit energy errors for the Stock-Watson model.}  The left panel shows the distribution of the largest (dots) and smallest (squares) micro step sizes used in each WALNUTS orbit (both R2P and D variants). The red vertical line indicates the NUTS step size, while the dotted line marks the largest micro step size considered by WALNUTS ($h/8$). The middle panel presents the total energy error per orbit, computed as $\max_j H(\theta^{(j)}, \rho^{(j)}) - \min_j H(\theta^{(j)}, \rho^{(j)})$, for NUTS and both versions of WALNUTS. Computational costs are comparable between WALNUTS-R2P and NUTS, while WALNUTS-D requires approximately 60\% of the gradient evaluations of the former two. The right panel displays the integration time statistic normalized by orbit length. These distributions do not raise any concerns. Note that the smaller step size used by NUTS results in more frequent orbit doublings, which explains the more triangular-like distribution of this statistic.}
    \label{fig:sw-diag}
\end{figure}

As a real-world example, we consider the Stock and Watson inflation rate model \citep{doi:10.1111/j.1538-4616.2007.00014.x}. Quarterly US inflation rates $y_t,\;t=1,\dots,T$ are modeled using three latent, nonlinearly coupled random walk factors $(z_t, x_t, \tau_t)$ as follows:
\begin{eqnarray}
   z_t|z_{t-1},\sigma &\sim& \N(z_{t-1},\sigma^2),\; t=2,\dots,T-1,\\
   x_t|x_{t-1},\sigma &\sim& \N(x_{t-1},\sigma^2),\; t=2,\dots,T,\\
   \tau_t|\tau_{t-1},z_{t-1} &\sim& \N(\tau_{t-1},\exp(z_{t-1})),\;t=2,\dots,T,\\
   y_{t}|\tau_t,x_t &\sim& \N(\tau_t,\exp(x_t)),\;t=1,\dots,T.
\end{eqnarray}
The description of the model is completed with the prior $\sigma^{-2} \sim \mathrm{Gamma}(5, 0.5)$, using the rate parameterization. The dataset used is the same as in \cite{kleppe2022connecting}, consisting of $T=252$ observations between 1955 (first quarter) and 2018 (first quarter). Figure \ref{fig:sw-diag} provides posterior quantiles of the three latent processes, and it is seen that the balance between persistent volatility (i.e. $\exp(z_t)$) and transient volatility (i.e. $\exp(x_t)$) shows substantial temporal variation \citep[see ][who first pointed out these effects]{doi:10.1111/j.1538-4616.2007.00014.x}.

Due to the nonlinear coupling of the latent factors, the posterior distribution of this model is known to exhibit substantial non-Gaussian structure, including pronounced funnel-like geometries.  To facilitate sampling, we reparameterize the model in terms of the innovations of the latent processes: $([z_1,(z_2-z_1)\sigma^{-1},\dots,(z_{T-1}-z_{T-2})\sigma^{-1}],[x_1,(x_{2}-x_1)\sigma^{-1},\dots,(x_{T}-x_{T-1})\sigma^{-1}],[\tau_1,(\tau_2-\tau_1)\exp(-0.5z_1),\dots,(\tau_T-\tau_{T-1})\exp(-0.5z_{T-1})],\log \sigma^2)$.  This transformation yields a posterior that can be sampled using an identity mass matrix. However, it is important to emphasize that the transformation does not eliminate the underlying nonlinearities. In particular, this model with this parameterization is known to require very small step sizes when using fixed step size HMC, motivating more elaborate preconditioning strategies as explored in \cite{kleppe2019}.

Posterior summaries from runs of 10,000 iterations of NUTS (with fixed step size $h = 0.002$) and WALNUTS (with local error tolerance $\delta = 0.3$, macro step size $h = 0.1$, and a minimum of 8 micro-steps per macro-step, enforced by starting the loop in \texttt{micro} at index 3 instead of 0) are shown in Figure~\ref{fig:sw-factors}. These parameter settings result in approximately equal computational cost. At first glance, the posterior estimates produced by NUTS and WALNUTS appear similar.

However, Figure~\ref{fig:sw-diag} reveals important differences in the energy error diagnostics. The energy errors from NUTS exhibit a substantial number of large deviations (e.g., greater than 2), indicating that many transitions diverge. This suggests that some regions of posterior support are effectively inaccessible to NUTS at this step size, even though such pathologies are not evident from the trace plots in Figure~\ref{fig:sw-factors}. A practitioner encountering such behavior would rightly be concerned about convergence.

In contrast, WALNUTS shows no such issues: energy errors remain stable across iterations, even though only local error control is enforced. The left panel of Figure~\ref{fig:sw-diag} shows the distribution of the largest and smallest micro step sizes used by the two variants of WALNUTS.
\emph{Remarkably, despite frequently using step sizes larger than NUTS, WALNUTS maintains stable energy errors by selectively using smaller steps only where needed!}

\section{Conclusion \& Outlook}

\label{sec:conclusion}

WALNUTS introduces a robust and flexible locally adaptive variant of HMC that performs well in difficult sampling problems where state-of-the-art methods like NUTS often struggle. Despite its locally adaptive nature, WALNUTS is not significantly more expensive than NUTS for well-conditioned targets, especially when warmup is appropriately tuned. When NUTS works, it works; but WALNUTS is designed for when it doesn’t. It offers a principled fallback that expands the range of models that can be reliably sampled without manual tuning or aggressive preconditioning.  In the near term, WALNUTS can be integrated into probabilistic programming languages such as Stan, PyMC, or NumPyro, where its plug-and-play adaptivity and robustness offers practical benefits with minimal user tuning. 

Importantly, the error-controlled adaptive integration scheme at the heart of WALNUTS is not specific to NUTS. It is a general-purpose tool that can be incorporated into a wide variety of HMC-type methods, including standard HMC, randomized HMC \cite{BoSa2017,Deligiannidis2021,BoEb2022} (and other Hamiltonian-based PDMPs \cite{chevallier2025towards}), kinetic Langevin samplers \cite{BouRabeeOberdoerster2024}, and generalized HMC variants \cite{turok2024sampling}. The approach we present here also opens up new directions for developing Hamiltonian-based samplers that remain robust and efficient even in challenging inference settings where the local geometry of the target distribution varies substantially. Below, we outline several promising directions for future work.

\begin{itemize}
\item \textbf{Mass matrix adaptivity.}
A natural extension is to introduce local adaptivity into the mass matrix, allowing it to vary with position to better capture heterogeneous curvature in the target distribution. This builds on ideas from Riemannian HMC \cite{GiCa2011} and recent advances in adaptive mass matrix tuning \cite{kleppe2016adaptive,whalley2024randomized,hird2023quantifying,tran2024tuning}, while aiming to retain the simplicity and computational efficiency of WALNUTS by avoiding expensive full Hessian evaluations.

\item \textbf{Ensemble-based approaches.}
Another approach to preconditioning WALNUTS is to run an ensemble of chains, following the strategy of affine-invariant methods \cite{foreman2013emcee,goodman2010ensemble}, including the recently introduced affine-invariant HMC variant \cite{chen2025new}. These methods leverage the positions of parallel chains to adapt proposals to the geometry of the target distribution.  Global affine invariance is not sufficient to handle highly multiscale or hierarchical targets such as Neal’s funnel where local structure can vary dramatically across regions. Still, ensemble-based strategies are promising. They provide a natural way to estimate low-rank structure and local scaling, which might be hard to detect with a single locally adaptive chain. Combining these techniques with WALNUTS' local step size adaptation could yield samplers that better capture local geometry.

\item \textbf{Improving the U-Turn Diagnostic.}
Another important direction is to revisit the U-turn diagnostic, which remains a heuristic inherited from standard NUTS. A more principled approach to deciding when and how to truncate orbit construction could significantly improve both robustness and sampling efficiency. Indeed, Figure~\ref{fig:funnel-qq} illustrates the limitations of the current U-turn-based criterion: the left panel shows that WALNUTS more accurately captures the far left tail of the marginal distribution of $\omega$, demonstrating its ability to explore high-curvature regions deep in the neck of the funnel. However, the trace plots in the right panel reveal slow mixing for both WALNUTS and NUTS, with poor performance persisting across the entire support. This suggests that the orbit length selection mechanism often terminates integration prematurely. Addressing this issue by developing other diagnostics than the U-turn one is a key direction for future work.

\end{itemize}

Taken together, these directions suggest the potential for a new class of locally adaptive HMC methods that offer improved robustness and efficiency across a broad range of challenging scenarios, including models with extreme anisotropy, funnel geometries, or high dimensionality. They point toward a new generation of adaptive MCMC algorithms that more fully exploit local geometry, without relying on global approximations or heavy precomputations, and that remain effective across a diverse range of model classes and data regimes.

\printbibliography

\clearpage

\appendix
\section{Pseudocode Implementation of WALNUTS} \label{sec:pseudocode}

\begin{mylisting}
    \begin{flushleft}
$\texttt{\bfseries WALNUTS}\!\left(\theta, \mu,     M, h, m_{\max}, \delta \right)$
    \vspace*{2pt}    \vspace*{2pt}
    \hrule
    \vspace*{2pt}
    \textrm{Inputs:}
 \begin{tabular}[t]{ll}
    $\theta \in \mathbb{R}^d$ & initial state
         \\[2pt] 
    $\mu:\mathbb{R}^{d} \rightarrow (0, \infty)$ & unnormalized target density 
     \\[2pt] 
    $M\in\R^{d\times d}$ & mass matrix
     \\[2pt] 
    $h > 0$ & macro step size
     \\[2pt] 
    $m_{\max} \in \mathbb{N}$ & maximum number of doublings  \\[2pt] 
$\delta >0$ &  maximum energy error   
    \end{tabular} 
    \\[2pt]
    Return: 
    \begin{tabular}[t]{ll}
    $\tilde{\theta} \in \mathbb{R}^d$ & next state
    \end{tabular}
    \vspace*{4pt}
    \hrule
    \vspace*{8pt}

    $\rho \sim \mathcal{N}(0, M)$  \\[4pt] 

   $(\tilde \theta, \tilde \rho) = (\theta_0,\rho_0) = (\theta, \rho)$ \\[4pt]

$w_0 = \mu(\theta_0) \; e^{-\frac{1}{2} \rho_0\tran M^{-1} \rho_0} $ \\[4pt] 

    $\calO=\left( (\theta_0,\rho_0) \right)$ \\[4pt]

    $\calW=\left( w_0 \right)$ \\[4pt]

    $B \sim \Unif(\{0,1\}^{m_{\max}})$ \\[4pt] 

    for $i$ from $1$ to $m_{\max}$:
   \vspace*{4pt}

    \qquad $\calO^{\old} = \calO =\left(  (\theta_a, \rho_a), \dots , (\theta_b, \rho_b) \right) $ \\[4pt] 

    \qquad $\calW^{\old} = \calW =\left( w_a, \dots ,  w_b \right) $ \\[4pt]

    \qquad if $B_i=1$: \\[4pt] 
    
    \qquad \qquad $\calO^{\ext},\calW^{\ext}=\texttt{extend-orbit-forward}(\theta_b, \rho_b, w_b,  \mu, M, h, \delta, 2^{i-1} ) $ \\[4pt] 

    \qquad \qquad $\calO =\calO^{\old} \odot \calO^{\ext} $ \\[4pt] 
    
    \qquad \qquad $\calW =\calW^{\old} \odot \calW^{\ext} $ \\[4pt]

    \qquad else:\hfill \\[4pt] 
    
    \qquad \qquad $\calO^{\ext},\calW^{\ext}=\texttt{extend-orbit-backward}(\theta_a, \rho_a, w_a, \mu, M, h, \delta, 2^{i-1} ) $ \\[4pt] 

   \qquad \qquad $\calO =\calO^{\ext} \odot  \calO^{\old}  $ \\[4pt] 
   
   \qquad \qquad $\calW =\calW^{\ext} \odot  \calW^{\old}  $ \\[4pt] 
   
    \qquad if $\texttt{sub-U-turn}(\calO^{\ext}, M)$: \\[4pt] 
    
    \qquad \qquad break \\[4pt]

     \qquad  $u\sim  \Unif((0,1))$ \\[4pt] 

    \qquad if $u \leq \dfrac{\sum\calW^{\ext}}
    {\sum \calW^{\old}} :$ \\[4pt]

    \qquad\qquad $(\tilde \theta, \tilde \rho) \sim \cat\left(  \calO^{\ext}, \calW^{\ext} \right)$ \\[4pt] 
    
    \qquad if $\texttt{U-turn}(\calO, M)$: \\[4pt] 
    
    \qquad \qquad break \\[4pt] 

    return $\tilde \theta$ \\[4pt] 
    \vspace*{6pt}
    \hrule
    \caption{\it WALNUTS algorithm with biased progressive state selection.}
    \label{algo:WALNUTS}
    \end{flushleft}
\end{mylisting}

\begin{mylisting}[t]
    \begin{flushleft}
    $\texttt{\bfseries U-turn}(\calO, M)$
    \vspace*{2pt}
    \hrule
    \vspace*{2pt}
    \textrm{Inputs:}
    \begin{tabular}[t]{ll}
   $\calO= \left((\theta^{\leftmost}, \rho^{\leftmost}), \dots, (\theta^{\rightmost}, \rho^{\rightmost})\right)$  
     & orbit, with $|\calO|$ a power of 2
    \end{tabular} 
    \vspace*{4pt}
    \hrule
    \vspace*{8pt}
     return $ (\rho^{\rightmost})^{\top} M^{-1} (\theta^{\rightmost} - \theta^{\leftmost}) <0~~\text{or}~~(\rho^{\leftmost})^{\top} M^{-1} (\theta^{\rightmost}  - \theta^{\leftmost}) <0$ \hfill 
    \vspace*{4pt}
    \hrule
    \caption{\it Check if the orbit $\calO$ satisfies the U-turn condition in Definition~\ref{defn:u-turn}.}
    \label{algo:U-turn}
    \end{flushleft}
\end{mylisting}
    
\begin{mylisting}[t]
    \begin{flushleft}
    $\texttt{\bfseries sub-U-turn}(\calO, M)$
    \vspace*{2pt}
    \hrule
    \vspace*{2pt}
    \textrm{Inputs:}
    \begin{tabular}[t]{ll}
   $\calO = \calO^{\leftmost} \odot \calO^{\rightmost}$    
    & orbit, with $|\calO^\leftmost| = |\calO^\rightmost|$ a power of 2
    \end{tabular} 
    \vspace*{4pt}
    \hrule
    \vspace*{8pt}
    if $\text{length}(\calO) < 2$: \ return \texttt{False}   \\[4pt]
    return $ \texttt{U-turn}(\calO, M) ~~\text{or}~~ \texttt{sub-U-turn}(\calO^{\text{left}},M)~~ \text{or}~~ \texttt{sub-U-turn}(\calO^{\text{right}},M)$ \hfill 
    \vspace*{4pt}
    \hrule
    \caption{\it Check if the orbit $\calO$ satisfies the sub-U-turn condition. The function \texttt{U-turn} is given in Listing~\ref{algo:U-turn}.}
    \label{algo:sub-U-turn}
    \end{flushleft}
\end{mylisting}

\begin{mylisting}[t]
    \begin{flushleft}
$\texttt{\bfseries extend-orbit-forward}\!\left(\theta_b, \rho_b, w_b,  \mu,  M, h, \delta, L \right)$
\vspace*{2pt}
    \hrule
    \vspace*{2pt}
    \textrm{Inputs:}
\begin{tabular}[t]{ll}
$(\theta_b, \rho_b) \in \mathbb{R}^{2d}$ & initial position, momentum \\[2pt] 
$w_b \in \mathbb{R}_{\ge 0}$ & initial weight \\[2pt] 
    $\mu:\mathbb{R}^{d} \rightarrow (0, \infty)$ & unnormalized target density 
     \\[2pt] 
    $M\in\R^{d\times d}$ & mass matrix
     \\[2pt] 
    $h > 0$ & macro step size
     \\[2pt] 
$\delta >0$ &  maximum energy error  \\[2pt] 
$L \in \mathbb{N}$ &  number of macro steps \\[2pt] 
\end{tabular}
   \\[2pt]
    Return: 
    \begin{tabular}[t]{ll}
      $\calO^{\ext} \in \left( \mathbb{R}^{2d} \right)^{L}$ & orbit \\[2pt]
       $\calW^{\ext} \in  (\mathbb{R}_{\ge 0})^{L}$ & weights
\\[2pt]
    \end{tabular}
    \vspace*{4pt}
    \hrule
    \vspace*{8pt}

    for $i$ from $b+1$ to $b+L$: \\[4pt]

     \qquad $\ell \sim p_{\micro}( \cdot \mid \micro(\theta_{i-1}, \rho_{i-1}, \mu,  M, h, \delta) ) $ \\[4pt]

     \qquad $(\theta_{i}, \rho_{i}) = \Phi_{h \ell^{-1}}^{\ell}(\theta_{i-1}, \rho_{i-1} )$ \\[4pt]
     
     \qquad $w_i = \dfrac{\mu(\theta_{i}) e^{-\frac{1}{2} (\rho_i)\tran M^{-1} \rho_i}}{\mu(\theta_{i-1}) e^{-\frac{1}{2} (\rho_{i-1})\tran M^{-1} \rho_{i-1}}}  \dfrac{p_{\micro}( \ell  \mid \micro(\theta_{i}, -\rho_{i}, \mu, M, h, \delta) )}{p_{\micro}( \ell \mid \micro(\theta_{i-1}, \rho_{i-1}, \mu, M, h, \delta) )}  w_{i-1}$ \\[4pt]

\null  return $\calO^{\ext} = \left( (\theta_{b+1}, \rho_{b+1}), (\theta_{b+2}, \rho_{b+2}), \dots, (\theta_{b+L}, \rho_{b+L}) \right)$ and $\calW^{\ext} = (w_{b+1}, w_{b+2}, \dots, w_{b+L})$  \hfill 
\vspace*{6pt}
    \hrule
    \caption{\it Generates an orbit of $L$ macro steps and associated weights over the index range $(b+1){:}(b+L)$. }
    \label{algo:extend-forward}
    \end{flushleft}
\end{mylisting}

\begin{mylisting}[ht]
    \begin{flushleft}
$\texttt{\bfseries extend-orbit-backward}\left(\theta_a, \rho_a, w_a,  \mu, M, h, \delta, L \right)$
\vspace*{2pt}
\hrule
\vspace*{2pt}
\textrm{Inputs:}
\begin{tabular}[t]{ll}
$(\theta_a, \rho_a) \in \mathbb{R}^{2d}$ & initial position, momentum \\[2pt] 
$w_a \in \mathbb{R}_{\ge 0}$ & initial weight \\[2pt] 
    $\mu:\mathbb{R}^{d} \rightarrow (0, \infty)$ & unnormalized target density 
    \\[2pt] 
    $M\in\R^{d\times d}$ & mass matrix
     \\[2pt] 
    $h > 0$ & macro step size
     \\[2pt] 
$\delta >0$ &  maximum energy error  \\[2pt] 
$L \in \mathbb{N}$ &  number of macro steps  \\[2pt] 
\end{tabular}
   \\[2pt]
    Return: 
    \begin{tabular}[t]{ll}
      $\calO^{\ext} \in \left( \mathbb{R}^{2d} \right)^{L}$ & orbit \\
       $\calW^{\ext} \in  \mathbb{R}^{L}$ & weights

    \end{tabular}
\vspace*{4pt}
\hrule
\vspace*{8pt}
for $i$ from $a-1$ down to $a-L$: \\[4pt] 

     \qquad $\ell \sim p_{\micro}( \cdot \mid \micro(\theta_{i+1}, -\rho_{i+1}, \mu, M, h, \delta) ) $ \\[4pt] 

     \qquad $(\theta_{i}, \rho_{i}) = \mathcal{F} \circ \Phi_{h \ell^{-1}}^{\ell}\circ \mathcal{F} (\theta_{i+1}, \rho_{i+1} )$ \\[4pt] 
     
     \qquad $w_i =\dfrac{\mu(\theta_{i}) e^{-\frac{1}{2} (\rho_i)\tran M^{-1} \rho_i}}{\mu(\theta_{i+1}) e^{-\frac{1}{2} (\rho_{i+1})\tran M^{-1} \rho_{i+1}}}  \dfrac{p_{\micro}( \ell  \mid \micro(\theta_{i}, \rho_{i}, \mu, M, h, \delta) )}{p_{\micro}( \ell  \mid \micro(\theta_{i+1}, -\rho_{i+1}, \mu, M, h, \delta) )}  w_{i+1}$ \\[4pt]

\null  return $\calO^{\ext} = \left( (\theta_{a-L}, \rho_{a-L}) , \dots, (\theta_{a-2}, \rho_{a-2}), (\theta_{a-1}, \rho_{a-1}) \right)$ and $\calW^{\ext} = (w_{a-L},  \dots, w_{a-2}, w_{a-1})$  \hfill 
\vspace*{4pt}
\hrule
\caption{\it Generates an orbit of $L$ macro steps and associated weights over the index range $(a-L){:}(a-1)$. }
\label{algo:extend-backward}
\end{flushleft}
\end{mylisting}

\begin{mylisting}[t]
    \begin{flushleft}
    $\texttt{\bfseries micro}(\theta, \rho,  \mu, M, h,  \delta)$
    \vspace*{2pt}
    \hrule
    \vspace*{2pt}
    \textrm{Inputs:}
\begin{tabular}[t]{ll}
$(\theta, \rho) \in \mathbb{R}^{2d}$ & initial position, momentum \\[2pt] 
$h_0 > 0$ &  macro step size  \\[2pt] 
     $\mu:\mathbb{R}^{d} \rightarrow (0, \infty)$ & unnormalized target density 
    \\[2pt] 
    $M\in\R^{d\times d}$ & mass matrix
     \\[2pt] 
    $h > 0$ & macro step size
     \\[2pt] 
$\delta >0$ &  maximum energy error  \\[2pt] 
\end{tabular}
   \\[2pt]
    Return: 
    \begin{tabular}[t]{ll}
      $\ell \in \mathbb{N}$ & step size reduction factor 
    \end{tabular}
    \vspace*{4pt}
    \hrule
    \vspace*{8pt}
    $(\theta^{(0)}, \rho^{(0)}) = (\theta, \rho)$

    for $i$ from $0$ to $\infty$   \\[4pt]

    \qquad $\ell=2^{i}$ \\[4pt]

    \qquad $h_{\micro}=h \; \ell^{-1}$ \\[4pt]

    \qquad     $H_{\max} = H_{\min} = -\log \mu(\theta^{(0)}) + \frac12 (\rho^{(0)})\tran M^{-1} \rho^{(0)}$ \\[4pt]

    \qquad for $j$ from $0$ to $\ell-1$  \\[4pt]

\qquad \qquad $\pos{\rho}{j + 1/2} = \pos{\rho}{j} + \frac{1}{2} \; h_{\micro} \; \nabla \log \mu(\pos{\theta}{j})$ 
\\[4pt]

\qquad \qquad $\pos{\theta}{j + 1} = \pos{\theta}{j} + h_{\micro} \; M^{-1} \pos{\rho}{j + 1/2}$ 
\\[4pt]

\qquad \qquad $\pos{\rho}{j + 1} = \pos{\rho}{j + 1/2} + \frac{1}{2} \; h_{\micro} \; \nabla \log \mu(\pos{\theta}{j + 1})$ 
\\[4pt]

    \qquad \qquad  $H^{(j+1)} = - \log \mu(\theta^{(j+1)}) + \frac12 (\rho^{(j+1)})\tran M^{-1} \rho^{(j+1)}$ \\[4pt]

    \qquad \qquad $H_{\max} = \max( H^{(j+1)}, H_{\max} )$ \\[4pt]
    
    \qquad \qquad $H_{\min} = \min( H^{(j+1)}, H_{\min} )$  \\[4pt]

    \qquad if $H_{\max} - H_{\min} \le \delta$  \\[4pt]
    
    \qquad \qquad return $\ell$ 
    
    \vspace*{4pt}
    \hrule
    \caption{\it Computes the smallest power-of-two integer $\ell$ such that $\ell$ leapfrog steps of micro step size $h_{\micro} = h \ell^{-1}$ keep the energy error within the user-specified threshold $\delta$, as defined in \eqref{eq:micro}.}
    \label{algo:micro_lf}
    \end{flushleft}
\end{mylisting}

\clearpage

\section{Practical Implementation}

\label{sec:practical}

This section outlines practical implementation strategies that reduce the computational and memory overhead of WALNUTS.   The pseudocode in Listing~\ref{algo:WALNUTS}, including dependencies, is intended to provide the reader with a clear conceptual understanding of the WALNUTS transition kernel. 
However, directly implementing this version can lead to substantial redundant computations and high memory usage.  In particular, storing all positions and momenta in an orbit $\mathcal{O}$ may require $O(d \cdot 2^{m_{\max}})$ memory, where $d$ is the dimension of the state space and $m_{\max}$ is the maximum number of doublings.

\subsection{Error-controlled leapfrog integration steps}

In the pseudocode for orbit expansion (Listings~\ref{algo:extend-forward} and~\ref{algo:extend-backward}), the calls to \texttt{micro} and the leapfrog integration steps used to build the final orbit are implemented as separate function calls for both forward and backward directions. While this modular structure aids conceptual understanding and code readability, it introduces computational overhead, as the same leapfrog steps may be recomputed multiple times.

In practice, the contents of the \texttt{for} loops in \texttt{extend-orbit-forward} and \texttt{extend-orbit-backward} are consolidated into a single routine that performs all relevant operations together. This combined implementation reduces redundant computation by reusing results where possible. The forward case is described below; the backward case is analogous.

\begin{itemize}
\item Compute $\widetilde{\ell}=\texttt{micro}(\theta_{i-1},\rho_{i-1},\mu,M,h,\delta)$, and store $(\theta_i, \rho_i) = \Phi_{h \widetilde{\ell}^{-1}}^{\widetilde{\ell}}(\theta_{i-1}, \rho_{i-1})$, which is already computed as part of the \texttt{micro} procedure.
\item Sample $\ell\sim p_{\micro}(\cdot|\widetilde{\ell})$.
\item If $\ell = \widetilde{\ell} = 0$: reuse the stored result $(\theta_i, \rho_i)$.
\item If $\ell = \widetilde{\ell} > 0$: use the precomputed result $(\theta_i, \rho_i) = \Phi_{h \widetilde{\ell}^{-1}}^{\widetilde{\ell}}(\theta_{i-1}, \rho_{i-1})$. To compute the weight factor, we then run \texttt{micro}$(\theta_i, -\rho_i, \mu, M, h, \delta)$ using step sizes corresponding to micro step count $2^j$ for $j = 0, 1, \dots, \log_2(\ell) - 1$. (There is no need to run the backward \texttt{micro} iteration for $j = \log_2(\ell)$ or higher due to reversibility of the integration; see Lemma~\ref{lemma:forward_backward}.)
\item Otherwise (i.e., if $\ell \ne \widetilde{\ell}$): proceed as in \texttt{extend-orbit-forward}. In this case, there are no obvious computational savings compared to the original algorithm, and the backward \texttt{micro} loop generally cannot be terminated early.
\end{itemize}

Note that there can be a substantial computational cost difference between the cases $\ell = \widetilde{\ell}$ and $\ell \ne \widetilde{\ell}$. This may partly explain why WALNUTS configured with $p_{\texttt{micro}}(\widetilde{\ell} \mid \widetilde{\ell}) = 1$ often appears to run faster in practice, even if it terminates orbit construction slightly earlier and more frequently.

\subsection{Memory-efficient orbit construction}

\begin{figure}
    \centering
    \includegraphics[width=0.99\linewidth]{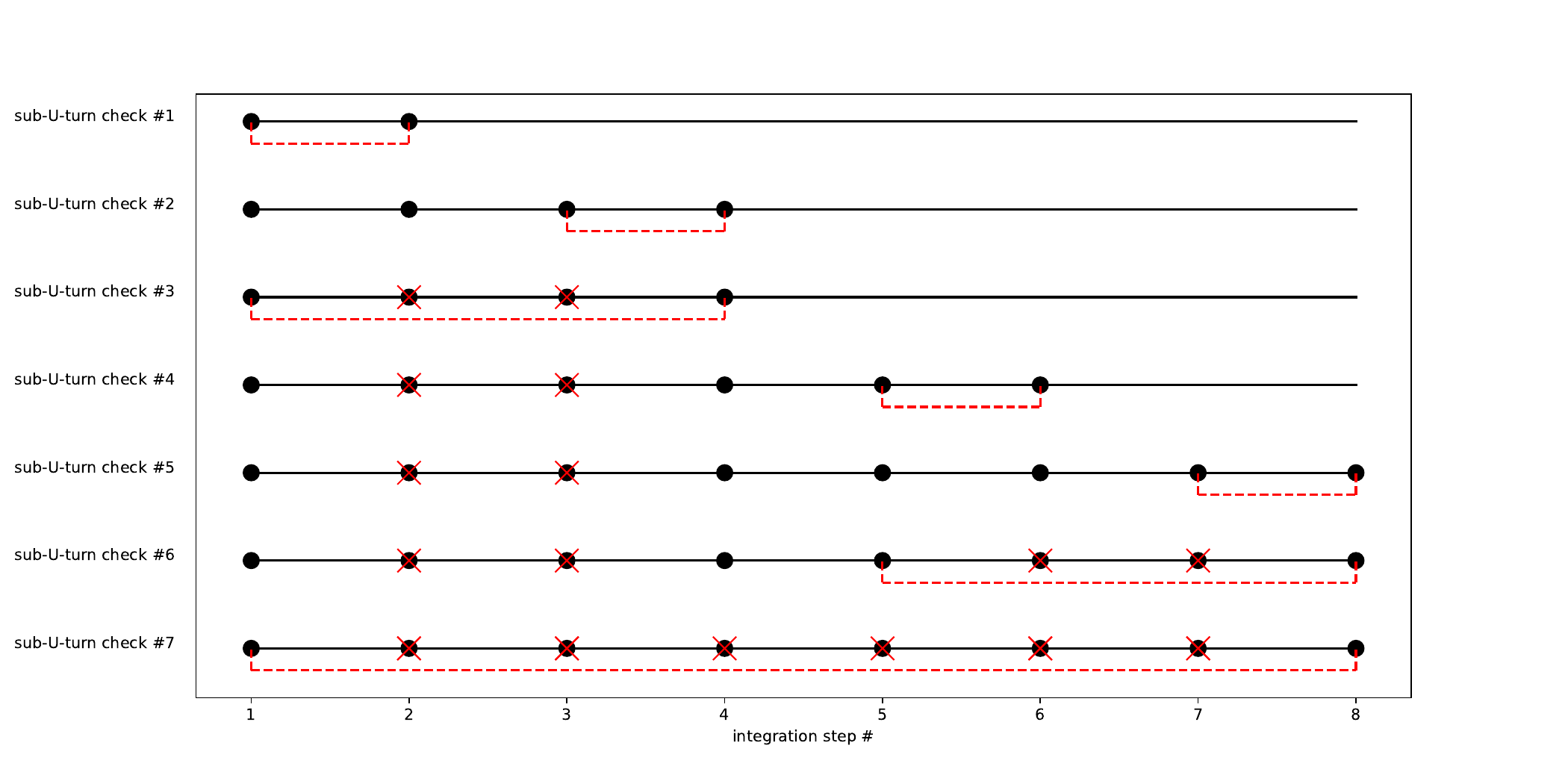}
    \caption{Sequence of sub-U-turn checks during forward orbit expansion with 8 integration steps. Black dots indicate computed integration steps; red Xs mark steps that have been discarded. Red underbraces denote the pairs of states involved in each U-turn check. A state can be safely deleted once it is covered by a completed U-turn check that spans beyond it on both sides.
}
    \label{fig:U-turn-seq}
\end{figure}

As also described in \cite{HoGe2014}, a more memory-efficient implementation, requiring only $O(d \cdot m_{\max})$ memory, is possible by checking the sub-U-turn condition and performing sampling from $\cat\left( \calO^{\ext}, \calW^{\ext} \right)$ concurrently during the execution of \texttt{extend-orbit-forward} or \texttt{extend-orbit-backward}. Since this approach is thoroughly explained in \cite{HoGe2014}, we provide only a brief summary here.

The categorical sampling step is conceptually straightforward. Suppose we wish to sample a random index $I \in \{1, \dots, N\}$ according to a sequence of (possibly unnormalized) weights $w_1, \dots, w_N$. This can be done in an online manner using the following recursive procedure:
\begin{itemize}
    \item Initialize $I = 1$ and $S_1 = w_1$.
    \item For $i = 2, \dots, N$:
    \begin{itemize}
        \item Set $S_i = S_{i-1} + w_i$,
        \item Draw $U_i \sim \mathrm{Unif}(0,1)$,
        \item With probability $w_i / S_i$, set $I = i$; otherwise, leave $I$ unchanged.
    \end{itemize}
\end{itemize}
At the end of this procedure, the variable $I$ is distributed according to the categorical distribution with probabilities proportional to $w_1, \dots, w_N$.

Sub-U-turn checks can also be performed concurrently with the  \texttt{extend-orbit-forward} and \texttt{extend-orbit-backward} routines. In fact, the algorithm in Listing~\ref{algo:sub-U-turn} is designed to operate sequentially on the states of the orbit as they are computed. Once a U-turn check has been performed between two states, say $i$ and $j$, the intermediate states $i+1, \dots, j-1$ are no longer needed and can be discarded to reduce memory usage.

Figure~\ref{fig:U-turn-seq} illustrates this concurrent procedure for \texttt{extend-orbit-forward} with 8 integration steps. New integration steps are computed in iterations 1, 2, 4, and 5 of the algorithm. At no point is the full sub-orbit stored in memory.

Concurrent U-turn checks can also lead to computational savings: if a U-turn is detected early (e.g., during check \#3 in Figure~\ref{fig:U-turn-seq}), then the remaining integration steps can be skipped, as no proposal will be drawn from that region and the orbit expansion stops.

Similar considerations apply when the weight factor from the reversibility check becomes zero. Specifically, if the reversibility check yields zero weights on both sides of the current state $(\theta_0, \rho_0)$, the orbit-building process can be terminated early. However, if only one side yields a zero weight, orbit construction must continue in the other direction.

\section{Exploring Tuning Strategies}\label{sec:tuning}

 This section focuses on tuning the macro step size $h$, which determines the temporal spacing between candidate states along the final orbit, and the energy error threshold $\delta$, which controls local step size refinement within each macro step. WALNUTS is designed to be robust to user input, but performance can be significantly improved by selecting $h$ and $\delta$ carefully in a warmup phase; that is, an initial portion of the Markov chain devoted to adapting algorithm parameters prior to collecting samples from the posterior. In this part, we explore approaches to tuning both parameters with the goal of balancing efficiency and robustness.

\subsection{Energy error threshold \texorpdfstring{$\delta$}{delta}}

To appropriately set the energy error threshold $\delta$ in the micro step selection function
\begin{equation} \label{eq:micro_app}
\micro(\theta,\rho, \mu, M, h, \delta) := \min\left\{ \ell \in 2^{\mathbb{N}} \mid H^+_\ell - H^-_\ell \leq \delta \right\},
\end{equation}
we aim to control the local energy error introduced by each integration step while ensuring robust global behavior of the sampler.

When choosing the energy error threshold $\delta > 0$, the primary objective is to prevent numerical divergences by bounding the energy error introduced in each macro integration step. An additional, equally important objective is to control the total change in energy along an orbit $\calO$. Doing so helps ensure that the distribution of the integration time index $i$, which determines the next step of the chain, is broadly spread across the orbit.

To see this, consider the idealized setting in which the leapfrog integrator is replaced by the exact Hamiltonian flow, and assume for simplicity that the orbit length is fixed. In this case, there is no energy error, and the normalized integration time index $i$ follows an exact symmetric triangular distribution centered at $0.5$ (see Theorem~\ref{thm:law_of_iprime}). When the orbit length is tuned to match the local scale of the target, this integration time distribution enables effective exploration along that scale.

In practice, the leapfrog integrator induces an energy error, which can distort this ideal integration time distribution and lead to biased or inefficient proposals --- particularly if the energy error varies substantially along the orbit. Controlling the energy error is therefore essential: it allows the sampler to approximate the idealized behavior locally, so that movement occurs effectively along the relevant scale of the target.

While it is possible to combine WALNUTS' local step size adaptation with global orbit-level adaptation, as proposed in \cite{BouRabeeCarpenterKleppeMarsden2024}, doing so may compromise the fine-grained adaptivity that WALNUTS is designed to achieve; namely, reducing the step size only in regions where it is truly necessary. Instead, we advocate a strategy that tunes the local energy error threshold $\delta$ so that the \emph{global} energy change along the orbit satisfies
\begin{equation}
H_{\calO}^+ - H_{\calO}^- < \Delta,
\label{eq:global-energy-error-thresh}
\end{equation}
with high probability, where $H_{\calO}^+$ and $H_{\calO}^-$ denote the maximum and minimum Hamiltonians observed along the orbit, and $\Delta > 0$ is a global energy error threshold. As shown in \cite{BouRabeeCarpenterKleppeMarsden2024}, the quantity $\exp(-(H_{\calO}^+ - H_{\calO}^-))$ is strongly correlated with the probability $P(i \ne 0)$ of selecting a nontrivial proposal. Enforcing the bound \eqref{eq:global-energy-error-thresh} with a prescribed probability level $p_a$  (e.g., $p_a=0.95$) therefore promotes stable and reliable sampling.

WALNUTS employs a \emph{variable step size leapfrog integrator}, which adaptively adjusts the step size within each macro integration interval to satisfy a local energy error threshold $\delta$. Specifically, the step size is recursively halved until the estimated energy error in a single step satisfies
\[
|H^+_{\ell} - H^-_{\ell}| \leq \delta.
\]
If an orbit consists of $n$ macro steps, each satisfying the above bound, then in the worst case, where all local energy errors accumulate constructively, then the total energy error satisfies
\[
H_{\calO}^+ - H_{\calO}^- \leq n \delta.
\]
While this is a conservative upper bound, it highlights the need for orbit-level control. In practice, the signs of energy errors often fluctuate, especially in well-conditioned regions of the target, leading to partial cancellation. Nonetheless, the \emph{envelope} of the energy error (i.e., the difference between the largest and smallest Hamiltonians along the orbit) may grow with orbit length. To capture this possibility, we model the accumulated energy variation as
\[
H_{\calO}^+ - H_{\calO}^- \approx \mathcal{K} \delta,
\]
where $\mathcal{K}$ is a \emph{local-to-orbit inflation factor}, a random variable whose distribution depends on the number of integration steps, the macro step size $h$, and the local geometry of the target distribution. In the ideal case of uncorrelated stepwise errors, $\mathcal{K}$ remains $\mathcal{O}(1)$ even for long orbits, due to frequent sign changes.

To ensure that \eqref{eq:global-energy-error-thresh} holds with probability at least $p_a$, we propose adapting $\delta$ during warmup by recording the observed inflation factor
\begin{equation}
\mathcal{K} = \frac{H_{\calO}^+ - H_{\calO}^-}{\delta}
\label{eq:inflation-fac-obs}
\end{equation}
for each constructed orbit. These values are collected across warmup iterations, and $\delta$ is then updated according to
\begin{equation}
\delta = \frac{\Delta}{q^\mathcal{K}_{p_a}},
\end{equation}
where $q^\mathcal{K}_{p_a}$ denotes the empirical $p_a$-quantile of the recorded inflation factors. Since each orbit contributes a single value of $\mathcal{K}$, this adaptation procedure is both memory-efficient and computationally inexpensive.

This strategy enables automatic and interpretable tuning of $\delta$, aligning local step size control with global energy error constraints, and thereby supporting the stability and robustness of the WALNUTS sampler.

\subsection{Macro step size \texorpdfstring{$h$}{h} (given \texorpdfstring{$\delta$}{delta})}

Recall that the function $\micro(\theta, \rho, \mu, M, h, \delta)$ returns the number of step size halvings needed to reduce the energy error below the threshold $\delta$, as defined in \eqref{eq:micro_app}.   In particular,  if no micro step size halving is needed, then $\micro(\theta, \rho, \mu, M, h, \delta) = 1$.  

To tune the macro step size $h$, the primary goal is efficiency. If $h$ is too small, WALNUTS will take unnecessarily many leapfrog steps, even in regions where larger steps would be sufficient. On the other hand, if $h$ is too large, the \texttt{micro} routine will frequently trigger step size halving to satisfy the energy error threshold $\delta$, which wastes computation. Moreover, overly large values of $h$ can cause the integration paths to become too coarse to detect U-turns or other important geometric features, potentially leading to poor proposal quality.

To navigate this trade-off, we propose choosing $h$ such that the probability of needing no step size halving is equal to a target value $\Gamma \in (0,1)$:
\begin{equation}
P\left(\texttt{micro}(\theta,\rho,\mu,M,h_0,\delta)=1 \right)= \Gamma \label{eq:h-adapt-target}
\end{equation}
This criterion means that, in a typical region of phase space, the initial macro step size $h$ is accepted without refinement in a fraction $\Gamma$ of cases. A value such as $\Gamma = 0.8$ typically works well in practice: most steps proceed without modification, while \texttt{micro} still adjusts the step size in regions where smaller steps are needed.

In short, this objective helps ensure that $h$ is large enough to be efficient across most of the target distribution, without compromising robustness or geometric resolution in more challenging areas.

\end{document}